%% file: main.tex
\documentclass[envcountsame,a4paper,11pt]{article}

\usepackage{diego-macro}
\usepackage{
  a4wide,
   tikz,
  amssymb,
  amsmath,
  xspace,
  enumerate,
  stmaryrd, 
xcolor,
colortbl
}
\usepackage{stmaryrd} 

\allowdisplaybreaks 
\usepackage[colorlinks=true]{hyperref}
\usepackage {color}
\usepackage{pdfsync}

\usepackage{lipsum}
\usepackage{enumitem,bigdelim,multirow}

\newif\iffull
\fullfalse

\usepackage{amsmath,amssymb,amsthm, xspace,a4wide,color}
\usepackage{xstring}

\usepackage{pdfsync}
\hypersetup{
  linkcolor  = blue,
  citecolor  = blue,
  urlcolor   = blue,
  colorlinks = true,
}

\input{macros-main}

\title{Axiomatizations for downward XPath on Data Trees}
\author{Sergio Abriola$^{1}$ \and Mar\'ia Emilia Descotte$^{2}$ \and Raul Fervari$^{3}$ \and Santiago Figueira$^{1}$}

\date{
\begin{small}
$^1$ Universidad de Buenos Aires, and ICC-CONICET, Argentina \\
$^2$ Universidad de Buenos Aires, and IMAS-CONICET, Argentina \\
$^3$ Universidad Nacional de C\'ordoba, and CONICET, Argentina\\
\end{small}}

\begin{document}
\maketitle

\begin{abstract}
We give sound and complete axiomatizations for XPath with data tests by `equality' or `inequality', and containing the single `child' axis. This data-aware logic predicts over data trees,
which are tree-like structures whose every node contains a label from a finite alphabet
and a data value from an infinite domain. The language allows us to compare data values of two nodes but cannot access the data values themselves (i.e.\ there is no comparison by constants).
Our axioms are in the style of equational logic, extending the axiomatization of data-oblivious XPath, by B.\ ten Cate, T.\ Litak and M. Marx. We axiomatize the full logic with tests by `equality' and `inequality', and also a simpler fragment with `equality' tests only. Our axiomatizations apply both to node expressions and path expressions.
The proof of completeness relies on a novel normal form theorem for XPath with data tests.
\end{abstract}


\section{Introduction} \label{sec:intro}
\input{intro}

\subsection{Organization}
In \S\ref{sec:prelim} we give the formal syntax and semantics of \xpathd with `child' axis, called $\xpd$. As we already mentioned, we also study a special syntactical fragment, called $\xpdeq$, whose all data-aware diamonds are of the form $\tup{\alpha=\beta}$, and keeps out those of the form $\tup{\alpha\neq\beta}$. In \S\ref{sec:compl-restr} we give a sound and complete axiomatic system for $\xpdeq$: in \S\ref{subsec:AxiomsDown} we state the needed axiom schemes, an extension of those introduced in \cite{cateLM10}; in \S\ref{subsec:nf} we define the normal forms for $\xpdeq$ (these are not an extension of those defined in \cite{cateLM10}) and state the corresponding normal form theorem; in \S\ref{subsec:completeness} we show the completeness result, whose more complex part lies in proving that any node expression in normal form is satisfiable in a canonical model. In \S\ref{sec:compl} we extend the previous
axiom schemes to get a sound and complete axiomatic system for $\xpd$. We follow the same route as for $\xpdeq$: axiom schemes (\S\ref{subsec:Axiomsneq}), normal form (\S\ref{sec:nf-neq}) and canonical model (\S\ref{completeness neq}). Those proofs requiring highly technical arguments were deferred to Appendix \ref{app}. Finally, in \S\ref{sec:conclusions} we close with some final remarks and future lines of research.

\section{Preliminaries} \label{sec:prelim}
\input{prelim}

\section{Axiomatic System for $\xpdeq$}\label{sec:compl-restr}

\subsection{Axiomatization} \label{subsec:AxiomsDown}
\input{axiom}

\subsection{Normal forms} \label{subsec:nf}
\input{nf}
\subsection{Completeness for node and path expressions}\label{subsec:completeness}
\input{compl-downward}

\section{Axiomatic System for $\xpd$}\label{sec:compl}

\subsection{Axiomatization}  \label{subsec:Axiomsneq}
\input{axiom-neq}
\subsection{Normal forms}\label{sec:nf-neq}
\input{nf-neq}

\subsection{Completeness for node and path expressions}\label{completeness neq}
\input{compl-neq}

%

\section{Conclusions}\label{sec:conclusions}
\input{conclu}

\bigskip

\paragraph{Acknowledgements.}
This work was partially supported by grant ANPCyT-PICT-2013-2011 and UBACyT 20020150100002BA.

\bibliographystyle{plain}
\bibliography{bib}


\appendix
\section{Missing proofs}\label{app}
\input{appendix-neq}

\end{document}

%% file: macros-main.tex
\theoremstyle{plain}
\newtheorem{theorem}{Theorem}
\newtheorem{proposition}[theorem]{Proposition}
\newtheorem{lemma}[theorem]{Lemma}
\newtheorem{corollary}[theorem]{Corollary}

\theoremstyle{definition}
\newtheorem{definition}[theorem]{Definition}
\newtheorem{observation}[theorem]{Observation}

\theoremstyle{remark}
\newtheorem{fact}[theorem]{Fact}

\newtheorem{example}[theorem]{Example}
\newtheorem{remark}[theorem]{Remark}
\newtheorem{sketch}[theorem]{Sketch}

\usepackage[textwidth=0.89in,textsize=scriptsize]{todonotes}

\newcommand{\nfP}[1]{P^-_{#1}}
\newcommand{\nfD}[1]{D^-_{#1}}
\newcommand{\nfN}[1]{N^-_{#1}}

\newcommand{\lbl}{\mathit{label}}

\newcommand{\dbracket}[1]{[\![ #1 ]\!]}

\newcommand{\midd}{\mathrel{\,\mid\,}}
\newcommand{\eqdef}                     
  {\stackrel{\scriptscriptstyle \mathrm{def}}{=}}

\newcommand{\md}{{\rm dd}}

\newcommand{\child}{{\to}} 
\newcommand{\childp}[1]{{\stackrel{#1}{\to}}} 

\newcommand{\xpath}{\text{XPath}\xspace}
\newcommand{\corexpath}{\text{Core-XPath}\xspace}
\newcommand{\coredataxpath}{\text{Core-Data-XPath}\xspace}

\newcommand{\xp}[1]{\xpathd({#1})}

\newcommand{\xpathd}{\ensuremath{\xpath_{=}}\xspace}
\newcommand{\down}{\downarrow}

\newcommand{\xpd}{\xp{\dow}\xspace}

\newcommand{\xpdeq}{\xpd^-}

\newcommand{\dow}{{\downarrow}}

\newcommand{\rtdow}{{\downarrow_*}}

\newcommand{\OMIT}[1]{}

\newcommand{\con}{\textbf{Con}\xspace}

\newcommand{\bu}{{\bf u}}
\newcommand{\bv}{{\bf v}}

\newcommand{\equivInstance}[2]{#1 \equiv #2} 
\newcommand{\leqInstance}[2]{#1 \leq #2} 

\newcommand{\botNode}{\mbox{\sc false}} 
\newcommand{\botPath}{\bot} 
\newcommand{\topNode}{\mbox{\sc true}} 

\newcommand{\sisi}{{\bf V_{=,\neq}}}
\newcommand{\sino}{{\bf V_{=,\lnot \neq}}}
\newcommand{\nosi}{{\bf V_{\lnot =,\neq}}}
\newcommand{\nono}{{\bf V_{\lnot =,\lnot \neq}}}

\newcommand{\prax}[1]{%
    \IfEqCase{#1}{%
        {1}{\hyperlink{praxone}{\sf PrAx1}\xspace}%
        {2}{\hyperlink{praxtwo}{\sf PrAx2}\xspace}%
        {3}{\hyperlink{praxthree}{\sf PrAx3}\xspace}%
    }%
}%

\newcommand{\ndax}[1]{%
    \IfEqCase{#1}{%
        {1}{\hyperlink{ndaxone}{\sf NdAx1}\xspace}%
        {2}{\hyperlink{ndaxtwo}{\sf NdAx2}\xspace}%
        {3}{\hyperlink{ndaxthree}{\sf NdAx4}\xspace}%
        {4}{\hyperlink{ndaxfour}{\sf NdAx3}\xspace}%
    }%
}%

\newcommand{\lbax}[1]{%
    \IfEqCase{#1}{%
        {1}{\hyperlink{lbaxone}{\sf LbAx1}\xspace}%
        {2}{\hyperlink{lbaxtwo}{\sf LbAx2}\xspace}%
    }%
}%

\newcommand{\eqax}[1]{%
    \IfEqCase{#1}{%
        {1}{\hyperlink{eqaxone}{\eqaxonen}\xspace}%
        {2}{\hyperlink{eqaxtwo}{\eqaxtwon}\xspace}%
        {3}{\hyperlink{eqaxthree}{\eqaxthreen}\xspace}%
        {4}{\hyperlink{eqaxfour}{\eqaxfourn}\xspace}%
        {5}{\hyperlink{eqaxfive}{\eqaxfiven}\xspace}%
        {6}{\hyperlink{eqaxsix}{\eqaxsixn}\xspace}%
        {7}{\hyperlink{eqaxseven}{\eqaxsevenn}\xspace}%
        {8}{\hyperlink{eqaxeight}{\eqaxeightn}\xspace}%
    }%
}%

\newcommand{\eqaxonen}{{\sf EqAx6}}%
\newcommand{\eqaxtwon}{{\sf EqAx1}}%
\newcommand{\eqaxthreen}{{\sf EqAx2}}%
\newcommand{\eqaxfourn}{{\sf EqAx3}}%
\newcommand{\eqaxfiven}{{\sf EqAx4}}%
\newcommand{\eqaxsixn}{{\sf EqAx7}}%
\newcommand{\eqaxsevenn}{{\sf EqAx5}}%
\newcommand{\eqaxeightn}{{\sf EqAx8}}%

\newcommand{\isaxone}{\hyperlink{isaxone}{\sf IsAx1}\xspace}
\newcommand{\isaxtwo}{\hyperlink{isaxtwo}{\sf IsAx2}\xspace}

\newcommand{\isaxfour}{\hyperlink{isaxfour}{\sf IsAx4}\xspace}
\newcommand{\isaxfive}{\hyperlink{isaxfive}{\sf IsAx5}\xspace}
\newcommand{\isaxsix}{\hyperlink{isaxsix}{\sf IsAx6}\xspace}
\newcommand{\isaxseven}{\hyperlink{isaxseven}{\sf IsAx7}\xspace}

\newcommand{\neqaxone}{\hyperlink{neqaxone}{\neqaxonen}\xspace}
\newcommand{\neqaxtwo}{\hyperlink{neqaxtwo}{\neqaxtwon}\xspace}
\newcommand{\neqaxfour}{\hyperlink{neqaxfour}{\neqaxfourn}\xspace}
\newcommand{\neqaxfive}{\hyperlink{neqaxfive}{\neqaxfiven}\xspace}
\newcommand{\neqaxeight}{\hyperlink{neqaxeight}{\neqaxeightn}\xspace}
\newcommand{\neqaxnine}{\hyperlink{neqaxnine}{\neqaxninen}\xspace}
\newcommand{\neqaxten}{\hyperlink{neqaxten}{\neqaxtenn}\xspace}

\newcommand{\neqaxfifteen}{\hyperlink{neqaxfifteen}{\neqaxfifteenn}\xspace}
\newcommand{\neqaxsixteen}{\hyperlink{neqaxsixteen}{\neqaxsixteenn}\xspace}
\newcommand{\neqaxseventeen}{\hyperlink{neqaxseventeen}{\neqaxseventeenn}\xspace}
\newcommand{\neqaxeighteen}{\hyperlink{neqaxeighteen}{\neqaxeighteenn}\xspace}

\newcommand{\neqaxonen}{{\sf NeqAx1}\xspace}
\newcommand{\neqaxtwon}{{\sf NeqAx4}\xspace}
\newcommand{\neqaxfourn}{{\sf NeqAx6}\xspace}
\newcommand{\neqaxfiven}{{\sf NeqAx6}\xspace}
\newcommand{\neqaxeightn}{{\sf NeqAx7}\xspace}
\newcommand{\neqaxninen}{{\sf NeqAx8}\xspace}
\newcommand{\neqaxtenn}{{\sf NeqAx9}\xspace}
\newcommand{\neqaxfifteenn}{{\sf NeqAx5}\xspace}
\newcommand{\neqaxsixteenn}{{\sf NeqAx2}\xspace}
\newcommand{\neqaxseventeenn}{{\sf NeqAx3}\xspace}
\newcommand{\neqaxeighteenn}{{\sf NeqAx10}\xspace}

\newcommand{\axiomRestrNeq}{\mathsf{XP}}
\newcommand{\axiomRestr}{\mathsf{XP^-}}

\newcommand{\eqfull}{\equiv}
\newcommand{\eqres}{\equiv^-}

%% file: intro.tex

XML (eXtensible Markup Language) is the most successful language for data exchange on
the web. It meets the requirements of a flexible, generic and platform-independent language.
An XML document is a hierarchical structure that can be abstracted as a data tree,
where nodes have labels (such as {\em LastName}) from a finite domain, and data values (such as 
{\em Smith}) from an infinite domain. For some tasks, data values can be disregarded (for instance, 
checking whether a given XML document conforms to a schema specification). But many applications
require data-aware query languages, that is, languages with the ability of comparing data
values. Indeed, the possibility to perform {\em joins} in queries or comparing for equality of data
values is a very common and necessary feature in database query languages.

XPath is the most widely used query language for XML documents; it is an open standard and constitutes a World Wide Web
Consortium (W3C) Recommendation~\cite{xpath:w3c}. XPath has syntactic operators or `axes' to navigate the tree using the `child', `parent', `sibling', etc.\ accessibility relations, and can make tests on intermediate nodes.~\corexpath~\cite{GKP05} is the fragment of XPath~1.0 containing only the navigational behavior of XPath, i.e.\ without any reference to the data in the queries.

\corexpath can be seen as a {\em modal language}, such as those used in
software verification, like Linear Temporal Logic (LTL)~\cite{BK08} or Propositional Dynamic Logic (PDL)~\cite{HKT00}.
XPath has been already investigated from a `modal' point of view. In~\cite{CFL10} this perspective
is illustrated by showing how some results on \corexpath fragments can be derived from classical 
results in modal logic. 
In particular, when the only accessibility relation is `child', \corexpath has many similarities with basic modal logic (BML). 
First, in the absence of data, an XML document just becomes a tree whose every node has a label from a finite domain; this is a special kind of Kripke models, when labels are represented as propositional letters. Second, any property expressed in BML can be translated to \corexpath and vice versa. There are also some differences: \corexpath may express not only properties $\varphi$ on nodes, called {\em node expressions}, but also on {\em paths}, called {\em path expressions}. When $\alpha$ is a path expression, its truth is evaluated on pairs of nodes instead of on individual nodes, as BML does. In a nutshell, $\alpha$ is true at $(x,y)$ if the path from $x$ to $y$ (which is unique, since our models are trees) satisfies the condition expressed by $\alpha$. 

The formal syntax and semantics of \corexpath will be given later in full detail, but let us now give a glimpse of it. If $\alpha$ is a path expression, in \corexpath we can write $\tup{\alpha}$, a node expression saying that there is a descendant $y$ of $x$ such that $(x,y)$ satisfies $\alpha$ (see Figure~\ref{fig:diamonds}(a)).

\begin{figure}[ht]
   \begin{center}
   \begin{tabular}{c@{\hskip 1in}c@{\hskip 1in}c}
   \includegraphics[scale=0.25]{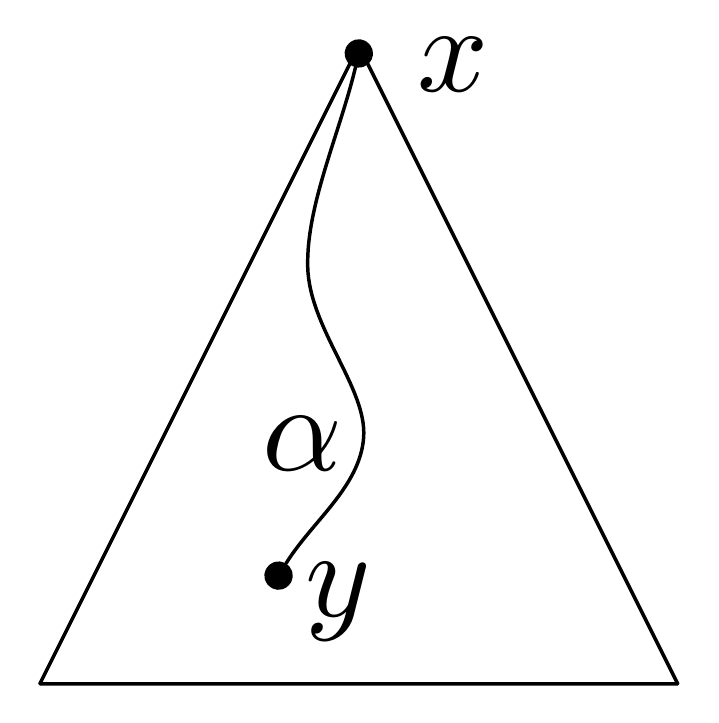}&\includegraphics[scale=0.25]{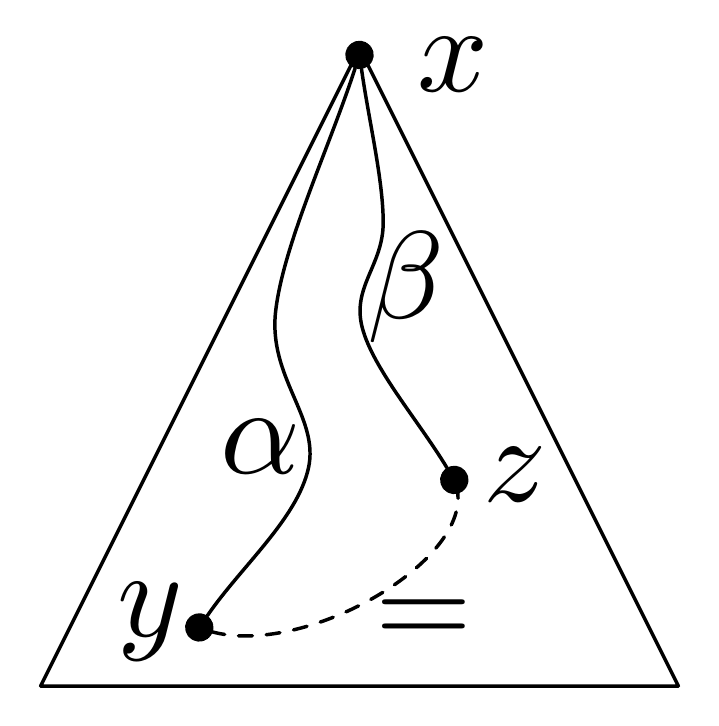}&\includegraphics[scale=0.25]{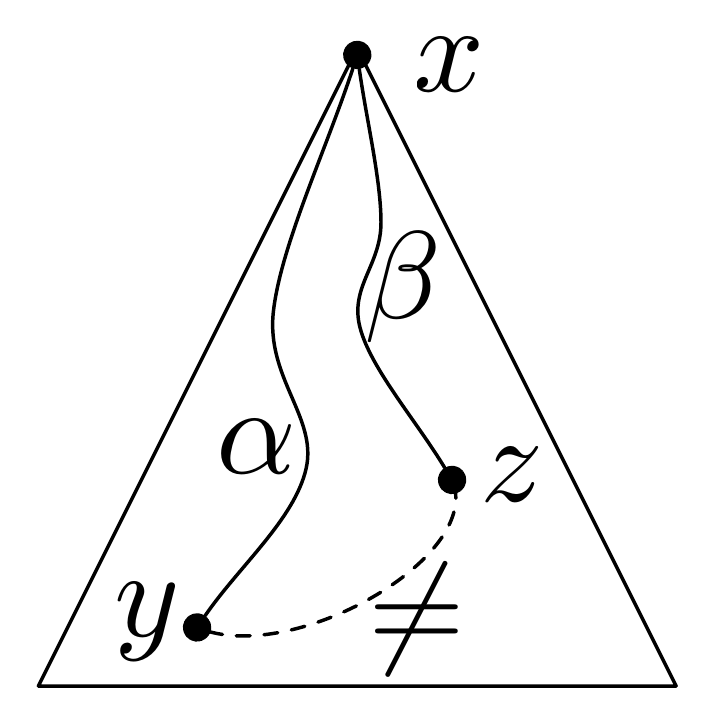}\\
   (a)&(b)&(c)
   \end{tabular}
   \end{center}
   \caption{In $x$ it holds: (a) The {\em modal} diamond $\tup{\alpha}$; (b) the data-aware diamond $\tup{\alpha=\beta}$; (c) the data-aware diamond $\tup{\alpha\neq\beta}$; }\label{fig:diamonds}
\end{figure}


Imagine that $\alpha$ simply expresses ``go to child; $\varphi$ holds; end of path''.  Then the node expression $\tup{\alpha}$ is translated as $\Diamond  \tilde \varphi$ in BML, where $\tilde\varphi$ is the recursive translation of $\varphi$ to the BML language. For illustrating a more complex path expression, suppose that the path expression $\beta$ expresses ``go to child; $\varphi$ holds; go to child; $\psi$ holds; end of path''. Then the node expression $\tup{\beta}$ is translated to the language of BML as $\Diamond (\tilde\varphi\wedge\Diamond \tilde\psi)$. \corexpath generalizes the `diamond' $\Diamond$ operator of BML to complex diamonds $\tup{\alpha}$, where $\alpha$ describes a property on a path. 
Conversely, any formula $\Diamond\varphi$ of BML can be straightforwardly translated to \corexpath as $\tup{\alpha}$, where $\alpha$ expresses ``go to child; $\varphi$ holds; end of path''.

By the (finite) tree model property of BML, the validity of a formula with respect to the class of all Kripke models is equivalent to the validity in the class of (finite) tree-shaped Kripke models. Since there are truth-preserving translations to and from \corexpath, it is not surprising that there exist axiomatizations of the node expressions fragment of \corexpath with `child' as the only accessibility operator. Interestingly, there are also axiomatizations of the path expressions fragment of it. Even more, there are also axiomatizations of all single axis fragments of \corexpath (those where the only accessibility relation is the one of `child', `descendant', `sibling', etc.), and also for the full \corexpath language~\cite{cateLM10}. 

\coredataxpath~\cite{BMSS09:xml:jacm} ---here called~\xpathd--- is the extension of~\corexpath with \linebreak (in)equality tests between attributes of elements in an XML document. The resemblance with modal languages is now more distant, since the models of \xpathd cannot be represented by Kripke models. A first attempt to represent a data tree as a Kripke model, would be to let any data value $v$ in the data tree correspond to a propositional letter $p_v$ in the Kripke model~\cite{KR16}. However, this would be unfair: in BML, $p_v$ is a licit formula expressing ``the value is $v$'' but this kind of construction is not permitted in \xpathd. Indeed, \xpathd can only compare data values by equality or inequality 
at the end of paths, but it cannot compare the data value of a node with a constant. The rationale of this feature is twofold: on the one hand, it remains a finitary language; on the other, its semantics is invariant over renaming of data values. \xpathd augments \corexpath expressivity with `data-aware diamonds' of the form $\tup{\alpha=\beta}$ and $\tup{
\alpha\neq\beta}$. The former is true at $x$ if there are descendants $y$ and $z$ of $x$ such that $\alpha$ is 
true at $(x,y)$ and $\beta$ is true at $(x,z)$, and $y$ and $z$ have the same data value (see Figure~\ref{fig:diamonds}(b)). The latter is true at $x$ if there are $y$ and $z$ as before but such that $y$ and $z$ have distinct data values (see Figure~\ref{fig:diamonds}(c)).
Observe that $\lnot\tup{\alpha=\beta}$ expresses that all pairs of paths satisfying $\alpha$ and $\beta$ respectively, starting in $x$, end up in nodes with different data values, while $\tup{\alpha\neq\beta}$ expresses that there is a pair of paths satisfying $\alpha$ and $\beta$ respectively which end up in nodes with different data value. One can see that $\tup{\alpha\neq\beta}$ is not expressible in terms of Boolean combinations of expressions of the form $\tup{\cdot=\cdot}$.

Whilst the model theory of \xpathd was recently investigated both for the node expressions fragment~\cite{ICDT14,ICDT14Jair}
and for the path expressions fragment,~\cite{ADF14, ADF14journal}, 
the only other research into the proof theory of $\xpathd$ outside of this work is for a simple fragment~\cite{datagl}.




Obtaining a complete axiomatization has applications in static analysis of queries, such as optimization through query rewriting. The idea here is to see equivalence axiom schemes as (undirected) rules for the rewriting of queries; in this context, the completeness of the axiomatic system means that a semantic equivalence between two node or path expressions must have a corresponding chain of rewriting rules that transform the first expression into the second one. Therefore, obtaining an axiomatization, along with all the proofs of the theorems involved in the demonstration of its completeness, can be used as a first step in finding effective strategies for rewriting queries into equivalent but less complex forms. 

Studying complete axiomatizations can also give us an alternative method for solving the validity problem, which is undecidable for the full logic \coredataxpath~\cite{GeertsF05}, but it is decidable when the only axis present in the language is `child', and in fact, also when adding `descendant'~\cite{Figueira12ACM} (and also for other fragments).


\subsection{Contributions}\label{subsec:contr}

We give sound and complete axiomatizations for \xpathd with `child' as the only axis. We extend the axiomatization of \corexpath given in~\cite{cateLM10} with the needed axiom schemes to obtain all validities of $\coredataxpath$. Our axiomatizations will be equational: all axiom schemes are of the form $\varphi\equiv\psi$ for node expressions $\varphi$ and $\psi$ or of the form $\alpha\equiv\beta$ for path expressions $\alpha$ and $\beta$, and inference rules will be the standard ones of equational logic. We show that an equivalence $\varphi\equiv\psi$ is derivable in the axiomatic system if and only if for any data tree, and any node $x$ in it, either $\varphi$ and $\psi$ are true at $x$ or  both are false at $x$. We also present a similar result for path expressions:  an equivalence $\alpha\equiv\beta$ is derivable if and only if for any data tree, and any pair of nodes $(x,y)$ in it, either $\alpha$ and $\beta$ are true at $(x,y)$ or  both are false 
at $(x,y)$. Our completeness proof relies on a normal form theorem for expressions of \xpathd with `child' axis, and a construction of a canonical model for any consistent formula in normal form inspired by~\cite{fine75}. 

We proceed gradually. To warm up, we first show an axiomatization for the fragment of $\xpathd$ with all Boolean operators, with data-aware diamonds of the form $\tup{\alpha=\beta}$, but keeping out those of the form $\tup{\alpha\neq\beta}$. This fragment is still interesting since it allows us to express the {\em join} query constructor. Then we give the axiomatization for the full \xpathd with `child' axis, whose proof is more involved but uses some ideas from the simpler case. 

\subsection{Related Work} \label{subsec:relwork}

As we mentioned before, there exist axiomatizations for navigational fragments of XPath with different axes~\cite{cateLM10}. Axiomatizations of other fragments of \corexpath have been investigated in~\cite{BenediktFK05}, and extensions with XPath 2.0 features have been addressed in~\cite{CateM09}. We found only a few attempts of axiomatizing modal logics with some notion of data value.

A logical framework to reason about data organization is investigated in~\cite{ArtemovK96}. 
They introduce {\em reference structures} as the model to represent data storage,
and a {\em propositional labeled modal language} to talk about such structures.
Both together model memory configurations, i.e., they allow storing data files, and
retrieving information about other cells' content and location of files. A sentence $\llbracket m \rrbracket A$ is
read as ``memory cell $m$ stores sentence $A$''. Then, data is represented by mean of sentences:
for instance, if data $c_i$ represents a number $N$, $c_i$ is the sentence ``this is a number $N$''
(same for other sorts of data).
This representation is quite different from our approach. Nevertheless, according to
our knowledge this is one of the first attempts on axiomatizing data-aware logics, by
introducing a Hilbert-style axiomatization.

{\em Tree Query Language} (TQL) is a formalism based on {\em ambient logic},
designed as a query language for semi-structured data. It allows checking schema
properties, extracting tags satisfying a property and also recursive queries.
The TQL data model is {\em information trees}, and the notation to talk about 
information trees is called {\em info-terms}. In~\cite{CardelliG04} an
axiomatization for info-terms is given in terms of a minimal congruence. 
This axiomatization is sound and complete with respect to the information
tree semantics. This is more related to our approach in the sense that we
consider data values as an equivalence relation.

The most closely related work is~\cite{datagl}, where an axiomatization for a very simple fragment of XPath, named {\em DataGL}, was
given. Following our informal description of \xpathd, DataGL allows for constructions of the form $\tup{\eps=\beta}$ and $\tup{\eps\neq\beta}$, where $\eps$ represents the empty path and $\beta = \rtdow [\varphi]$ is a path of the form `go to descendant; $\varphi$ holds; end of path'.
In particular, they introduce a sound and complete sequent calculus for this logic and derive PSPACE-completeness for the validity problem.

%% file: prelim.tex

\paragraph{Syntax of $\xpd$.}
We work with a simplification of
$\xpath$, stripped of its syntactic sugar and with the only axis being the `child' relation, notated $\dow$. We consider fragments of
\xpath that correspond to the navigational part of \xpath 1.0 with
data equality and inequality.  $\xpd$ is a two-sorted language, with
\defstyle{path expressions} (which we write $\alpha, \beta, \gamma$) expressing properties of paths, and
\defstyle{node expressions} (which we write $\varphi, \psi, \rho$), expressing properties of nodes.  

 The language \defstyle{Downward XPath}, notated $\xpd$ is defined
by mutual recursion as follows:
\begin{align*}
  \alpha, \beta \; &\Coloneqq \; \eps \; \midd \; \down \; \midd [\varphi]  \midd
  \alpha\beta \midd \alpha \cup \beta \\
  \varphi, \psi \; &\Coloneqq \; a \midd \lnot \varphi \midd \varphi \land \psi
  \midd \tup{\alpha} \midd \tup{\alpha= \beta} \midd \tup{\alpha \neq \beta}, &&
  a \in \A
\end{align*}
where $\A$ is a finite set of labels.


Other Boolean operators, such as $\vee$, $\to$, are defined as usual. We define the node expressions $\topNode$ and $\botNode$, and the path expression $\botPath$, as follows: 
\begin{center}
\begin{tabular}{rll}
$\topNode$ & $\eqdef$ & $\tup{\eps}$  \\
$\botNode$ & $\eqdef$ & $\neg \topNode$  \\ 
$\botPath$ & $\eqdef$ & $[\neg\tup{\eps}]$  
\end{tabular}
\end{center}

As we remark later, these expressions behave as expected in the axiomatic systems we design.

We notate $\xpdeq$ to the syntactic fragment which does not use the last rule $\tup{\alpha \neq \beta}$. An \defstyle{{\rm $\xpd$}-formula} [resp.\ \defstyle{{\em $\xpdeq$}-formula}] is either a node expression or a path expression of $\xpd$ [resp.\ $\xpdeq$].

\medskip

\newcommand{\len}{{\rm len}}
We define the \defstyle{length} of an $\xpd$-path expression  $\alpha$, notated $\len(\alpha)$, as follows: 
$$
\begin{array}{rcl@{\hskip 1in}rcl}
  \len(\eps) &=& 0
  &
  \len(\alpha\beta) &=& \len(\alpha)+\len(\beta) \\
  \len(\dow) &=& 1
  &
  \len(\alpha\cup\beta) &=& \max\{\len(\alpha),\len(\beta)\} \\
  \len([\varphi]) &=& 0
  &&&
\end{array}
$$
%
%
%
We write $\md$ to denote the \defstyle{downward depth} \cite{ICDT14Jair} of an $\xpd$-formula,
which measures `how deep' such formula can see, and is defined as follows:
\begin{center}
\begin{tabular}{rllrll}
$\md(a)$ & $=$ & $0$ & $\md(\eps)$ & $=$ & $0$ \\
$\md(\neg\varphi)$ & $=$ & $\md(\varphi)$ & $\md(\dow)$ & $=$ & $1$  \\ 
$\md(\varphi\land\psi)$ & $=$ & $\max\{\md(\varphi),\md(\psi)\}$ & $\md([\varphi])$ & $=$ & $\md(\varphi)$  \\
$\md(\tup{\alpha})$ & $=$ & $\md(\alpha)$ & $\md(\alpha\beta)$ & $=$ & $\max\{\md(\alpha),\md(\beta),\len(\alpha)+\md(\beta)\}$\\ 
$\md(\tup{\alpha*\beta})$ & $=$ & $\max\{\md(\alpha),\md(\beta)\}$ & $\md(\alpha\cup\beta)$ & $=$ & $\max\{\md(\alpha),\md(\beta)\}$ \\   
\end{tabular}
\end{center}
where $a\in\A$ $\varphi,\psi$ are node expressions of $\xpd$, $\alpha,\beta$ are path expressions of $\xpd$, and $*\in\{=,\neq\}$.
%
%

\paragraph{Data trees.}
We introduce {\em data trees}, the structures in which we interpret $\xpd$-formulas. Usually, a data tree is defined as a tree whose every node contains a label from a finite alphabet $\A$ and a data value from an infinite domain. An example of a data tree is depicted in Figure  \ref{fig:example-tree}(a).
\begin{figure}[ht]
   \begin{center}
   \begin{tabular}{c@{\hskip .5in}c}
   \includegraphics[scale=0.275]{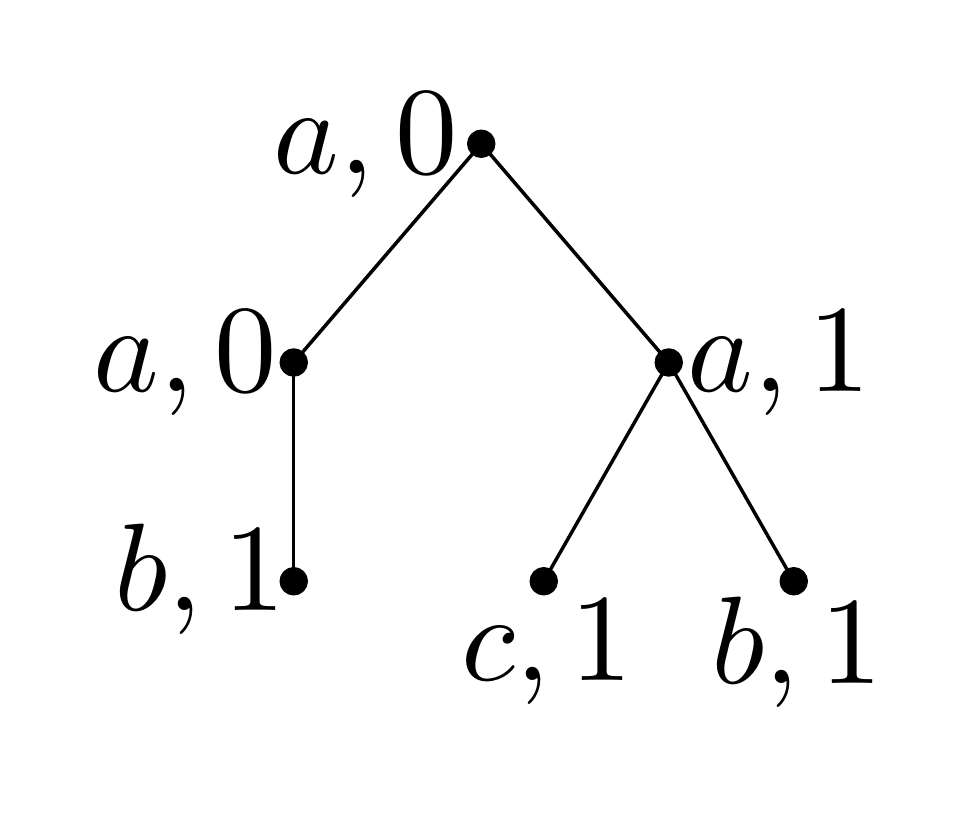}&\includegraphics[scale=0.275]{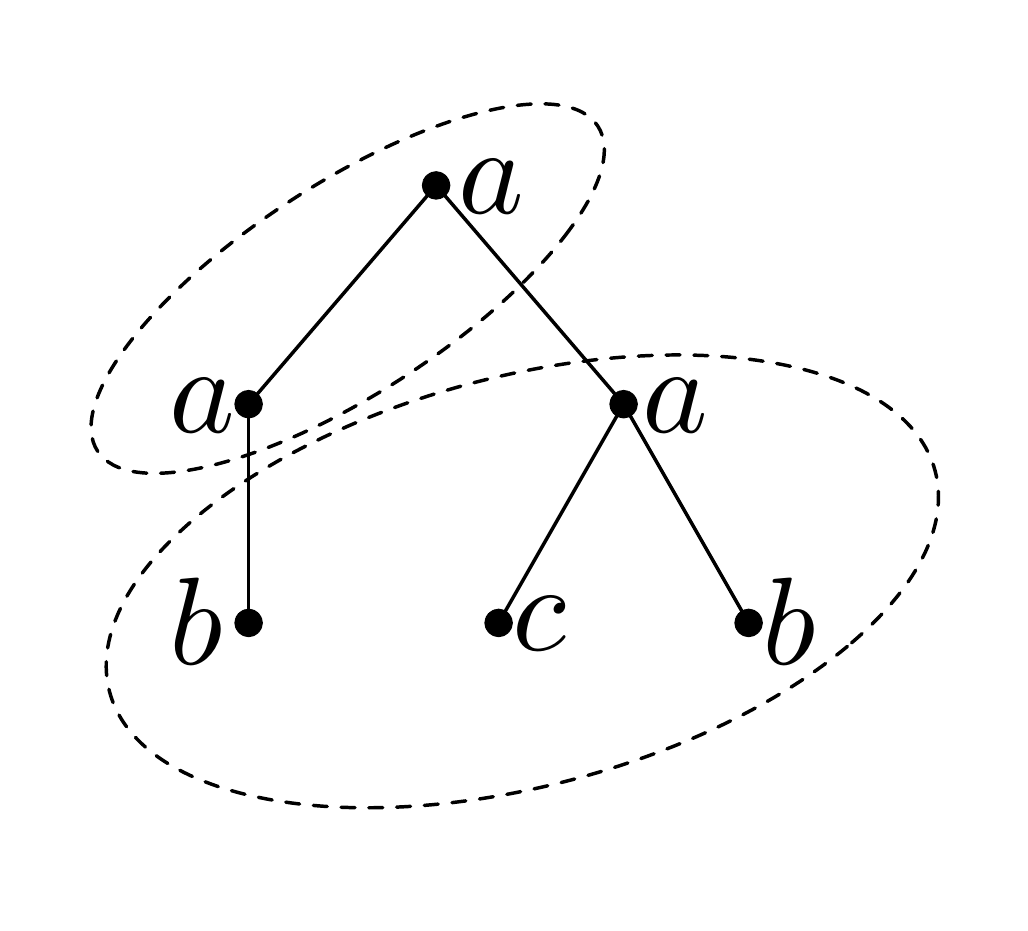}\\
   (a)&(b)
   \end{tabular}
   \end{center}
   \caption{(a) A data tree. Nodes are tagged with $(\ell,n)$ meaning that its label is $\ell$ and its data-value is $n$. (b) Our view of data tree: a node-labeled tree and a partition over its nodes.}\label{fig:example-tree}
\end{figure}
Our logical language, whose formal semantics is defined below, will be able to compare the data value of two nodes by equality or inequality but it will not be able to compare against a concrete value. Hence we will work with an abstraction of the usual definition of data tree: instead of having data values in each node of the tree, we have an equivalence relation between the nodes or, equivalently, a partition. We identify two nodes with the same data value as being related by the equivalence relation, or belonging to the same equivalence class in the partition ---see Figure  \ref{fig:example-tree}(b). While this is not the classical view of a data tree, it is more convenient for our purposes, and it is equivalent, as far as the semantics of our logical language is concerned.
\begin{definition}
Let $\A$ be a finite set of {\em labels}, a \emph{data tree} $\Tt$ is a pair $(T,\pi)$, where $T$ is a tree (i.e.\ a connected acyclic graph such that every node has exactly one parent, except the root, which has no parent) whose nodes are labeled with elements from $\A$, and $\pi$ is a partition over the nodes of $T$. We use $T$ indistinctly to denote the set of nodes of $\Tt$ or the structure of the labeled tree. Given two nodes $x,y\in T$ we write $x\child y$ if $y$ is a 
child of $x$ and $x\childp{i} y$ (for $i\geq 1)$ as a short for
$$
(\exists z_0,\dots,z_{i}\in T)\ x=z_0\child z_1\child\dots\child z_i=y.
$$
Observe that in particular $x\child y$ iff $x\childp{1}y$. 

We denote with $[x]_\pi$ the class of $x$ in the partition $\pi$, and with $label(x)\in \A$ the node's label.
We say that $\Tt,x$ is a {\em pointed data tree}, and $\Tt,x,y$ is a {\em two-pointed data tree}.
\end{definition}
%

%
%

\paragraph{Semantics of \boldmath{$\xpd$}.}

Let us introduce the semantics of $\xpd$-formulas.
%
Let  $\Tt=(T,\pi)$ be a data tree. We define the semantics of $\xpd$ on $\Tt$ (notated as $\dbracket{\, \cdot \,}^\Tt$) in Table  \ref{tab:semantics}.
\begin{table}[ht]
\begin{align*}
  \dbracket{\eps}^\Tt & = \{(x,x) \mid x \in T\}\\
  \dbracket{\dow}^\Tt & = \{(x,y) \mid x \child y\} \\
  \dbracket{\alpha \beta}^\Tt & = \{(x,z) \mid \hbox{there exists } y\in T \hbox{ with } \ (x,y)
     \in \dbracket{\alpha}^\Tt,(y,z) \in \dbracket{\beta}^\Tt\} \\
  \dbracket{\alpha \cup \beta}^\Tt & = \dbracket{\alpha}^\Tt \cup \dbracket{\beta}^\Tt\\
  \dbracket{[\varphi]}^\Tt & = \{(x,x) \mid x \in \dbracket{\varphi}^\Tt\}\\
  \dbracket{a}^\Tt & = \{ x \in T  \mid \lbl(x) = a \}\\
  \dbracket{\lnot \varphi}^\Tt & =  T \setminus \dbracket{\varphi}^\Tt\\
  \dbracket{\varphi \land \psi}^\Tt & = \dbracket{\varphi}^\Tt \cap \dbracket{\psi}^\Tt\\
  \dbracket{\tup{\alpha}}^\Tt & = \{ x \in T \mid \hbox{there exists } y\in T \hbox{ with }
     (x,y) \in \dbracket{\alpha}^\Tt \} \\
  \dbracket{\tup{\alpha = \beta}}^\Tt & = \{ x \in T \mid \hbox{there exist } y, z\in T \hbox{ with } (x,y) \in \dbracket{\alpha}^\Tt, (x,z) \in
  \dbracket{\beta}^\Tt, [y]_\pi=[z]_\pi)\} \\
  \dbracket{\tup{\alpha \neq \beta}}^\Tt & = \{ x \in T \mid \hbox{there exist } y,z \in T \hbox{ with } (x,y) \in \dbracket{\alpha}^\Tt, (x,z) \in
  \dbracket{\beta}^\Tt,[y]_\pi\neq[z]_\pi)\}
  \end{align*}	
  \caption{Semantics of $\xpd$}\label{tab:semantics}
  \end{table}

 Let $\Tt,x$ be a pointed data tree and $\varphi$ a node expression, we write
 $\Tt,x\models\varphi$ to denote $x\in\dbracket{\varphi}^\Tt$, and we say
 that $\Tt,x$ satisfies $\varphi$ or that $\varphi$ is true at $\Tt,x$.
 Let $\Tt,x,y$ be a two-pointed data tree and $\alpha$ a path expression,
 we write $\Tt,x,y\models\alpha$ to denote $(x,y)\in\dbracket{\alpha}^\Tt$,
 and we say that  $\Tt,x,y$ satisfies $\alpha$ or that $\alpha$ is true at $\Tt,x,y$. We say that a node expression $\varphi$ is satisfiable in a data tree~$\Tt$ if $\Tt,r\models\varphi$, where $r$ is the root of $\Tt$. We say that $\varphi$ is satisfiable if it is satisfiable in some data tree $\Tt$.

\begin{example}\label{ex:semantics}Consider the data tree of Figure  \ref{fig:example-tree} with root $x$.
\begin{enumerate}
\item $\tup{\dow=\dow[a]\dow[b]}$ is true at $x$ because there is a path of length 1, and there is a path of length 2 (with labels $a$ in the second node and $b$ in the third one) ending in nodes with the same data value.

\item $\tup{\eps=\dow\dow}$ is false at $x$ because there are no paths of length 2 ending in nodes with the same data value as $x$.

\item  $\lnot\tup{\dow\dow\neq\dow\dow}$ is true at $x$ because all paths of length 2 end in nodes with the same data value.

\item\label{example:4} $\tup{\dow[a]\dow[b]=\eps}$ is false at $x$ because no path of length 2 with labels $a$ in the second node and $b$ in the third node, end in a node with the same data value as $x$. 

\item\label{example:5} $\tup{\dow[a\wedge \tup{\dow[b]}]=\eps}$ is true at $x$ because $x$ has a child with label $a$, satisfying $\tup{\dow[b]}$, and with the same data value as $x$.

\end{enumerate}

\end{example}

 We say that two node expressions $\varphi,\psi$ of $\xpd$ are equivalent
(notation: $\models\varphi\equiv\psi$) iff $\dbracket{\varphi}^\Tt=\dbracket{\psi}^\Tt$
for all data trees $\Tt$. Similarly, path expressions $\alpha,\beta$ of $\xpd$
are equivalent (notation: $\models\alpha\equiv\beta$) iff $\dbracket{\alpha}^\Tt=\dbracket{\beta}^\Tt$
for all data trees $\Tt$. 

Let $\Tt,x,y$ and $\Tt',x',y'$ be two-pointed data trees, we say that
$\Tt,x\eqfull\Tt',x'$  [resp. $\Tt,x\eqres\Tt',x'$] iff for all node expressions $\varphi$ of $\xpd$ [resp. $\xpdeq$] 
we have $\Tt,x\models\varphi$ iff
$\Tt',x'\models\varphi$, and we say that $\Tt,x,y\eqfull\Tt',x',y'$ [resp. $\Tt,x,y\eqres\Tt',x',y'$] iff for all path expressions $\alpha$ of $\xpd$ [resp. $\xpdeq$] $\Tt,x,y\models\alpha$ iff $\Tt',x',y'\models\alpha$.  

\newcommand{\restr}[1]{\!\restriction\!\!{{\scriptstyle #1}}}

Let $\Tt=(T,\pi)$ be a data tree. When $T'$ is a subset of $T$, we write $\pi \restr{T'}$ to denote the restriction of the partition $\pi$ to $T'$.
Let $x\in T$, and let $X$ be the set of $x$ and all its descendants in $T$, i.e.\ $X=\{x\}\cup\{y\in T\mid (\exists i\geq 1)\ x\childp{i} y\}$. We define $\Tt\restr{x}=(T\restr{x},\pi\restr{x})$ as the data tree that consists of the subtree of $T$ that is hanging from $x$, maintaining the partition of that portion.

The logic $\xpd$ is local in the same way as the basic modal logic:
\begin{proposition}\label{prop:local}
Let $(\Tt,\pi)$ be a data tree. Then
\begin{itemize}
\item $\Tt,x\eqfull\Tt\restr{x},x$. 
\item If $y,z$ are descendants of $x$ in $\Tt$, then $\Tt,y,z\eqfull\Tt\restr{x},y,z$.
\end{itemize}
\end{proposition}

\paragraph{Inference rules.}

An $\xpd$-\defstyle{node equivalence} is an expression of the form $\equivInstance{\varphi}{\psi}$, where $\varphi, \psi$ are node expressions of $\xpd$. An $\xpd$-\defstyle{path equivalence}  is an expression of the form $\equivInstance{\alpha} {\beta}$, where $\alpha, \beta$ are path expressions. 
An \defstyle{axiom} is either a node equivalence or a path equivalence.

For $P,Q$ both path expressions or both node expressions, we say that $\equivInstance{P}{Q}$ is \emph{derivable} (or also that $P$ is \defstyle{provably equivalent} to $Q$) from a given set of axioms $\Sigma$ (notation $\Sigma\vdash P\equiv Q$) if it can be obtained from them using the standard rules of equational logic:
	\begin{enumerate}
	\item\label{rule:refl} $\equivInstance{P}{P}$.
	\item\label{rule:sym} If $\equivInstance{P}{Q}$, then $\equivInstance{Q}{P}$.
	\item\label{rule:trans} If $\equivInstance{P}{Q}$ and $\equivInstance{Q}{R}$, then $\equivInstance{P}{R}$.
	\item\label{rule:repl} If $\equivInstance{P}{Q}$ and $R'$ is obtained from $R$ by replacing some occurrences of $P$ by $Q$, then $\equivInstance{R}{R'}$.
	\end{enumerate}
%
 We write \defstyle{$\leqInstance{\varphi}{\psi}$} when $\equivInstance{\varphi \lor \psi}{\psi}$, and we write \defstyle{$\leqInstance{\alpha}{\beta}$} when $\equivInstance{\alpha \cup \beta}{\beta}$. 

\begin{definition}[Consistent Node and Path Expressions] \label{def:con-nodeexp}
Let $\Sigma$ be a set of axioms.
We say that a node  expression [resp. path expression] $P$ of $\xpd$ is \defstyle{$\Sigma$-consistent} 
if $\Sigma\not\vdash\equivInstance{P}{\botNode}$ [resp. $\Sigma\not\vdash\equivInstance{P}{\botPath}$]. 
We define $\con_{\Sigma}$ as the set of \defstyle{$\Sigma$-consistent node expressions}.
\end{definition}

%

%% file: axiom.tex

The main theorems of this article are the ones about the completeness of the proposed axiomatizations. These theorems have two main ingredients: one is a normal form theorem that allows to rewrite any consistent node or path expression in terms of normal forms. The other one is the construction of a canonical model for any consistent node expression in normal form. As it is usually the case, at the same time, we give (through the set of axioms) the definition of consistency. So, an axiom (or an axiom scheme) could have been added either because it was needed to prove the normal form theorem or because it was needed to guarantee that every unsatisfiable formula is inconsistent ---the key fact is that we have a much better intuition of what should be satisfiable than of what should be consistent. Of course we should be careful that the added axioms are sound but that is quite intuitive.          

In Table~\ref{tab:axiomseq} we list the axiom schemes for the fragment $\xpdeq$. This list includes all the axiom schemes from \cite{cateLM10} for the logic \corexpath with single axis `child' (second, third and fourth blocks) and adds the new axiom schemes for data-aware diamonds of the form $\tup{\alpha=\beta}$ (last block). Also, remember that in our data trees each node satisfies exactly one label. We add two axiom schemes to handle this issue (first block). This is unessential for our development, and could be dropped to axiomatize $\xpd$ over data trees whose every node is tagged with multiple labels, with minor changes to the definitions of normal forms.

Let $\axiomRestr$ be the set of all instantiations of the axiom schemes of Table~\ref{tab:axiomseq} for a fixed alphabet~$\A$. In the scope of this section we will often say that a node expression is {\em consistent} meaning that it is $\axiomRestr$-consistent (as in Definition \ref{def:con-nodeexp}).


\begin{table}[h!]
\centering
\colorbox{black!10}{
\begin{tabular}{lrcl}
\hline
\multicolumn{4}{l}{\bf Axioms for labels\vspace{.05in}}
\\
\hypertarget{lbaxone}{{\sf LbAx1}} &$\topNode$&$\equiv$&$\Lor_{a \in \A} a$  
\\
\hypertarget{lbaxtwo}{{\sf LbAx2}} &$\botNode$&$\equiv$&$a \land b$ (where $a \neq b$) 
\vspace{.1in}\\
\hline
\multicolumn{4}{l}{\bf Path axiom schemes for predicates\vspace{.05in}}
\\
\hypertarget{praxone}{{\sf PrAx1}}  &$(\alpha[\lnot\tup{\beta}])\beta$&$\equiv$&$\botPath$
\\ 
\hypertarget{praxtwo}{{\sf PrAx2}} &$[\topNode]$&$\equiv$&$\eps$
\\
\hypertarget{praxthree}{{\sf PrAx3}} &$[\varphi \lor \psi]$&$\equiv$&$[\varphi] \cup [\psi]$
\vspace{.1in}\\
\hline\multicolumn{4}{l}{\bf Path axiom schemes for idempotent semirings\vspace{.05in}}
\\
\hypertarget{isaxone}{{\sf IsAx1}} &$(\alpha\cup\beta)\cup\gamma$&$\equiv$&$\alpha\cup(\beta\cup\gamma)$
\\
\hypertarget{isaxtwo}{{\sf IsAx2}} &$\alpha\cup\beta$&$\equiv$&$\beta\cup\alpha$
\\
\hypertarget{isaxthree}{{\sf IsAx3}} &$\alpha\cup\alpha$&$\equiv$&$\alpha$ 
\\
\hypertarget{isaxfour}{{\sf IsAx4}} &$\alpha(\beta\gamma)$&$\equiv$&$(\alpha\beta)\gamma$
\\
\multirow{2}{*}{\hypertarget{isaxfive}{{\sf IsAx5}}}\rdelim\{{2}{0pt}[] &$\eps\alpha$&$\equiv$&$\alpha$
\\
 &$\alpha\eps$&$\equiv$&$\alpha$ 
 \\
\multirow{2}{*}{\hypertarget{isaxsix}{{\sf IsAx6}}}\rdelim\{{2}{0pt}[]  &  $\alpha(\beta\cup\gamma)$ & $\equiv$ & $(\alpha\beta)\cup(\alpha\gamma)$
\\
& $(\alpha\cup\beta)\gamma$ & $\equiv$ & $(\alpha\gamma)\cup(\beta\gamma)$ 
\\
\hypertarget{isaxseven}{{\sf IsAx7}} & $\botPath\cup\alpha$ & $\equiv$ & $\alpha$
\vspace{.1in}
\\
\hline\multicolumn{4}{l}{\bf Node axiom schemes\vspace{.05in}}
\\
\hypertarget{ndaxone}{{\sf NdAx1}} &$\varphi$&$\equiv$&$\lnot(\lnot\varphi \lor \psi) \lor \lnot(\lnot\varphi \lor \neg\psi)$
\\
\hypertarget{ndaxtwo}{{\sf NdAx2}} &$\tup{[\varphi]}$&$\equiv$&$\varphi$
\\
\hypertarget{ndaxfour}{{\sf NdAx3}} &$\tup{\alpha\cup\beta}$&$\equiv$&$\tup{\alpha}\vee\tup{\beta}$
\\
\hypertarget{ndaxthree}{{\sf NdAx4}} &$\tup{\alpha\beta}$&$\equiv$&$\tup{\alpha[\tup{\beta}]}$
\vspace{.1in}\\
\hline\multicolumn{4}{l}{\bf Node axiom schemes for equality\vspace{.05in}}
\\
\hypertarget{eqaxtwo}{\eqaxtwon} &$\tup{\alpha=\beta}$&$\equiv$&$\tup{\beta=\alpha}$ 
\\
\hypertarget{eqaxthree}{\eqaxthreen} &$\tup{\alpha \cup \beta = \gamma}$&$\equiv$&$\tup{\alpha = \gamma} \lor \tup{\beta = \gamma}$ 
\\
\hypertarget{eqaxfour}{\eqaxfourn} &$\varphi \land \tup {\alpha = \beta}$&$\equiv$&$\tup{[\varphi] \alpha = \beta}$
\\
\hypertarget{eqaxfive}{\eqaxfiven} &$\tup{\alpha=\beta}$&$\leq$&$\tup{\alpha}$
\\
\hypertarget{eqaxseven}{\eqaxsevenn} &$\tup{\gamma[\tup{\alpha=\beta}]}$&$\leq$&$\tup{\gamma\alpha=\gamma\beta}$ 
\\
\hypertarget{eqaxone}{\eqaxonen} &$\tup{\alpha = \alpha}$&$\equiv$&$\tup{\alpha}$
\\
\hypertarget{eqaxsix}{\eqaxsixn} &$\tup{\alpha=\eps}\wedge\tup{\beta=\eps}$&$\leq$&$\tup{\alpha=\beta}$
\\
\hypertarget{eqaxeight}{\eqaxeightn} &$\tup {\alpha = \beta[\tup{\eps = \gamma}]}$&$\leq$&$ \tup{\alpha = \beta \gamma}$
\\
\hline
\end{tabular}
} 
\caption{Axiomatic system $\axiomRestr$ for $\xpdeq$}\label{tab:axiomseq}
\end{table} 

\medskip

Observe that {\sf PrAx4} from \cite[Table 3]{cateLM10}, defined by

\begin{center}
\colorbox{black!10}{
\begin{tabular}{lrcl}
\hline
{\sf PrAx4}\qquad\  &$(\alpha\beta)[\varphi]$&$\equiv$&$\alpha(\beta[\varphi])$
\\
\hline
\end{tabular}
}
\end{center} 

\noindent is not present in our axiomatization because, due to our language design, it is a particular case of \isaxfour.

\medskip

The following syntactic equivalences will be useful for the next sections:

\begin{fact}\label{fact boolean}
As seen in~\cite{cateLM10}, $\topNode$, $\botNode$, and $\botPath$ behave as expected: $\axiomRestr\vdash \equivInstance{\varphi \lor \topNode}{\topNode}$, $\axiomRestr\vdash \equivInstance{ \alpha [\botNode]} { \botPath}$, et cetera. 
Furthermore, we have the following from~\cite[Table 6]{cateLM10}:
	\begin{center}
	\colorbox{black!10}{
	\begin{tabular}{l@{\ \ }l} 
	\hline 
	{\bf Der1} & $\axiomRestr\vdash\equivInstance{\varphi \lor \psi}{\psi \lor \varphi}$ \\
	{\bf Der2} & $\axiomRestr\vdash\equivInstance{\varphi \lor (\psi \lor \rho)}{(\varphi \lor \psi) \lor \rho}$ \\
	{\bf Der12} & $\axiomRestr\vdash\leqInstance{\tup{\alpha\beta}}{\tup{\alpha}}$ \\
	{\bf Der13} & $\axiomRestr\vdash\equivInstance{\tup{\alpha[\botNode]}}{\botNode}$ \\
	{\bf Der21} &  $\axiomRestr\vdash\equivInstance{\alpha[\varphi][\psi]}{\alpha[\varphi\wedge\psi]}$ \\
	\hline
	\end{tabular}
	}
	\end{center}

We note that in order to prove the previous derivations one needs to use {\prax{1}, \prax{2},  \prax{3}, \isaxone, \isaxtwo, \isaxfour, \isaxfive, \isaxsix, \isaxseven, \ndax{1},  \ndax{2}, \ndax{4} and \ndax{3}  }.\footnote{
{\bf Der1} uses \isaxtwo, \ndax{2}, and \ndax{4}. 
{\bf Der2} uses \isaxone (and \ndax{2} and \ndax{4} again). 
We can now derive all the axioms of Boolean algebras by also using \ndax{1}. 
{\bf Der12} also uses \prax{2}, \prax{3}, \isaxfour, \isaxfive, \isaxsix, and \ndax{3}.
{\bf Der13} does not need further axioms.
{\bf Der21} also uses \prax{1} and \isaxseven.}

As a consequence of {\bf Der1}, {\bf Der2} and Huntington's equation \ndax{1}, we can derive  all the axioms of Boolean algebras from the axioms in $\axiomRestr$  \cite{huntington1933new, huntington1933boolean}. In what follows, we will often use the Boolean properties without explicitly referencing them. In particular, we use the fact that $\axiomRestr\vdash\psi \le \botNode$ implies that $\psi$ is an inconsistent node expression, and that $\axiomRestr\vdash\varphi\leq\psi$ implies that $\varphi\land\lnot\psi$ is inconsistent.
\end{fact}

Sometimes we use \isaxone, \isaxfour, \eqax{2}, \eqax{5}, and \eqax{1} without explicitly mentioning them. We omit such steps in order to make the proofs more readable. 

\medskip

It is not difficult to see that the axioms $\axiomRestr$ are sound for $\xpdeq$:

\newcommand{\thmsoundness}[2]
{
\begin{enumerate}
 \item Let $\varphi$ and $\psi$ be node expressions of {\rm #1}. Then  ${#2}\vdash\varphi\equiv\psi$ implies $\models \varphi\equiv\psi$. 
 
 \item Let $\alpha$ and $\beta$ be path expressions of {\rm #1}. Then ${#2}\vdash\alpha\equiv\beta$ implies $\models \alpha\equiv\beta$.
\end{enumerate}
} 
\begin{proposition}[Soundness of $\xpdeq$]\label{prop:corr-restr}
~
\thmsoundness{$\xpdeq$}{\axiomRestr}
\end{proposition} 

\begin{proof}
Equational rules are valid because we have compositional semantics, and 
the proof that all the axioms schemes from Table~\ref{tab:axiomseq} are sound is straightforward.
\end{proof}

%% file: nf.tex

When working in \corexpath, the only diamond in the language is the {\em modal} diamond of the form $\tup{\alpha}$, where $\alpha$ is a path expression. In the absence of data-aware diamonds any node expression $\tup{[\varphi]\dow\beta}$ is equivalent to $\varphi\wedge \tup{\dow[\tup{\beta}]}$. Hence when the only axis is `child', the only path expressions that we need are of the form $\dow[\tup{\cdot}]$, of length $1$, and therefore the only diamonds that we need are of the form $\tup{\dow[\psi]}$, which in the basic modal logic is written simply as $\Diamond\psi$. This rewriting of path expressions is carried out in \cite{cateLM10}, and so normal forms have somewhat the same flavour as in the basic modal logic. 

When data shows up, this rewriting is no longer possible: the node expression $\tup{\alpha=\beta}$ checks if there are nodes with equal data value {\em at the end of paths $\alpha$ and $\beta$.} So these paths cannot be compressed as before. For an easy example, observe that the data-aware diamond $\tup{\dow[a]\dow[b]=\eps}$ is not equivalent to $\tup{\dow[a\wedge \tup{\dow[b]}]=\eps}$ (see items \ref{example:4} and \ref{example:5} of Example~\ref{ex:semantics}).

The normal forms we will introduce are inspired by the classic {\em Disjunctive Normal Form} (DNF) for propositional logic, in which literals are formulas of smaller depth.
Our normal forms will take into account path expressions of arbitrary length, and this makes our definition more involved than the one in \cite{cateLM10}. We introduce them in this section for the language $\xpdeq$. This definition will be extended to the general logic $\xpd$ in \S\ref{sec:nf-neq}.

We define the sets $\nfP{n}$ and $\nfN{n}$, which contain the path and node expressions of $\xpdeq$, respectively,  in normal form at level $n$:
\begin{definition}[Normal form for $\xpdeq$] \label{def:cons-ne}
\begin{eqnarray*}
\nfP{0} &=& \left\{\eps\right\}\\
\nfN{0} &=& \{a\wedge \tup{\eps=\eps}\mid a\in\A\}   \\
\nfP{n+1} &=& \{\eps\} \cup \left\{\dow[\psi]\beta\mid \psi\in \nfN{n},\beta\in \nfP{n} \right\}\\
\nfD{n+1} &=& \left\{\tup{\alpha=\beta}\mid \alpha,\beta\in \nfP{n+1}\right\}\\ 
\nfN{n+1} &=& \left\{a \wedge \bigwedge_{\varphi\in C} \varphi  \wedge \bigwedge_{\varphi\in \nfD{n+1}\setminus C} \neg\varphi \mid C\subseteq \nfD{n+1}, a\in\A\right\} \cap \con_\axiomRestr.
\end{eqnarray*}
\end{definition}

Observe that we define normal forms by mutual recursion among three kinds of sets: $\nfP{n}$, $\nfD{n}$ and $\nfN{n}$ (for some $n$), which are
sets of path expressions, data-aware diamonds, and node expressions, respectively. They consist of expressions that can look forward up to a certain downward depth. The index $n$ indicates which maximum downward depth we are exploring, both in path and node expressions.
Base cases are the simplest expressions of each kind (depth  $0$). New path expressions are constructed by using node and path expressions $\psi$ and~$\beta$ from a previous level of their respective type, and exploring one more step using~$\dow$. 
Data-aware diamond expressions are auxiliary expressions consisting of equalities between two path expressions of the same level.
Finally, node expressions in normal form at some level $n$ are formed of consistent conjunctions of positive and negative data-aware diamond expressions of level $n$. Notice that at each level $i$, each conjunction in $\nfN{i}$ has one conjunct of the form $a$ with $a\in \A$ which provides a label for the current node. 
Finally, let us remark that it would suffice that $\nfN{0}$ contains formulas of the form $a$, for $a\in\A$.  However, we include instead formulas of the form $a\wedge\tup{\eps=\eps}$ (containing the tautology $\tup{\eps=\eps}$) only for technical reasons.

\begin{example} \label{example:normalForm}
Let us see some examples of expressions in normal form. We consider only two labels $a$ and $b$, and ignore  
redundancies (if we write $\tup{\alpha=\beta}$, we do not write $\tup{\beta=\alpha}$). The sets $\nfP{1}$ and $\nfD{1}$ are as follows:
\begin{eqnarray*}
	\nfP{1}&=&\{\dow[a\land\tup{\eps=\eps} ]\eps,\dow[b\land\tup{\eps=\eps} ]\eps,\eps\}
	\\
	\nfD{1}&=&\{\tup{\dow[a\land\tup{\eps=\eps} ]\eps=\dow[b\land\tup{\eps=\eps} ]\eps},\tup{\eps=\dow[a\land\tup{\eps=\eps} ]\eps},\tup{\eps=\dow[b\land\tup{\eps=\eps} ]\eps}, \\
	& & \tup{\dow[a\land\tup{\eps=\eps} ]\eps=\dow[a\land\tup{\eps=\eps} ]\eps},\tup{\dow[b\land\tup{\eps=\eps} ]\eps=\dow[b\land\tup{\eps=\eps} ]\eps},\tup{\eps=\eps}\}
\end{eqnarray*}
An example of a node expression in normal form at level 1, i.e.\ a node expression in $\nfN{1}$, is 
\begin{eqnarray*}
\varphi & = & a \wedge \tup{\eps=\eps} \wedge \tup{\dow[a\land\tup{\eps=\eps} ]\eps=\dow[b\land\tup{\eps=\eps} ]\eps} 
\wedge \tup{\dow[a\land\tup{\eps=\eps} ]\eps=\dow[a\land\tup{\eps=\eps} ]\eps} 
\\
&&  \wedge \tup{\dow[b\land\tup{\eps=\eps} ]\eps=\dow[b\land\tup{\eps=\eps} ]\eps}
\wedge \neg\tup{\eps=\dow[a\land\tup{\eps=\eps} ]\eps} \wedge \neg\tup{\eps=\dow[b\land\tup{\eps=\eps} ]\eps}.
\end{eqnarray*}
\end{example}

The following lemmas (\ref{lem:eq_in_psi}, \ref{nextPathIsConjunct} and \ref{lemma:inconsistent}) are very intuitive and their proofs are straightforward. 

\begin{lemma}
\label{lem:eq_in_psi}
Let $\psi\in \nfN{n}$ and $\alpha,\beta\in \nfP{n}$. Let $\Tt,x$ be a pointed data tree,
such that $\Tt,x\models\psi$ and $\Tt,x\models\tup{\alpha=\beta}$. Then $\tup{\alpha=\beta}$
is a conjunct of $\psi$.
\end{lemma}

\begin{proof}
The case when $n=0$ follows from the definitions of $\nfP{0}$ and $\nfN{0}$. 
If $n>0$, since $\alpha,\beta\in \nfP{n}$, by definition of $\nfD{n}$, we have $\tup{\alpha=\beta}\in \nfD{n}$.
Because $\psi\in \nfN{n}$, either $\tup{\alpha=\beta}$ or its negation is a conjunct of $\psi$. Suppose that the latter occurs, then $\Tt,x\models\neg\tup{\alpha=\beta}$,
and, by hypothesis, $\Tt,x\models\tup{\alpha=\beta}$, which is a contradiction.
%
%
\end{proof}

\begin{lemma} \label{nextPathIsConjunct}
Let $\psi\in \nfN{n}$ and $\alpha\in \nfP{n}$.  If $[\psi]\alpha$ is consistent then $\tup{\alpha=\alpha}$ is a conjunct of~$\psi$. As an immediate consequence, if $\tup{\dow[\psi]\alpha}$ is consistent then $\tup{\alpha=\alpha}$ is a conjunct of~$\psi$.
\end{lemma}

\begin{proof}
Since $\alpha\in \nfP{n}$, then either $\alpha=\eps$ or $\alpha$ is of the form
$
\dow[\psi_1]\dots\dow[\psi_k]\eps
$
for some $1\leq k\leq n$, and $\psi_i\in \nfN{n-i}$. If $\alpha=\eps$, we are done, as $\tup{\eps=\eps}$ is always a conjunct of $\psi$ by consistency. 
Else, since $\tup{\alpha = \alpha} \in \nfD{n}$, $\tup{\alpha=\alpha}$ or its negation is a conjunct of $\psi$. By using {\bf Der 21} from Fact \ref{fact boolean}, \eqax{1} and \prax{1} consecutively, one can see that the latter case is not possible, because $[\psi]\alpha$ is consistent. Then $\tup{\alpha=\alpha}$ is a conjunct of $\psi$.
\end{proof}

\begin{lemma}\label{lemma:inconsistent}
For every pair of distinct elements $\varphi,\psi \in \nfN{n}$, $\varphi\wedge\psi$ is
inconsistent.
\end{lemma}

\begin{proof}
If $n=0$, then $\varphi=a \wedge \tup{\eps =\eps}$ and $\psi=b \wedge \tup{\eps =\eps}$,
with $a, b \in \A$ and $a\neq b$. Then by  \lbax{2}, we have $\axiomRestr\vdash\equivInstance{\varphi\wedge\psi}{\botNode}$, i.e.,  $\varphi\wedge\psi$ is inconsistent.

Let $\varphi$ and $\psi$ be distinct normal forms of degree $n>0$, then we have two possibilities: 

\begin{itemize}
 \item If $\varphi$ and $\psi$ differ in the conjunct of the form $a$ with $a \in \A$, then we use an argument similar to the one used for the base case.
 
 \item If not, then there is $\sigma \in \nfD{n}$ such that, without loss of generality, $\sigma$ is a conjunct of $\varphi$ and $\lnot \sigma$ is a conjunct of $\psi$. Therefore, because $\varphi \wedge \psi$ contains $\sigma\wedge\neg\sigma$ as a sub-expression, we have $\axiomRestr\vdash\equivInstance{\varphi\wedge\psi}{\botNode}$, i.e., it is inconsistent.
\end{itemize}
This concludes the proof.
\end{proof}

\begin{lemma}\label{lemma:inconsistent-path}
Let $\alpha,\beta\in\nfP{n}$. If there is a data tree $\Tt$ and nodes $x,y\in T$ such that $\Tt,x,y\models\alpha$ and $\Tt,x,y\models\beta$, then $\alpha=\beta$.
\end{lemma}

\begin{proof}
Let $\alpha=\dow[\psi_1]...\dow[\psi_i]\eps$ and $\beta=\dow[\rho_1]...\dow[\rho_j]\eps$. By definition of the semantics of $\xpdeq$, $i=j$ since $\Tt,x,y\models\alpha$ and $\Tt,x,y\models\beta$. In particular, there are nodes $z_i \in T$, $1 \le k \le i$, such that $\Tt, z_i \models \psi_k, \Tt, z_i \models \rho_k$ for all $1 \le k \le i$.
Using Proposition~\ref{prop:corr-restr}, we obtain that $\psi_k \land \rho_k$ is consistent for all $k=1,\dots,i$ , and thus by Lemma~\ref{lemma:inconsistent} we have that $\psi_k=\rho_k$ for all $k=1,\dots,i$. Then we conclude that $\alpha=\beta$.
\end{proof}

The following lemma is a normal form result for the special case of data-aware `diamond' node expressions in $\nfD{n}$:
\begin{lemma}\label{lem:d} 
Let $n>0$ and $a\in\A$. If $\varphi\in \nfD{n}$ is consistent then there are $\psi_1,\dots,\psi_k\in\nfN{n}$ such that $\axiomRestr \vdash \equivInstance{a\wedge \varphi}{ \bigvee_i\psi_i}$ 
\end{lemma}

\begin{proof}
Take 
$$\psi =  \Lor \left( \left\{a \wedge \bigwedge_{\psi \in C} \psi \wedge \bigwedge_{\psi \in \nfD{n} \setminus C} \lnot \psi \mid \mbox { $C \subseteq \nfD{n}$, $\varphi \in C$ }\right\} \cap \con_\axiomRestr \right).$$
It can be seen that $\axiomRestr \vdash \equivInstance{a  \land \varphi}{\psi}$. Notice that the above disjunction is not empty.
Indeed,  let $\nfD{n}\setminus\{\varphi\}=\{\psi_1,\dots,\psi_k\}$, and define $\varphi_0=a\wedge \varphi$ and $\varphi_{i+1}=\varphi_i\wedge\psi_{i+1}$ if $\varphi_i\wedge\psi_{i+1}$ is consistent and $\varphi_{i+1}=\varphi_i\wedge\lnot\psi_{i+1}$ otherwise. By \ndax{1} either $\varphi_i\wedge\psi_{i+1}$ or $\varphi_i\wedge\lnot\psi_{i+1}$ is consistent, and hence $\varphi_i$ is consistent for all $i$. This means that $\varphi_k$ is consistent and hence it is one of the disjuncts of the above formula.
\end{proof}

The next lemma states that expressions in any $\nfP{n}$ or $\nfN{n}$ are provably equivalent to the union or disjunction, respectively, of expressions in higher levels of $\nfP{n}$ or $\nfN{n}$.

\begin{lemma}\label{lem:m}
Let $m>n$. If $\varphi\in \nfN{n}$ then there are $\varphi_1\dots\varphi_k\in\nfN{m}$ such that $\axiomRestr \vdash \equivInstance{\varphi}{\bigvee_{i}\varphi_i}$. If $\alpha\in \nfP{n}$ then there are $\alpha_1\dots\alpha_k\in\nfP{m}$ such that $\axiomRestr \vdash \equivInstance{\alpha}{\bigcup_{i}\alpha_i}$.
\end{lemma}

\begin{proof}
Observe that it suffices to show this result for $m=n+1$.

The basic idea is to proceed by induction over $n$, first proving the result for $\nfP{n}$ and then using that for the case of $\nfN{n}$.

The base case for $\nfP{0}$ is trivial, while the case for $\varphi \in \nfN{0}$ is easy by taking the disjunction of all node expressions in $\nfN{1}$ which contain the same label as $\varphi$ as a conjunct.

Now for the inductive case $\alpha = \nfP{n+1}$, if $\alpha = \epsilon$ then the result is trivial, and otherwise $\alpha = \dow [\psi] \beta$ with $\psi \in \nfN{n}$ and $\beta \in \nfP{n}$. We now use the inductive hypothesis on $\psi$ and $\beta$ and distribute into a union in $\nfP{n+2}$ using \prax{3} and the path axiom schemes for idempotent semirings. 
The case $\varphi \in \nfN{n+1}$ is solved similarly, using that we know the result holds for path expressions in $\nfP{n+1}$.

%
\end{proof}

It is easier to prove that every consistent formula is satisfiable over expressions in normal form than over the general case, as we can rely on the particular structure of those expressions. However, these proofs would be of little use if expressions in normal form only represented a small subset of all possible expressions. That is not really the case:  Theorem \ref{thm:normal-form} below will show that all node expressions (and also all path expressions) are provably equivalent to a disjunction of expressions in normal form. 

\begin{example}\label{ex:nf-equiv}
As a simple example of these equivalences, take the language with only three labels $a$, $b$, and $c$, and consider the node expression $\varphi = \lnot a$. Then $\axiomRestr\vdash \equivInstance{\varphi}{(b \wedge \tup{\eps=\eps}) \lor (c \wedge \tup{\eps=\eps}})$, where $b \wedge \tup{\eps=\eps}$ and $c \land \tup{\eps=\eps}$ are node expressions in $\nfN{0}$.

For a slightly more complex example, related with Example \ref{example:normalForm}, take the language with only the labels $a$ and $b$, and consider the node expression 
$$\varphi = \tup{[a] \dow [a] = \dow[b]} \land \lnot \tup{\eps = \dow[a]}.$$
 Then 
$
\axiomRestr\vdash \equivInstance{\varphi}{\psi_1 \lor \psi_2},
$ where
\begin{align*}
\psi_1 &= \psi \wedge \neg\tup{\eps=\dow[b\land\tup{\eps=\eps}]\eps}
\\
\psi_2 &= \psi \wedge \tup{\eps=\dow[b\land\tup{\eps=\eps}]\eps}
\\
\psi &= a \wedge \tup{\eps=\eps} \wedge \tup{\dow[a\land\tup{\eps=\eps}]\eps=\dow[b\land\tup{\eps=\eps}]\eps} \wedge \tup{\dow[a\land\tup{\eps=\eps}]\eps=\dow[a\land\tup{\eps=\eps}]\eps}\wedge
\\
& \quad \wedge \tup{\dow[b\land\tup{\eps=\eps}]\eps=\dow[b\land\tup{\eps=\eps}]\eps} \wedge \lnot  \tup{\eps=\dow[a\land\tup{\eps=\eps}]\eps}
\end{align*}
\end{example}

\newcommand{\thmnormalform}[4]
{
Let $\varphi$ be a consistent node expression of {\rm #1} such that $\md(\varphi)= n$. Then 
$
{#2} \vdash \equivInstance{\varphi}{\bigvee_{i}\varphi_i}
$ for some $(\varphi_i)_{1\leq i\leq k}\in {#3}$.
Let $\alpha$ be a consistent path expression of {\rm #1} such that $\md(\alpha)= n$. Then $
{#2} \vdash \equivInstance{\alpha}{\bigcup_{i}[\varphi_i]\alpha_i}
$
for some $(\alpha_i)_{1\leq i\leq k}\in {#4}$ and $(\varphi_i)_{1\leq i\leq k}\in{#3}$. Furthermore, if $\alpha$ is $\eps$ or starting with $\dow$ then
$
{#2} \vdash \equivInstance{\alpha}{\bigcup_{i}\alpha_i}
$
for some $(\alpha_i)_{1\leq i\leq k}\in {#4}$.  
}
\begin{theorem}[Normal form for $\xpdeq$]\label{thm:normal-form}
\thmnormalform{$\xpdeq$}{\axiomRestr}{\nfN{n}}{\nfP{n}}
\end{theorem}

\begin{proof}

We show that if $F$ is a consistent formula of $\xpdeq$ such that $\md(F)= n$, then 
\begin{enumerate}[label={\em \alph*})]
\item  if $F$ is a node expression then $\axiomRestr \vdash \equivInstance{F}{\bigvee_{i}\psi_i}$ for some $(\psi_i)_{1\leq i\leq k}\in\nfN{n}$; 

\item  if $F$ is a path expression then $\axiomRestr \vdash \equivInstance{F}{\bigcup_{i}[\varphi_i]\alpha_i}$ for some $(\alpha_i)_{1\leq i\leq k} \in\nfP{n}$ and $(\varphi_i)_{1\leq i\leq k}\in\nfN{n}$; furthermore, if $F$ is $\eps$ or starts with $\dow$, then $\axiomRestr \vdash \equivInstance{F}{\bigcup_{i}\alpha_i}$ for some $(\alpha_i)_{1\leq i\leq k} \in\nfP{n}$.  
\end{enumerate}
\newcommand{\cc}{\textrm{c}}
%
Because of \eqax{1}, it is enough to prove the lemma for the fragment of $\xpdeq$ without diamonds of the form $\tup{\alpha}$. 
We proceed by induction on the complexity of $F$, denoted by $\cc$, and defined for the specific purpose of this proof as follows:
\begin{center}
\begin{tabular}{rllrll}
$\cc(a)$ & $=$ & $1$ & $\cc(\eps)$&$=$&$0$ \\
$\cc(\neg\varphi)$ & $=$ & $1+\cc(\varphi)$ & $\cc(\dow)$ & $=$ & $1$  \\ 
$\cc(\varphi\land\psi)$ & $=$ & $1+\cc(\varphi)+\cc(\psi)$ & $\cc(\alpha\beta)$ & $=$ & $\cc(\alpha)+\cc(\beta)$ \\
$\cc(\tup{\alpha=\beta})$ & $=$ & $1+\cc(\alpha)+\cc(\beta)$  & $\cc(\alpha\cup\beta)$ & $=$ & $1+\cc(\alpha)+\cc(\beta)$  \\
&&&  $\cc([\varphi])$ & $=$ & $2+\cc(\varphi)$
\end{tabular}
\end{center}
%
Observe that the only node expressions of least complexity (namely, 1) are the labels $a$ or $\tup{\eps =\eps}$, that the only path expressions of least complexity (namely, 0) are those of the form $\eps\dots\eps$, and that the only path expressions of complexity 1 consist of one $\dow$ symbol concatenated with any number of $\eps$ symbols at both sides (that number might be $0$, leaving the path expression $\dow$). Observe also that $\cc(\varphi \wedge \tup{\alpha=\beta})<\cc(\tup{[\varphi]\alpha=\beta})$.

\paragraph{Base case.} If the complexity of $F$ is $0$ then it is the path expression $\eps\dots\eps$, which is provably equivalent to $\eps$ by \isaxfive. Since $\eps \in \nfP{0}$, {\em b}) is immediate. If the complexity of $F$ is $1$ then $F$ is either a node expression which consists of a single label or $\tup{\eps=\eps}$, or $F$ is the path expression $\dow$ (eventually concatenated with $\eps$ but those path expressions are all provably equivalent to $\dow$ by \eqax{1}). If $F=a$  ($a\in \A$), then {\em a}) is immediate, since using $\eqax{1}$ and Boolean reasoning we have $\axiomRestr \vdash F \equiv a\wedge\tup{\eps=\eps}$,  so $a\wedge\tup{\eps=\eps}\in\nfN{0}$, and we finish by applying Lemma~\ref{lem:m}. If $F=\tup{\eps=\eps}$, then {\em a}
) follows from $\eqax{1}$, \lbax{1}, Boolean reasoning and Lemma~\ref{lem:m}. If $F=\dow$, 
 by \isaxfive, \lbax{1} and \prax{2} we have $\axiomRestr \vdash F\equiv \dow[\bigvee_{a\in\A}a]\eps$, and by \prax{3} and \isaxsix, we conclude $\axiomRestr \vdash F\equiv \bigcup_{a\in\A}\dow[a]\eps$ (observe that $\dow[a]\eps \equiv \dow[a \land \tup{\eps =\eps}] \eps \in \nfP{1}$).

\paragraph{Induction.} If the complexity of $F$ is greater than $1$, then $F$ involves some of the operators $\lnot, \land , \tup{\ }, \cup , [\ ]$ or a concatenation different from the ones of complexity $1$ mentioned above. We will perform the inductive step for each of these operators.

If $F = \varphi \land \psi$ or $\lnot\varphi$, we reason as in the propositional case.
If $F = \varphi \land \psi$, we use the inductive hypothesis  on $\varphi$ and $\psi$ to obtain that $\axiomRestr \vdash \equivInstance{F}{\Lor_{i} \varphi_i} \land \Lor_{j} \psi_j$, where $\varphi_i \in \nfN{\md{(\varphi)}}$ for all $i$ and $\psi_j \in \nfN{\md{(\psi)}}$ for all $j$. Actually, we can assume that $\varphi_i, \psi_j \in \nfN{n}$ for all $i,j$ by Lemma \ref{lem:m}.
We now use Boolean distributive laws to prove that $F$ is equivalent to $\Lor_{i, j} (\varphi_{i} \land \psi_{j})$. We then use Lemma~\ref{lemma:inconsistent} plus the consistency of $F$ to remove from that expression redundant conjunctions (if $\varphi_{i} = \psi_{j}$, from $\varphi_i \land \psi_j$ we just keep $\varphi_{i}$) and inconsistent conjunctions (cases where $\varphi_{i} \neq \psi_{j}$).

If $F = \lnot \varphi$, we have by inductive hypothesis  that $\axiomRestr \vdash \equivInstance{\lnot\varphi}{\lnot \Lor_{1 \le i \le m}\varphi_i}$, 
and we can again assume by Lemma~\ref{lem:m} that $\varphi_i \in \nfN{n}$ for all $i$. Expanding each $\varphi_i$ into $a_i \land \Land_{\rho\in C_i} \rho \land \Land_{\rho\in \nfD{n}\setminus C_i} \lnot \rho$ (where $C_i \subseteq \nfD{n}$) and then using Boolean algebra, we have $\axiomRestr \vdash \equivInstance{\lnot\varphi}{\Land_{1 \le i \le m}(\lnot a_i \lor \Lor_{\rho\in C_i} \lnot \rho  \lor \Lor_{\rho\in \nfD{n}\setminus C_i} \rho)}$. 
We now use Boolean distributive laws to get $\axiomRestr \vdash \equivInstance{\lnot\varphi}{\Lor_{\omega \in \Omega} \Land_{1 \le i \le m} \omega(i)}$, where each $\omega(i)$ is either $\lnot a_i$, some $\lnot \rho$ for $\rho \in C_i$, or some $\rho \in \nfD{n} \setminus C_i$, and where $\Omega$ contains all possible strings $\omega$ of length $m$ formed in that way. 
We now use \lbax{1} to get  
$\axiomRestr \vdash \equivInstance{\lnot\varphi}{\Lor_{\omega \in \Omega} \Lor_{a \in \A} a \land \Land_{1 \le i \le m} \omega(i)}$. 
Then, we eliminate repetitions in conjunctions of node expressions, and use properties of Boolean algebra 
 to eliminate inconsistencies; also, as each disjunct has some positive occurrences of some $a \in \A$, we can use $\lbax{2}$ and eliminate the (redundant) occurrences of negation of labels. 
 So now we have that $\axiomRestr \vdash \equivInstance{\lnot \varphi} {\Lor_{\omega \in \Omega} \Lor_{a \in \A} \psi_{\omega,a}}$, where each $\psi_{\omega, a}$ is of the form $a \land \Land_{\rho \in C} \rho \land \Land_{\rho \in D} \lnot \rho$, with $C,D \subseteq \nfD{n}$ and $C \cap D = \emptyset$. However we do not necessarily have  $D = \nfD{n} \setminus C$, so these conjunctions may not add up to be of the form of a node expression in $\nfN{n}$: 
to add the conjunctions needed in order to get normal forms, we proceed as in the proof of Lemma~\ref{lem:d} to complete each $a \land \Land_{\rho \in C} \rho \land \Land_{\rho \in D} \lnot \rho$ into $\Lor_{j} (a \land \Land_{\rho \in C_j} \land \Land_{\rho \in \nfD{n} \setminus C_j} \lnot \rho)$, where $C \subseteq C_j$ for all $j$.
Finally, we have obtained a set $(C_k)_{k \in K}$ of subsets of $\nfD{n}$ such that $\axiomRestr \vdash \equivInstance{\lnot\varphi}{\Lor_{k \in K} (a_k \land \Land_{\rho \in C_k} \rho \land \Land_{\rho \in \nfD{n} \setminus C_k} \lnot \rho)}$. 

If $F$ is of the form $\tup{\alpha=\beta}$, we reason as follows. Since $c(\alpha)< c(\tup{\alpha=\beta})$,
by inductive hypothesis, we have $\axiomRestr \vdash \equivInstance{\alpha}{\bigcup_{i}[\varphi_i]\alpha_i}$ for some $\alpha_i \in\nfP{n}$ and $\varphi_i \in\nfN{n}$ (We may have to use also Lemma~\ref{lem:m}, \prax{3}, \isaxsix, and {\bf Der21} of Fact \ref{fact boolean} if $\md(\alpha) < \md(\tup{\alpha=\beta})$). Similarly, we can turn $\beta$ into $\bigcup_{j}[\psi_j]\beta_j$.  
 Using \eqax{3}, \eqax{4}, and \eqax{2}, we obtain $\axiomRestr \vdash \equivInstance{F}{\Lor_{i, j}\varphi_i\land\psi_j\land\tup{\alpha_i = \beta_j}}$. We then use \lbax{1} and Boolean reasoning to get $\axiomRestr \vdash \equivInstance{F}{\Lor_{i, j} \varphi_i \land \psi_j\land (\Lor_{a \in \A} a \land \tup{\alpha_i = \beta_j})}$, and then distribute the $\land$, use Lemma~\ref{lem:d} over each $a \land \tup{\alpha_i = \beta_i}$, and eliminate inconsistencies using Lemma~\ref{lemma:inconsistent} to obtain our desired result.

Suppose that $F$ is a path expression. Without loss of generality, we can assume that $F\neq[\varphi]$ or $F\neq\alpha \cup \beta$ because in those cases, there exist an equivalent expression of the same complexity that is a concatenation ($[\varphi]\epsilon$ or $(\alpha \cup \beta) \eps$ respectively). Then we only need to prove the result for the concatenation in order to conclude the proof. 
Also without loss of generality we may assume that $F$ does not start with $\eps$, since in that case there exist an equivalent expression of the same complexity that doesn't start with $\eps$.
In case $F$ is a concatenation $F=\alpha \beta$ that doesn't start with $\eps$, we split the proof in three different cases according to the form of $\alpha$ (note that, by \isaxfour, we can assume that $\alpha$ is not a concatenation itself).  

If $F$ is of the form $[\varphi]\beta$ then by \isaxfive we may suppose that $\beta$ ends in $\eps$ and 
by {\bf Der21} of Fact~\ref{fact boolean} we may suppose that $\beta$ is either $\eps$ or starts with $\dow$ (notice that $\eps$ does not count in the complexity of a formula and that the expression in the left hand side of {\bf Der21} of Fact~\ref{fact boolean} has a complexity greater than the one in the right hand side).   By inductive hypothesis, $\varphi$ is provably equivalent to $\bigvee_i \varphi_i$ for some $(\varphi_i)_{i} \in\nfN{n}$ (We may have to use Lemma~\ref{lem:m} to increase the degree). Therefore, by \prax{3}, $[\varphi]$ is provably equivalent to $\bigcup_i[\varphi_i]$. 
By inductive hypothesis, $\beta$ is provably equivalent to $\bigcup_j\beta_j$ for some $(\beta_j)_{j} \in\nfP{n}$ (again, we may have to use Lemma~\ref{lem:m}). Hence $F$ is provably equivalent to $(\bigcup_i[\varphi_i])(\bigcup_j\beta_j)$, and by \isaxsix we conclude that $F$ is provably equivalent to $\bigcup_{i,j}[\varphi_i]\beta_j$ as we wanted to show.

If $F$ is of the form $\dow\beta$, we use inductive hypothesis to show that $\beta$ is provably equivalent to $\bigcup_{i}[\varphi_i]\beta_i$ for some $(\beta_i)_{i} \in\nfP{n-1}$ and  $(\varphi_i)_{i}\in\nfN{n-1}$.
By \isaxsix, we conclude that $F$ is provably equivalent to $\bigcup_{i}\dow[\varphi_i] \beta_i$, and $\dow[\varphi_i] \beta_i\in\nfP{n}$ as we wanted to show.

Finally, if $F$ is of the form $(\gamma\cup\delta)\beta$, then, by \isaxsix, $F \equiv (\gamma \beta)\cup (\delta \beta)$. The result follows from inductive hypothesis for $\gamma \beta$ and $\delta \beta$ (as usual, we may have to use Lemma~\ref{lem:m}, \prax{3}, \isaxsix, and {\bf Der21} of Fact~\ref{fact boolean} to increase the degree). 
%
%
\end{proof}

%% file: compl-downward.tex

In this section we show that for node expressions $\varphi$ and $\psi$ of $\xpdeq$, the equivalence $\varphi\equiv\psi$ is derivable from the axiom schemes of Table~\ref{tab:axiomseq} if and only if $\varphi$ is $\xpdeq$-semantically equivalent to $\psi$. We also show the same result for path expressions of $\xpdeq$. 

We first introduce the main lemma of this section, and then continue to its consequences; as the proof of this lemma is very extensive, we postpone it to Section~\ref{construccion}.
\begin{lemma}\label{lem:construction-restr}
Any node expression $\varphi\in \nfN{n}$ is satisfiable.
\end{lemma}
Based on the above lemma, we arrive to the next theorem, which is the main result of this section:
\newcommand{\thmcompletenessnode}[2]
{
\ 

\begin{enumerate}
\item Let $\varphi$ and $\psi$ be node expressions of {\rm #1}. Then ${#2}\vdash\varphi\equiv\psi$ iff $\models\varphi\equiv\psi$.

\item Let $\alpha$ and $\beta$ be path expressions of {\rm #1}. Then ${#2}\vdash\alpha\equiv\beta$ iff $\models\alpha\equiv\beta$.
\end{enumerate}
}

\begin{theorem}[Completeness of $\xpdeq$]\label{thm:sat-restr}
\thmcompletenessnode{$\xpdeq$}{\axiomRestr}
\end{theorem}

\begin{proof}
Let us show item 1. Soundness follows from Proposition~\ref{prop:corr-restr}. 

For completeness, suppose $\models\varphi\equiv\psi$. 
Now assume that $\varphi$ is consistent and $\psi$ is not. 
On the one hand, by Theorem~\ref{thm:normal-form}, there is $n$ such that $\varphi$ is provably equivalent to $\bigvee_{1 \le i \le k}\varphi_i$, for $\varphi_i\in \nfN{n}$. 
By Lemma~\ref{lem:construction-restr}, to be proved next,  we have that in particular $\varphi_1$ (and hence $\varphi$) is satisfiable. On the other hand, by Proposition~\ref{prop:corr-restr}, $\psi$ is unsatisfiable, and this contradicts the fact that $\models\varphi\equiv\psi$.
This shows that if $\varphi$ is consistent then so is $\psi$. Symmetrically, one can show that if $\psi$ is consistent, then so is $\varphi$. Therefore, either $\varphi$ and $\psi$ are consistent or $\varphi$ and $\psi$ are inconsistent. In the latter case, we trivially have $\axiomRestr\vdash\varphi\equiv\psi$.

In case $\varphi$ and $\psi$ are consistent, by Theorem~\ref{thm:normal-form} and Lemma~\ref{lem:m}, there is $n$ and  node expressions $\varphi'$ and $\psi'$ which are disjunctions of node expressions in $\nfN{n}$ such that $\axiomRestr\vdash\varphi\equiv\varphi'$ and $\axiomRestr\vdash\psi\equiv\psi'$. 

Suppose that $\varphi'$ contains a disjunct $\varphi''$ which is not a disjunct of $\psi'$. By Lemma~\ref{lem:construction-restr}, $\varphi''$ is satisfiable in some data tree $\+T$. By Lemma~\ref{lemma:inconsistent}, for any disjunct $\psi''$ of $\psi'$ we have that $\varphi''\wedge\psi''$ is inconsistent, and by Proposition~\ref{prop:corr-restr}, unsatisfiable. Hence $\psi'$ is not satisfiable in $\+T$, and so $\not\models\varphi\equiv\psi$, a contradiction. The case when $\psi'$ contains a disjunct which is not a disjunct of $\varphi'$ is analogous.

Then $\varphi'$ and $\psi'$ are identical, modulo reordering of disjunctions, and so $\axiomRestr\vdash\varphi'\equiv\psi'$ which implies $\axiomRestr\vdash\varphi\equiv\psi$.

\bigskip

For item 2, soundness follows from Proposition~\ref{prop:corr-restr}. For completeness, suppose $\models\alpha\equiv\beta$. 

Suppose that $\alpha$ is consistent and $\beta$ is not. 
On the one hand, by Theorem~\ref{thm:normal-form},  there is $n$ such that $\alpha$ is provably equivalent to $\bigcup_{1 \le i \le k}[\varphi_i]\alpha_i$, with $\alpha_i\in \nfP{n}$ and $\varphi_i\in\nfN{n}$. 
 Furthermore, we can assume that $[\varphi_1]\alpha_1$ is consistent (if it is not, we simply remove it from the disjunction) and so $\tup{\alpha_1=\alpha_1}$ is a conjunct of $\varphi_1$ by Lemma \ref{nextPathIsConjunct}. By Lemma~\ref{lem:construction-restr}, the node expression $\varphi_1$ is satisfiable. Then, since $\tup{\alpha_1=\alpha_1}$ is a conjunct of $\varphi_1$, the path expression $[\varphi_1]\alpha_1$ is satisfiable, and so $\alpha$ is satisfiable. On the other hand, by Proposition~\ref{prop:corr-restr}, $\beta$ is unsatisfiable, and this contradicts the fact that $\models\alpha\equiv\beta$.
This shows that if $\alpha$ is consistent then so is $\beta$. Symmetrically, one can show that if $\beta$ is consistent, then so is $\alpha$. Therefore, either $\alpha$ and $\beta$ are consistent or $\alpha$ and $\beta$ are inconsistent. In the latter case, we trivially have $\axiomRestr\vdash\alpha\equiv\beta$.

Suppose both $\alpha$ and $\beta$ are consistent.
By Theorem~\ref{thm:normal-form} plus Lemma \ref{lem:m} we have that there is $n$ and path expressions $\alpha_1\dots\alpha_k,\beta_1\dots\beta_\ell$ in $\nfP{n}$ and node expressions 
$\varphi_1\dots\varphi_k,\psi_1\dots\psi_\ell\in\nfN{n}$
such that $\axiomRestr\vdash\alpha\equiv\bigcup_{1 \le i \le k}[\varphi_i]\alpha_i$ and $\axiomRestr\vdash\beta\equiv\bigcup_{1 \le j \le \ell}[\psi_j]\beta_j$. Furthermore, we can assume that $\tup{\alpha_i=\alpha_i}$ is a conjunct of $\varphi_i$ for $i=1\dots k$ and $\tup{\beta_j=\beta_j}$ is a conjunct of $\psi_j$ for $j=1\dots\ell$.
%

Now, suppose that 
\begin{equation}\label{eqn:aux}
[\varphi_i]\alpha_i\notin\{[\psi_1]\beta_1,\dots,[\psi_\ell]\beta_\ell\}
\end{equation}
 for some $i$. Since $\varphi_i\in \nfN{n}$, by Lemma~\ref{lem:construction-restr}, there is a data tree $\Tt=(T,\pi)$ with root $r$ such that $\Tt,r\models\varphi_i$. Since $\tup{\alpha_i=\alpha_i}$ is a conjunct of $\varphi_i$, we have that there is $y\in T$ such that $\Tt,r,y\models\alpha_i$.
 
Let us show that $\Tt,r,y\not\models[\psi_j]\beta_j$ for any $j\leq\ell$. Fix any $j$. By \eqref{eqn:aux}, we have that $\varphi_i\neq\psi_j$ or $\alpha_i\neq\beta_j$. In the first case,  $\Tt,r,y\not\models[\psi_j]\beta_j$ follows from Lemma~\ref{lemma:inconsistent} and Proposition~\ref{prop:corr-restr} (in particular $\Tt,r\not\models\psi_j$). If $\varphi_i=\psi_j$, we have $\alpha_i\neq\beta_j$ and $\Tt,r,y\not\models[\psi_j]\beta_j$ follows from Lemma~\ref{lemma:inconsistent-path}. 

So we have that $\Tt,r,y \models \alpha$ but $\Tt,r,y\not\models \beta$, 
a contradiction with our hypothesis that $\models\alpha\equiv\beta$. Hence for any $i$ there is $j$ such that $[\varphi_i]\alpha_i=[\psi_j]\beta_j$. Analogously one can show that for any $j$ there is $i$ 
such that $[\psi_j]\beta_j=[\varphi_i]\alpha_i$.
Then $\bigcup_{1\le i\leq k}[\varphi_i]\alpha_i$ and $\bigcup_{1\le  j\leq \ell}[\psi_j]\beta_j$ are identical, modulo reordering of unions, and so $\axiomRestr\vdash\alpha\equiv\beta$.
\end{proof}


%

All we need to complete the argument is to prove Lemma~\ref{lem:construction-restr}. Doing this involves the rest of the section.

\subsubsection{Canonical model}\label{construccion}
\input{construction}

\input{verif}





%% file: construction.tex

In order to prove Lemma~\ref{lem:construction-restr}, we construct, recursively in $n$ and for every $\varphi\in \nfN{n}$, a data tree $\Tt^\varphi=(T^\varphi,\pi^\varphi)$ such that $\varphi$ is satisfiable in $\Tt^\varphi$. 

For the base case, if $\varphi\in \nfN{0}$ and $\varphi=a \wedge \tup{\eps =\eps}$ with $a \in\A$,  simply define the data tree $\Tt^\varphi=(T^\varphi,\pi^\varphi)$ where $T^\varphi$ is a tree which consists of the single node $x$ with label $a$, and $\pi^\varphi = \{\{x\}\}$. 

%

Now let $\varphi\in \nfN{n+1}$. Since $\varphi$ is a conjunction as in Definition~\ref{def:cons-ne}, it is enough to guarantee that the following conditions hold (we now observe that these conditions are enough because of \eqax{2}, but we usually avoid these observations of symmetry):

\begin{enumerate}[label=(C\arabic*)]
 \item\label{label} If $a\in \A$ is a conjunct of $\varphi$, then the root $r^\varphi$ of $\Tt^\varphi$ has label $a$. 
 
 \item\label{epsiloneqpsialpha} If $\tup{\eps = \dow [\psi] \alpha }$ is a conjunct of $\varphi$, then there is a child $r^{\bf v}$ (where $\bv=(\psi,\alpha)$; we will introduce this notation in time to formalize the construction) of the root $r^\varphi$ of $\Tt^\varphi$ at which $\psi$ is satisfied and a node $x^{\bv}$ with the same data value as $r^\varphi$
 such that $\Tt^{\varphi}, r^\bv, x^{\bv} \models \alpha$.  
 
 \item\label{psialphaeqrhobeta} If $\tup{\dow [\psi]\alpha = \dow [\rho] \beta}$ is a conjunct of $\varphi$, then there are two children $r^\bu_1$, $r^\bu_2$ of the root $r^\varphi$ of $\Tt^\varphi$ at which $\psi$ and $\rho$ are satisfied respectively, and there are nodes $x^\bu$ and $y^\bu$ with the same data value such that $\Tt^{\varphi}, r^\bu_1, x^{\bu} \models \alpha$ and $\Tt^{\varphi}, r^\bu_2, y^{\bu} \models \beta$.

 \item\label{noepsiloneqpsialpha} If $\lnot \tup{\eps = \dow [\psi] \alpha }$ is a conjunct of $\varphi$, then for each child $z$ of the root $r^\varphi$ of $\Tt^\varphi$ at which $\psi$ is satisfied, if $x$ is a node such that $\Tt^\varphi, z, x \models \alpha$, then the data value of $x$ is different than the one of $r^\varphi$.
 
 \item\label{nopsialphaeqrhobeta} If $\lnot \tup{\dow [\psi]\alpha = \dow [\rho] \beta}$ is a conjunct of $\varphi$, then for each children $z_1, z_2$ of the root $r^\varphi$ of $\Tt^\varphi$ at which $\psi$ and $\rho$ are satisfied respectively, if $w_1, w_2$ are nodes such that $\Tt^\varphi, z_1, w_1 \models \alpha$ and $\Tt^\varphi, z_2, w_2 \models \beta$, then the data values of $w_1$ and $w_2$ are different.
 
\end{enumerate}

Since the construction of the canonical model requires some technical notation that might hinder the understanding of the ideas behind it, we will begin with an intuitive description of the construction.

\subsubsection*{Insight into the construction}

The idea to achieve all the previous conditions is to incrementally build a tree such that it satisfies at its root conditions~\ref{label}, \ref{noepsiloneqpsialpha}, and \ref{nopsialphaeqrhobeta}, then also~\ref{epsiloneqpsialpha} (without spoiling any previous conditions), and finally also~\ref{psialphaeqrhobeta}.

 First we start with a root $r^{\varphi}$ labeled $a$, where $a$ is the label present in $\varphi$ (so that condition~\ref{label} is satisfied). 
At this point in the construction, as we only have one node, conditions~\ref{noepsiloneqpsialpha} and~\ref{nopsialphaeqrhobeta} are trivially satisfied. On the contrary,~\ref{epsiloneqpsialpha} and~\ref{psialphaeqrhobeta} might not be satisfied, and require a positive action (i.e.\ changing the current model) to make them true.
We want to add witnesses that guarantee the satisfaction of~\ref{epsiloneqpsialpha} and~\ref{psialphaeqrhobeta}, and we will achieve this by the use of the inductive hypothesis to construct new trees that we will hang as children of $r^{\varphi}$.
However, adding witnesses jeopardizes the satisfaction of~\ref{noepsiloneqpsialpha} and~\ref{nopsialphaeqrhobeta}, so we need to do it carefully enough.

First we add witnesses in order to satisfy condition~\ref{epsiloneqpsialpha}. If $\psi \in \nfN{n}$, by inductive hypothesis, there exists a tree $\Tt^{\psi}$ such that $\psi$ is satisfied at $\Tt^{\psi}$. Also, if $\tup{\eps=\dow [\psi]\alpha}$ is a conjunct of $\varphi$, by the consistency of $\varphi$, Lemma~\ref{nextPathIsConjunct} and the inductive hypothesis, there is a pair of nodes satisfying $\alpha$ in $\Tt^{\psi}$ and starting at its root. In this case, we will hang a copy of $\Tt^{\psi}$ (or perhaps a slightly modified copy of it constructed in order not to spoil condition \ref{noepsiloneqpsialpha}) as a child of $r^{\varphi}$ and merge the equivalence class of $r^{\varphi}$ to the equivalence class of the endpoint $x^\bv$ of a specially chosen pair of nodes satisfying $\alpha$ and beginning at the root of $\Tt^{\psi}$ (see Figure~\ref{fig:idea1}(a)). 
\begin{figure}[h!]
   \begin{center}
   \begin{tabular}{c@{\hskip .5in}c@{\hskip .5in}c}
   \includegraphics[scale=0.25]{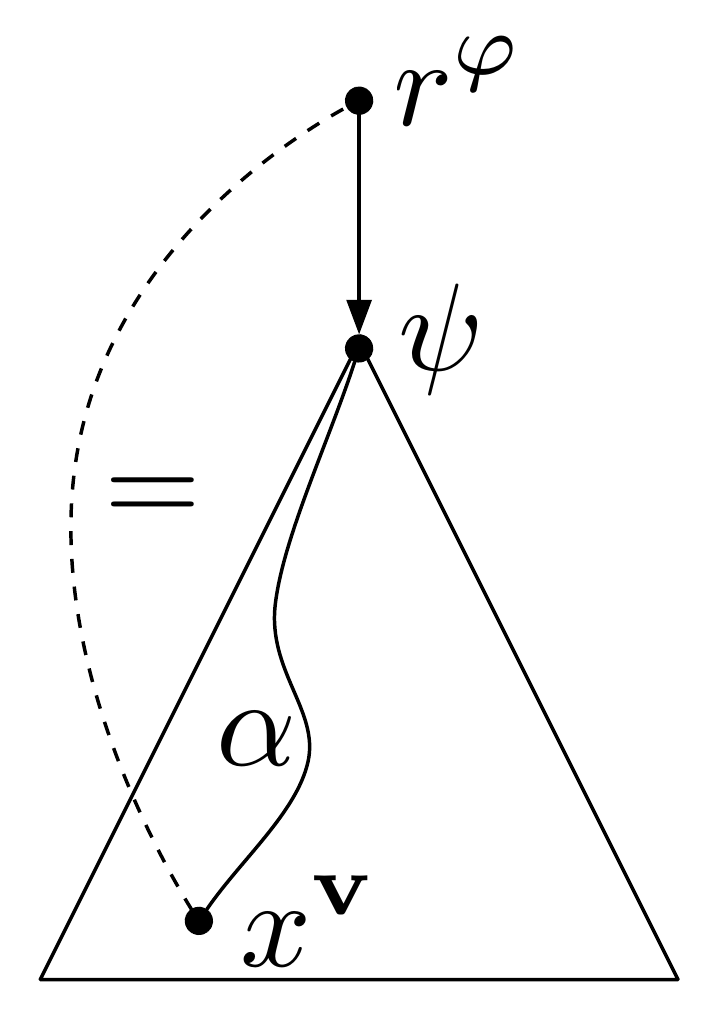}
& \includegraphics[scale=0.25]{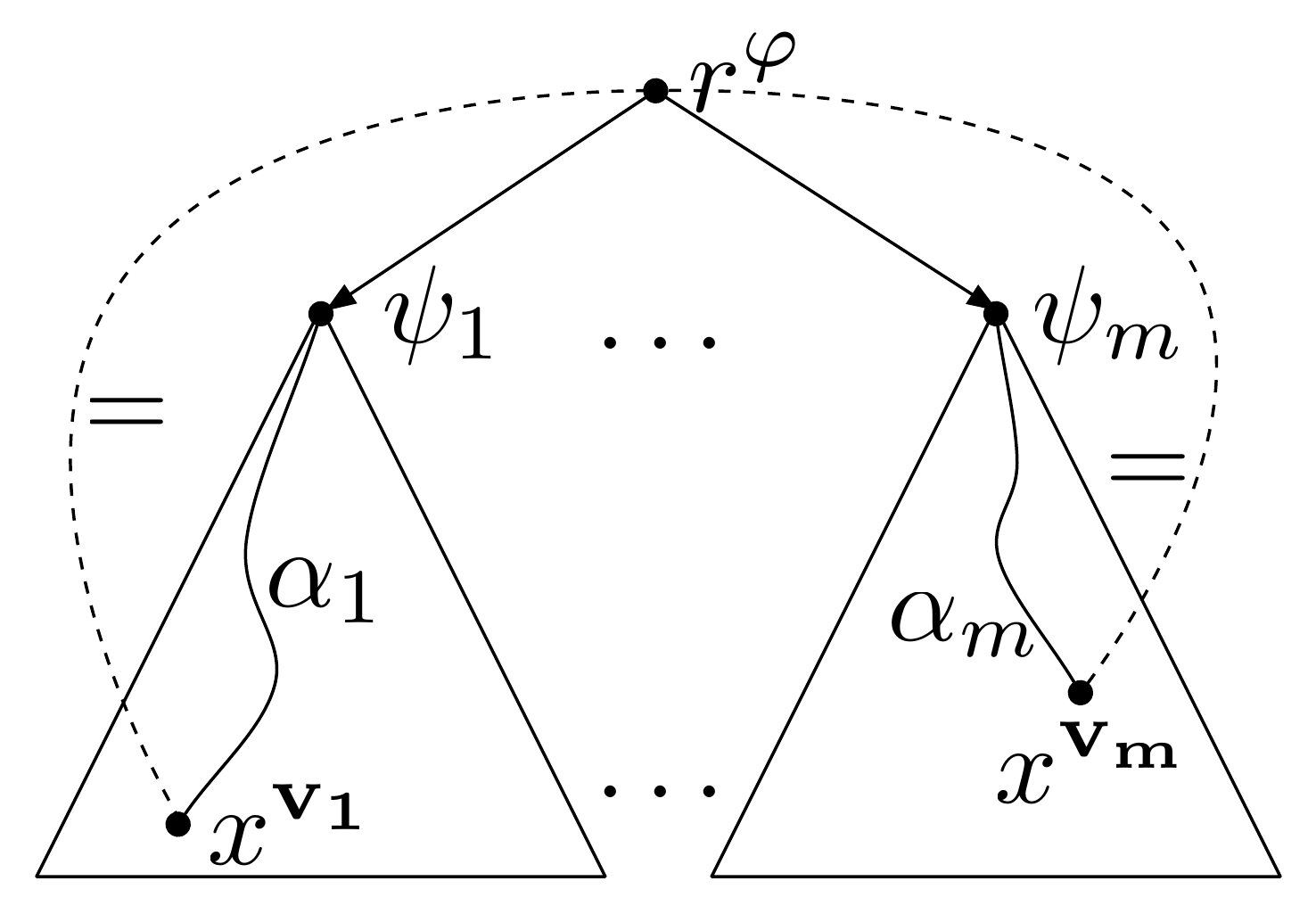}
& \includegraphics[scale=0.25]{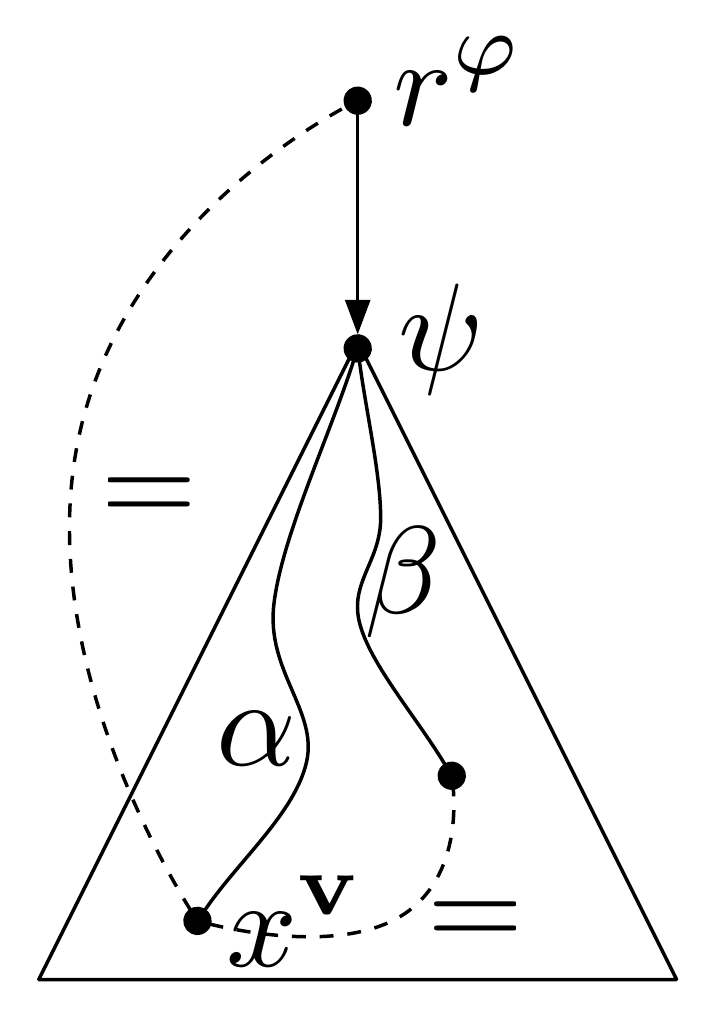}\\
   (a)&(b)&(c)
   \end{tabular}
   \end{center}
   \caption{(a) A witness for $\tup{\eps=\dow[\psi]\alpha}$; (b) we repeat the process of (a) for each conjunct $\tup{\eps=\dow[\psi_1]\alpha_1},\dots,\tup{\eps=\dow[\psi_m]\alpha_m}$ of $\varphi$; (c) by adding a witness for $\tup{\eps=\dow[\psi]\alpha}$, we may be creating an unwanted witness for $\tup{\eps=\dow[\psi]\beta}$.}\label{fig:idea1}
\end{figure}
This is the only merging required; other classes in $\Tt^{\psi}$ remain disjoint from the previous constructed part of $\Tt^{\varphi}$. In this way, we will guarantee condition~\ref{epsiloneqpsialpha} (see Figure~\ref{fig:idea1}(b)).
Since the other equivalence classes of $\Tt^{\psi}$ will remain disjoint from the rest of the tree $\Tt^{\varphi}$ all along the construction and since 
two different normal forms cannot be satisfied at the same point (see Lemma~\ref{lemma:inconsistent} plus Proposition \ref{prop:corr-restr}), the only way in which this process could spoil condition~\ref{noepsiloneqpsialpha} is that there is $\beta \in \nfP{n}$ such that $\lnot \tup{\eps = \dow [\psi] \beta}$ is a conjunct of $\varphi$ and a pair of nodes satisfying $\beta$ in $\Tt^{\psi}$, starting at its root and ending in a point with the same data value as $x^\bv$ (see Figure~\ref{fig:idea1}(c)). 
But Lemma~\ref{Lemma:ClavePlusMinus} ensures that (maybe with changes to $\Tt^\psi$) we can choose $x^\bv$ to avoid this situation. Then, since we only add nodes to the equivalence class of the root $r^{\varphi}$ by this process, the only way in which we could spoil condition~\ref{nopsialphaeqrhobeta} is if $\varphi$ has conjuncts $\tup{\eps=\dow [\psi] \mu}$, $\tup{\eps=\dow [\rho] \delta} $ and $\lnot \tup{\dow [\psi] \mu= \dow [\rho] \delta}$ for some $\psi, \rho \in \nfN{n}$, $\mu, \delta \in \nfP{n}$. But this is clearly unsatisfiable and so our axioms should tell us that it is inconsistent (see \eqax{6}).  

We  now proceed to add witnesses in order to satisfy condition~\ref{psialphaeqrhobeta}. By an argument similar to the one given for condition~\ref{epsiloneqpsialpha}, if $\tup{\dow[\psi] \alpha =\dow [\rho] \beta}$ is a conjunct of $\varphi$, there are, by inductive hypothesis, trees $\Tt^{\psi}$ and $\Tt^{\rho}$ at which $\psi$ and $\rho$ are satisfied and pairs of nodes  satisfying $\alpha$ and $\beta$ starting at their respective roots. We will hang a copy of each of those trees (or perhaps slightly modified copies of them) as children of $r^{\varphi}$ and we will merge the equivalence classes (in $\Tt^{\psi}$ and $\Tt^{\rho}$) of the ending points $x^\bu, y^\bu$ of a specially chosen pair of nodes satisfying $\alpha$ (and starting at the root of $\Tt^{\psi}$) and a specially chosen pair of nodes satisfying $\beta$ (and starting at the root of $\Tt^{\rho}$) as mentioned before (see Figure~\ref{fig:idea2}(a)). %
\begin{figure}[h!]
   \begin{center}
   \begin{tabular}{c@{\hskip .5in}c}
   \includegraphics[scale=0.25]{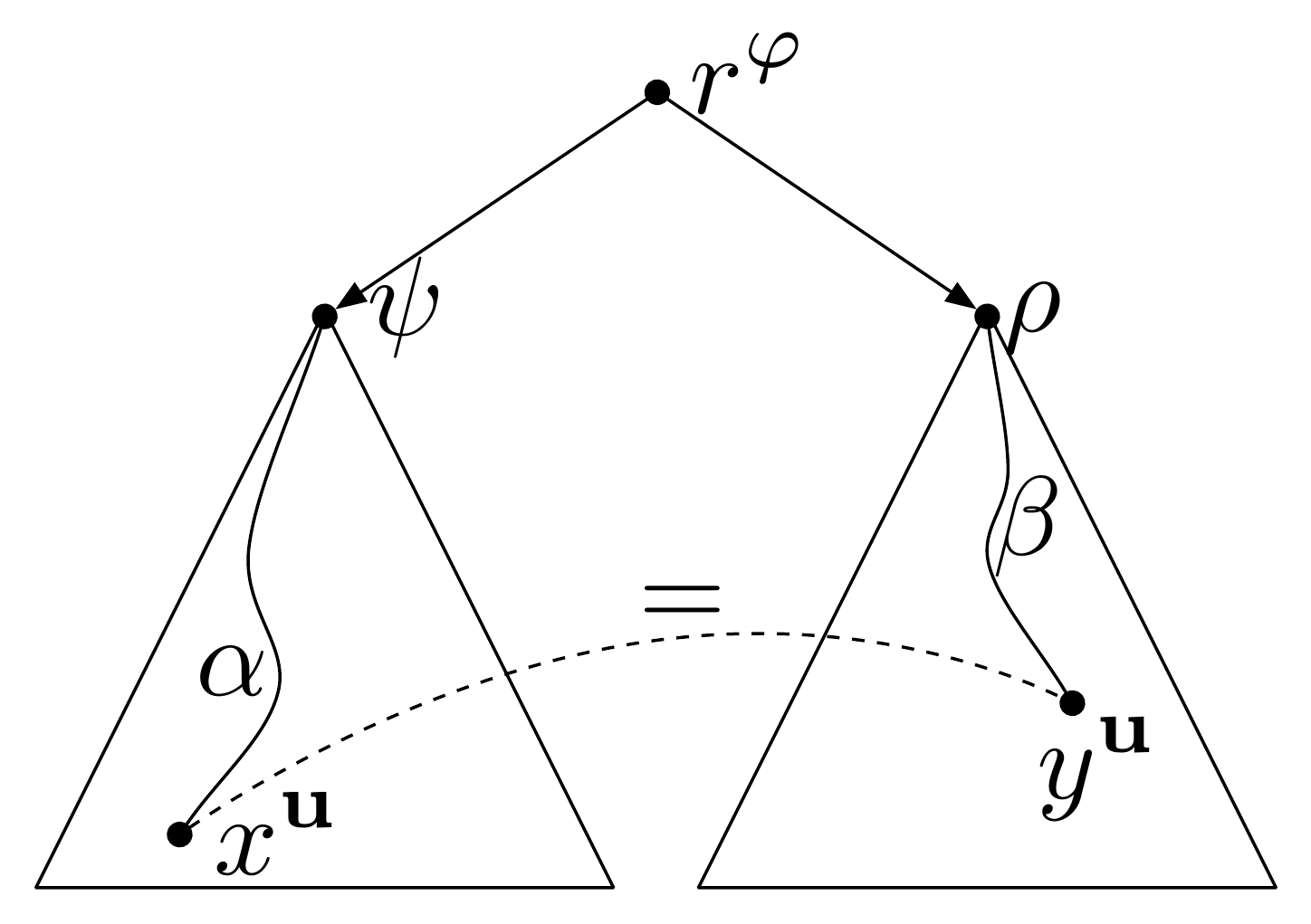}
& \includegraphics[scale=0.25]{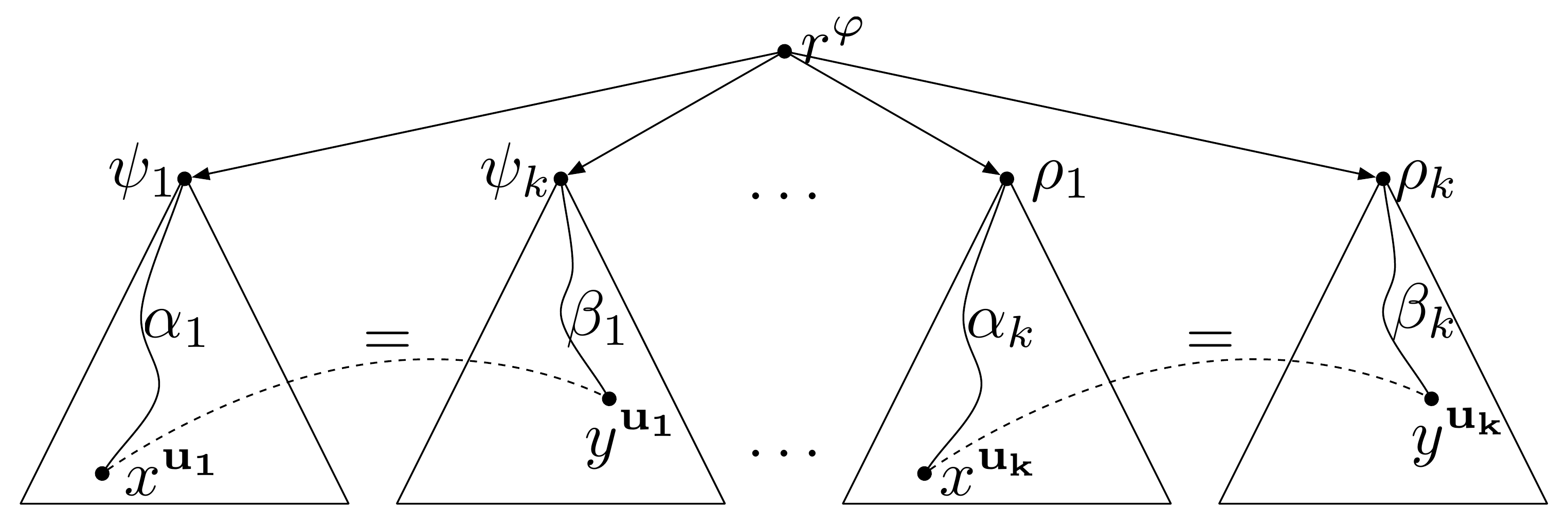}
\\
   (a)&(b)
   \\
  \multicolumn{2}{c}{\includegraphics[scale=0.25]{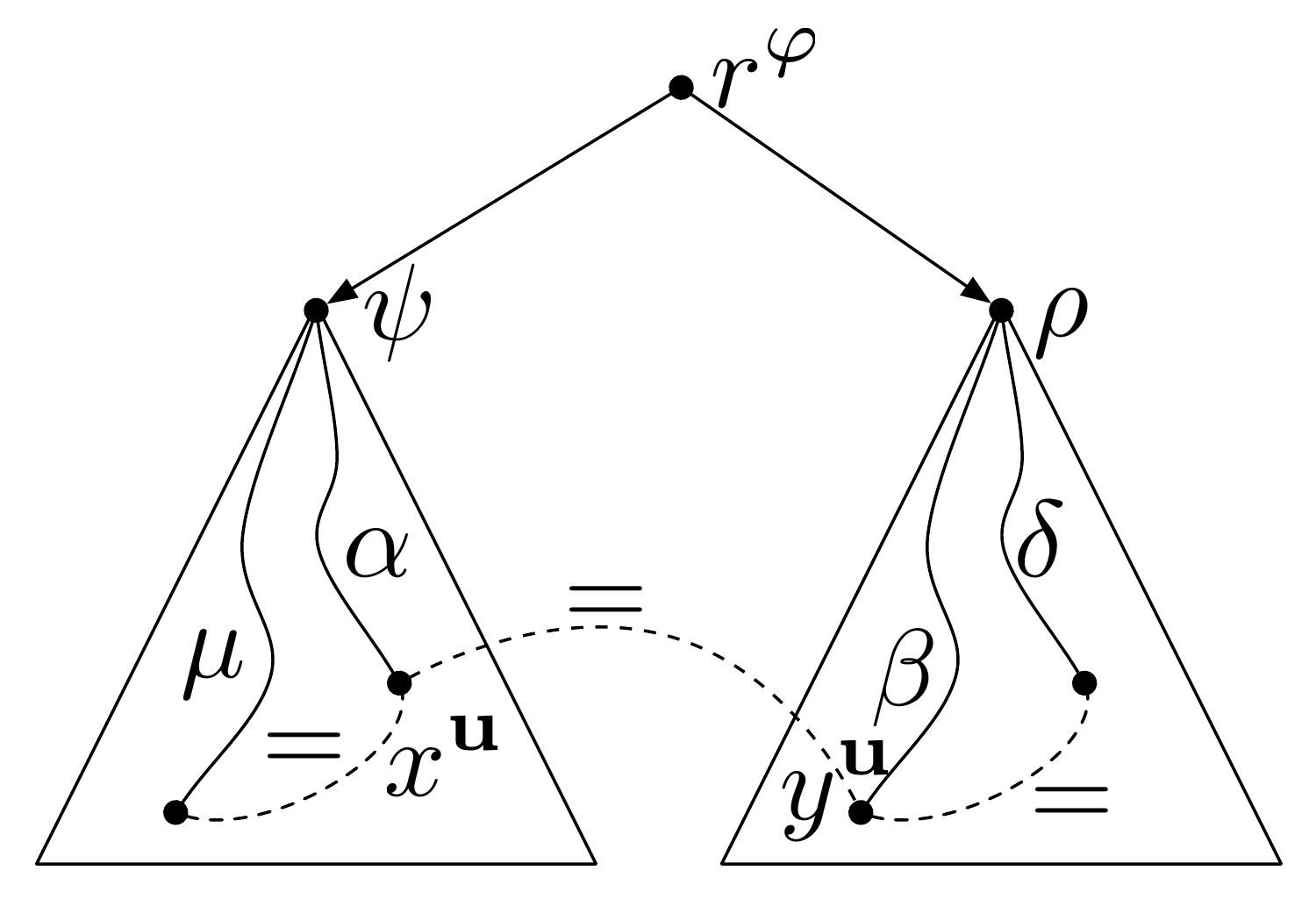}
  }
  \\
  \multicolumn{2}{c}{(c)}
   \end{tabular}
   \end{center}
   \caption{(a) A witness for $\tup{\dow[\psi]\alpha=\dow[\rho]\beta}$; (b) we repeat the process of (a) for each conjunct $\tup{\dow[\psi_1]\alpha_1=\dow[\rho_1]\beta_1},\dots,\tup{\dow[\psi_k]\alpha_k=\dow[\rho_k]\beta_k}$ of $\varphi$; (c) by adding a witness for $\tup{\dow[\psi]\alpha=\dow[\rho]\beta}$, we may be creating an unwanted witness for $\tup{\dow[\psi]\mu=\dow[\rho]\delta}$.}\label{fig:idea2}
\end{figure}
%
%
Note that all the classes in $\Tt^{\psi}$ and $\Tt^\rho$ remain disjoint from the previous constructed part of $\Tt^{\varphi}$. In this way, we guarantee condition~\ref{psialphaeqrhobeta} (see Figure~\ref{fig:idea2}(b)).
Since we are not adding any nodes to the class of $r^\varphi$, it is clear that we cannot spoil condition~\ref{epsiloneqpsialpha} by performing this procedure. With a similar argument as the one given before, the only way in which we can spoil condition~\ref{nopsialphaeqrhobeta} is that there are $\psi, \rho \in \nfN{n}$, $\alpha,\beta, \mu, \delta \in \nfP{n}$ such that $\tup{\dow[\psi]\alpha =\dow[\rho]\beta}$ and $\lnot \tup{\dow [\psi] \mu= \dow [\rho] \delta}$ are conjuncts of $\varphi$, a pair of nodes satisfying $\mu$ beginning at the root of  $\Tt^{\psi}$ and ending in a point with the same data value as $x^\bu$, and a pair of nodes satisfying $\delta $ beginning at the root of $\Tt^{\rho}$ and ending in a point with the same data value as $y^\bu$ (see Figure~\ref{fig:idea2}(c)). But Lemma~\ref{Lemma:ClavePlusMinus} ensures that we can choose $x^\bu$ and $y^\bu$ to 
avoid this situation.

\subsubsection*{Formalization}

In order to formalize the construction described above, we introduce the following key lemma:


\begin{lemma}\label{Lemma:ClavePlusMinus} 
 Let $\psi_0\in \nfN{n}$, $\alpha, \beta_1,\dots,\beta_m\in \nfP{n}$. Suppose that there exists a tree $\Tt^{\psi_0}=(T^{\psi_0},\pi^{\psi_0})$ with root $r^{\psi_0}$ such that $\Tt^{\psi_0}, r^{\psi_0}\models \psi_0$ and for all $i=1,\dots,m $ there exists $\gamma_i\in \nfP{n+1}$ such that $\tup{\gamma_i=\dow[\psi_0]\alpha} \land \lnot \tup{\gamma_i=\dow[\psi_0]\beta_i}$ is consistent. Then there exists a tree $\widetilde{\Tt^{\psi_0}}=(\widetilde{T^{\psi_0}},\widetilde{\pi^{\psi_0}})$ with root $\widetilde{r^{\psi_0}}$ and a node $x$ such that:
\begin{itemize} 
 \item $\widetilde{\Tt^{\psi_0}}, \widetilde{r^{\psi_0}}\models \psi_0$, 
 \item $\widetilde{\Tt^{\psi_0}}, \widetilde{r^{\psi_0}}, x\models \alpha$, and 
 \item $[x]_{\widetilde{\pi^{\psi_0}}}\neq [y]_{\widetilde{\pi^{\psi_0}}}$ for all $y$ such that $\widetilde{\Tt^{\psi_0}}, \widetilde{r^{\psi_0}}, y\models \beta_i$ for some $i=1,\dots,m$.
\end{itemize}


\end{lemma}
\begin{proof}

Suppose that $\alpha=\dow[\psi_1]\dots\dow[\psi_{j_0}]\eps$, where $\psi_k \in \nfN{n-k} $ for all $k=1,\dots,j_0$.
If $j_0 = 0$ (that is, $\alpha = \eps$), then it suffices to take $\widetilde{\Tt^{\psi_0}} = \Tt^{\psi_0}$ and $x = r^{\psi_0}$.  
We only need to show that then $\lnot \tup{ \eps = \beta_i }$ is a conjunct of $\psi_0$ for all $i$.
Indeed, assuming instead that $\tup{\eps = \beta_i }$  is a conjunct of $\psi_0$ for some $i$, we have
  \begin{align*}
  \tup{\gamma_i=\dow[\psi_0]\alpha} \land \lnot \tup{\gamma_i=\dow[\psi_0]\beta_i} & \equiv\tup{\gamma_i=\dow[\psi_0\land \tup{ \eps= \beta_i }]} \land \lnot \tup{\gamma_i=\dow[\psi_0]\beta_i} & \tag{Hypothesis $\tup{\eps = \beta_i }$  is a conjunct of $\psi_0$} \\
  & \leq  \tup{\gamma_i=\dow[\psi_0]\beta_i} \land \lnot \tup{\gamma_i=\dow[\psi_0]\beta_i}\tag{{\bf Der21} (Fact~\ref{fact boolean}) \& \eqax{8}} \\
  & \equiv \botNode \tag{Boolean}
  \end{align*}
%
%
%
which is a contradiction with our assumption that $\tup{\gamma_i=\dow[\psi_0]\alpha} \land \lnot \tup{\gamma_i=\dow[\psi_0]\beta_i}$ is consistent, by standard propositional reasoning.

If $j_0 >0$, to define $\widetilde{\Tt^{\psi_0}}=(\widetilde{T^{\psi_0}},\widetilde{\pi^{\psi_0}})$ we modify the tree $\Tt^{\psi_0}=(T^{\psi_0},\pi^{\psi_0})$.
From the consistency of $\tup{\gamma_i=\dow[\psi_0]\alpha}$ for some $i$, by Lemma~\ref{nextPathIsConjunct}, we conclude that $\tup{\alpha=\alpha}$ is a conjunct of $\psi_0$. Hence there is $z\in T^{\psi_0}$, $z \neq r^{\psi_0}$, such that $\Tt^{\psi_0},r^{\psi_0},z\models\alpha$. 

Before proceeding to complete the proof of this case, we give a sketch of it.
We prove that we cannot have a witness for $\beta_i$  with the same data value as $z$ in the subtree $T^{\psi_0} \restr{z}$. Intuitively this is because, in that case, $\alpha$ would be a prefix of $\beta_i$, say $\beta_i=\alpha\delta$, and $\tup{\eps =\delta}$ would be a conjunct of $\psi_{j_0}$. Then $\tup{\gamma_i=\dow[\psi_0]\alpha} \land \lnot \tup{\gamma_i=\dow[\psi_0]\beta_i}$ would be unsatisfiable (and thus it should be inconsistent)  for any choice of $\gamma_i$ which is a contradiction. But our hypotheses are not enough to avoid having a witness for $\beta_i$ in the class of $z$ outside $T^{\psi_0} \restr{z}$; thus we need to change the tree in order to achieve the desired properties. We replicate the subtree $T^{\psi_0} \restr{z}$ but using fresh data values (different from any other data value already present in $\Tt^{\psi_0}$), see Figure \ref{fig:dibu-one}. It is clear that in this way, the second and the third conditions will be satisfied by the root of this new subtree. 
The 
first condition will also remain true because the positive conjuncts will 
remain valid since we are not suppressing any nodes, and the negative ones will not be affected either because every node we add has a fresh data value. 

Now we formalize the previous intuition.
Call $p$ the parent node of the aforementioned $z \in T^{\psi_0}$ and define $\widetilde{T^{\psi_0}}$ as $T^{\psi_0} \sqcup T^x$, where we define $T^x$ as $T^{\psi_0} \restr{z}$, and in  $\widetilde{T^{\psi_0}}$ the root of $T^x$ is a child $x$ of $p$. Define $\widetilde{\pi^{\psi_0}}$ as $\pi^{\psi_0} \sqcup \pi^{z}$; observe that the data values of $T^x$ differ from all those of the rest of $\widetilde{T^{\psi_0}}$ (see Figure~\ref{fig:dibu-one}). 

\begin{figure}[h!]
   \begin{center}
   \includegraphics[scale=0.25]{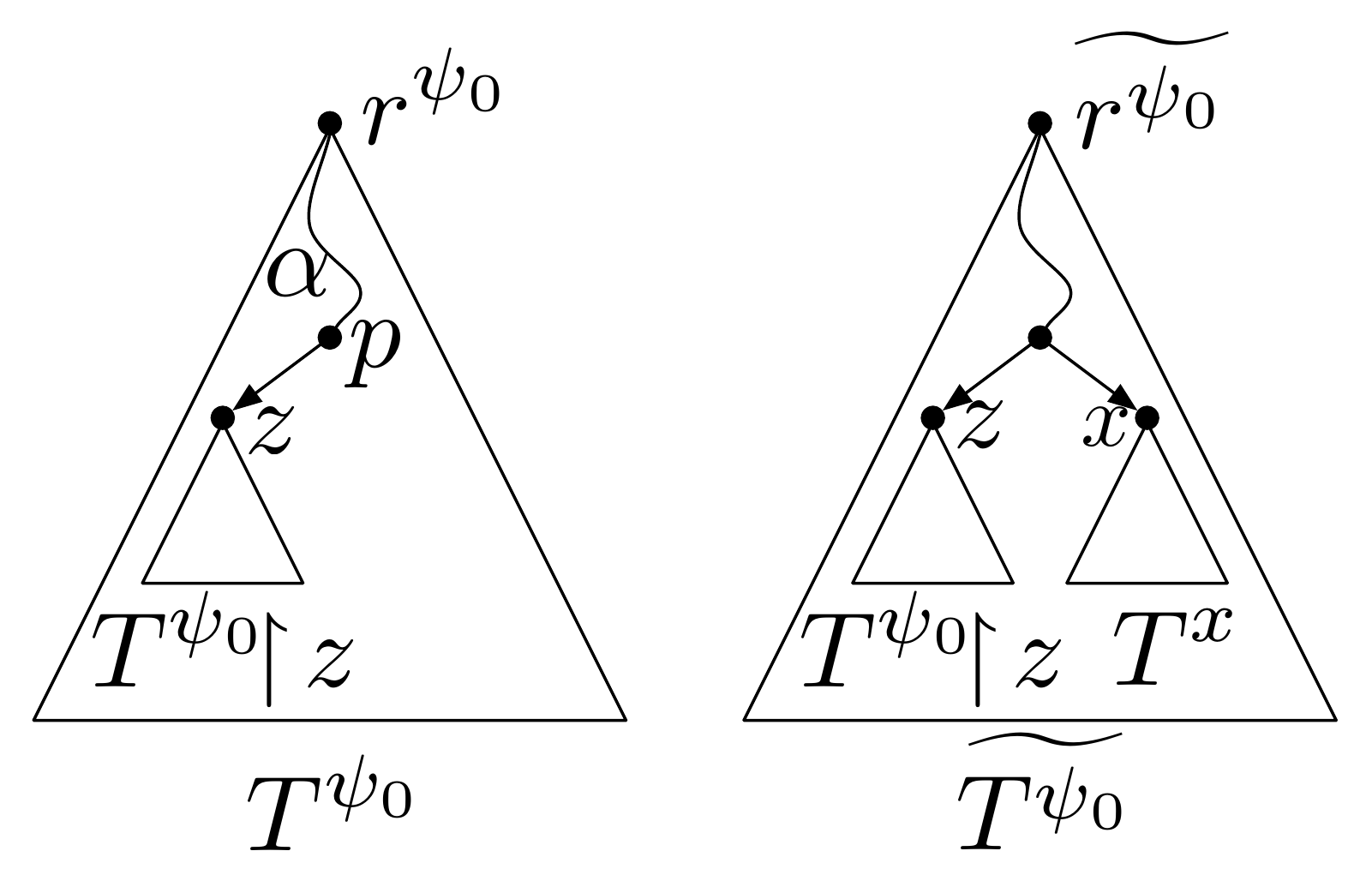}
   \end{center}
   \caption{$T^x = T^{\psi_0} \restr{z}$ is a new subtree with disjoint data values to the rest of $\widetilde{\Tt^{\psi_0}}$. The new node $x$ satisfies $\widetilde{\Tt^{\psi_0}}, \widetilde{r^{\psi_0}}, x \models \alpha$.} \label{fig:dibu-one}
\end{figure}

We now check that this new tree $\widetilde{T^{\psi_0}}$ satisfies $\psi_0$ at its root $\widetilde{r^{\psi_0}}$.
We prove by induction that $x_j$, the $j$-th ancestor of $x$ (namely $x_j\childp{j}x$, and we let $x_0:=x$), satisfies $\widetilde{\Tt^{\psi_0}}, x_j \models \psi_{j_0-j}$. This proves both that $\widetilde{\Tt^{\psi_0}}, \widetilde{r^{\psi_0}}\models \psi_0$ and that $\widetilde{\Tt^{\psi_0}}, \widetilde{r^{\psi_0}}, x\models \alpha$.
For the base case $j=0$, the result is straightforward from Proposition~\ref{prop:local}: $T^x$  is a copy of $T^{\psi_0} \restr{z}$, with $z$ satisfying $\psi_{j_0}$.
For the inductive case, assume that the result holds for $x_0, \dots, x_j$. We want to see that it holds for $x_{j+1}$. To do this, we verify that every conjunct of $\psi_{j_0-j-1}$ is satisfied at $x_{j+1}$:

\begin{itemize}
 \item If the conjunct is a label, it is clear that $x_{j+1}$ still has that label in $\widetilde{\Tt^{\psi_0}}$, as it has not been changed by the construction.
 
 \item  If the conjunct is of the form $\tup{\mu_1 = \mu_2}$, then it must still hold in $\widetilde{\Tt^{\psi_0}}$ by inductive hypothesis plus the fact that our construction did not remove nodes.
 
 \item If the conjunct is of the form $\lnot \tup{\mu_1 = \mu_2}$, 
we observe that, by inductive hypothesis plus the fact that the data classes of nodes in $T^x$ are disjoint with those of the rest of $\widetilde{\Tt^{\psi_0}}$, then $\tup{\mu_1 = \mu_2}$ can only be true in $x_{j+1}$ if there are witnesses $y_1, y_2$ in the same equivalence class in the new subtree $T^x$ such that $\widetilde{\Tt^{\psi_0}}, x_{j+1}, y_1 \models \mu_1$ and $\widetilde{\Tt^{\psi_0}}, x_{j+1}, y_2 \models \mu_2$. 
In that case, we have that 
$$
\mu_1 = \dow{[\psi_{j_0-j}] \dow \dots \dow [\psi_{j_0}] \hat{\mu_1}}\mbox{\quad and \quad}\mu_2 = \dow{[\psi_{j_0-j}] \dow \dots \dow [\psi_{j_0}] \hat{\mu_2}}
$$
 for some $\hat{\mu_1}, \hat{\mu_2}$, and that $\widetilde{\Tt^{\psi_0}}, x_0, y_1 \models \hat{\mu_1}$, $\widetilde{\Tt^{\psi_0}}, x_0, y_2 \models \hat{\mu_2}$. Therefore, by Lemma~\ref{lem:eq_in_psi}, $\tup{\hat{\mu_1} = \hat{\mu_2}}$ is a conjunct of $\psi_{j_0}$, and then $\Tt^{\psi_0}, z \models \tup{\hat{\mu_1} = \hat{\mu_2}}$, a contradiction with our assumption that $\lnot \tup{\dow{[\psi_{j_0-j}] \dow \dots \dow [\psi_{j_0}] \hat{\mu_1}} = \dow{[\psi_{j_0-j}] \dow \dots \dow [\psi_{j_0}] \hat{\mu_2}}}$ is a conjunct of $\psi_{j_0-(j+1)}$ and   $\Tt^{\psi_0}, x_{j+1} \models \psi_{j_0-(j+1)}$.
\end{itemize}

To conclude the proof, we only need to check that $[x]_{\widetilde{\pi^{\psi_0}}}\neq [y]_{\widetilde{\pi^{\psi_0}}}$ for all $y$ such that   $\widetilde{\Tt^{\psi_0}}, \widetilde{r^{\psi_0}}, y\models \beta_i$ for some $i=1,\dots,m$. 
Suppose that $\beta_i=\dow[\rho_1]\dots\dow[\rho_{l_0}]\eps$. If $l_0< j_0$ or $\rho_l\neq \psi_l$ for some $l=1,\dots,j_0$, then the result follows immediately from the construction. 
Otherwise, $l_0 \ge j_0$ and $\rho_l = \psi_l$ for all $l=1,\dots,j_0$ and so, by hypothesis, there exists $\gamma_i \in \nfP{n+1}$ such that 
$$
\tup{\gamma_i = \dow[\psi_0] \dow [\psi_1]\dow\dots\dow[\psi_{j_0}]\eps} \land \lnot \tup{\gamma_i = \dow[\psi_0]\dots\dow[\psi_{j_0}]\dow[\rho_{j_0+1}]\dots\dow[\rho_{l_0}]\eps}
$$
is consistent. We  prove that $\lnot \tup{\eps = \dow[\rho_{j_0+1}]\dots\dow [\rho_{l_0}] \eps}$ is a conjunct of $\psi_{j_0}$, from which our desired property follows immediately since we have proved that $\widetilde{\Tt^{\psi_0}}, x \models \psi_{j_0}$. 
Aiming for a contradiction, suppose instead that $\tup{\eps = \dow[\rho_{j_0+1}]\dots\dow [\rho_{l_0}] \eps}$ is a conjunct of $\psi_{j_0}$. 
Then, as $\alpha=\dow[\psi_1]\dots\dow[\psi_{j_0}]\eps$, we can derive that $\equivInstance{\alpha}{\alpha[\tup{\eps = \dow[\rho_{j_0+1}]\dots\dow [\rho_{l_0}] \eps}]}$ ({\bf Der21} (Fact~\ref{fact boolean})). Also observe that $\axiomRestr\vdash \equivInstance{\beta_i}{\alpha \dow[\rho_{j_0+1}]\dots\dow [\rho_{l_0}]\eps}$, and then 
we have
%
  \begin{align*}
  \tup{\gamma_i=\dow[\psi_0]\alpha}& \equiv \tup{\gamma_i=\dow[\psi_0] \alpha[\tup{\eps = \dow[\rho_{j_0+1}]\dots\dow [\rho_{l_0}] \eps}]} \tag{{\bf Der21} (Fact~\ref{fact boolean})} \\
  & \leq \tup{\gamma_i = \dow[\psi_0] \alpha \dow[\rho_{j_0+1}]\dots\dow [\rho_{l_0}]\eps}\tag{\eqax{8}} \\
  & \equiv \tup{\gamma_i = \dow[\psi_0] \beta_i}
  \end{align*}
%
But using simple propositional reasoning, we have a contradiction with our hypothesis that $\tup{\gamma_i = \dow[\psi_0] \alpha} \land \lnot \tup{\gamma_i = \dow[\psi_0] \beta_i }$ was consistent, a contradiction that came from assuming that $\tup{\eps = \dow[\rho_{j_0+1}]\dots\dow [\rho_{l_0}] \eps}$ was a conjunct of $\psi_{j_0}$.
\end{proof}

Now that we have proved this lemma, we proceed to the formal construction of $\Tt^{\varphi}$, for $\varphi\in \nfN{n+1}$ (recall the base case at the beginning of \S\ref{construccion}).

Consider the following sets:
\begin{align*}
{\bf V}&=\left\{(\psi,\alpha) \mid \tup{\eps=\dow[\psi]\alpha} \hbox{ is a conjunct of } \varphi \right\}
\\
{\bf U}&=\left\{ (\psi, \alpha, \rho, \beta) \mid \tup{\dow[\psi]\alpha=\dow[\rho]\beta} \hbox{ is a conjunct of }\varphi \right\}
\end{align*}



\paragraph {\bf Rule 1. Witnesses for \boldmath{${\bf v}=(\psi,\alpha)\in V$}.} We define a data tree  $\Tt^{\bf v}=(T^{\bf v},\pi^{\bf v})$ with root $r^{\bf v}$.
By inductive hypothesis, there exists a tree $\Tt^{\psi}$ such that $\psi$ is satisfiable in that tree.  
In Lemma~\ref{Lemma:ClavePlusMinus}, consider 
\begin{align*}
\psi_0&:=\psi
\\
\Tt^{\psi_0}&:=\Tt^{\psi}
\\
\alpha&:=\alpha
\\
\{\beta_1,\dots,\beta_m\}&:=\{\beta \in \nfP{n} \mid \lnot \tup{\epsilon=\dow[\psi]\beta} \hbox{ is a conjunct of } \varphi\}
\\
\gamma_i&:=\epsilon \mbox{ for all $i=1,\dots,m$}
\end{align*}

\noindent Then there exists $\widetilde{\Tt^{\psi}} =(\widetilde{T^\psi}, \widetilde{\pi^\psi})$ with root $\widetilde{r^\psi}$ and a node $x$ such that 
 \begin{itemize} 
\item $\widetilde{\Tt^{\psi}}, \widetilde{r^\psi} \models \psi$, 
\item $\widetilde{\Tt^{\psi}}, \widetilde{r^\psi}, x\models \alpha$, 
\item $[x]_{\widetilde{\pi^\psi}} \neq [y]_{\widetilde{\pi^\psi}}$ for all $y$ such that there is $\beta \in \nfP{n}$ with $(\psi, \beta) \not\in {\bf V}$, and $\widetilde{\Tt^\psi},\widetilde{r^\psi}, y\models \beta$

\end{itemize}
Define $\Tt^{\bf v}$ as $\widetilde{\Tt^{\psi}} $, and $x^{\bf v}$ as $x$. The root $r^{\bf v}$ and the partition $\pi^{\bf v}$ are the ones of $\widetilde{\+T^\psi}$.

\paragraph {\bf Rule 2. Witnesses for \boldmath{${\bf u}=(\psi,\alpha,\rho,\beta)\in U$}.}
We define data trees $\Tt^{\bf u}_1=(T^{\bf u}_1,\pi^{\bf u}_1)$ and $\Tt^{\bf u}_2=(T^{\bf u}_2,\pi^{\bf u}_2)$ with roots $r^{\bf u}_1$, $r^{\bf u}_2$ respectively. 
By inductive hypothesis, there exist trees $\Tt^\psi=(T^\psi,\pi^\psi)$ (with root $r^\psi$) and $\Tt^\rho=(T^\rho,\pi^\rho)$ (with root $r^\rho$) such that $\psi$ is satisfiable in $\Tt^\psi$ and $\rho$ is satisfiable in $\Tt^\rho$. 
In Lemma~\ref{Lemma:ClavePlusMinus}, consider 
\begin{align*}
\psi_0&:=\psi
\\
\Tt^{\psi_0}&:=\Tt^{\psi}
\\
\alpha&:=\alpha
\\
\{\beta_1,\dots,\beta_m\}&:=\{\gamma \in \nfP{n} \mid \lnot \tup{\dow[\rho]\beta=\dow[\psi]\gamma} \hbox{ is a conjunct of } \varphi\}
\\
\gamma_i&:=\dow[\rho]\beta \mbox{ for all $i=1,\dots,m$} 
\end{align*}
Then there exist $\widetilde{\Tt^\psi}=(\widetilde{T^\psi},\widetilde{\pi^\psi})$ with root $\widetilde{r^\psi}$ and a node $x$ such that:
\begin{itemize}
\item $\widetilde{\Tt^\psi},\widetilde{r^\psi}\models \psi$, 
\item $\widetilde{\Tt^\psi},\widetilde{r^\psi}, x\models \alpha$ 
\item $[x]_{\widetilde{\pi^\psi}}\neq [y]_{\widetilde{\pi^\psi}}$ for all $y$ such that there is $\gamma\in \nfP{n}$ with $\widetilde{\Tt^\psi},\widetilde{r^\psi}, y\models \gamma$  and $\lnot \tup{\dow[\rho]\beta = \dow[\psi]\gamma }$ is a conjunct of $\varphi$. 
\end{itemize}
Define $T^{\bf u}_1$ as $\widetilde{T^\psi}$, $\pi^{\bf u}_1$ as $\widetilde{\pi^\psi}$, $r^{\bf u}_1$ as $\widetilde{r^\psi}$ and $x^{\bf u}=x \in T^{\bf u}_1$. Now let  
$$
\left\{\mu_1,\dots,\mu_r \right\}=\left\{\mu \in \nfP{n} \mid \mbox{ there exists } \ y\in T^{\bf u}_1 \mbox{ such that } \Tt^{\bf u}_1,r^{\bf u}_1, y\models \mu \mbox{ and } [y]_{\pi^{\bf u}_1}=[x^{\bf u}]_{\pi^{\bf u}_1}\right\}.
$$
Then it follows that $\tup{\dow[\rho]\beta=\dow[\psi]\mu_j}$ is a conjunct of $\varphi$ for all $j=1,\dots,r$.  
In Lemma~\ref{Lemma:ClavePlusMinus}, consider
\begin{align*}
\psi_0&:=\rho
\\
\Tt^{\psi_0}&:=\Tt^{\rho}
\\
\alpha&:=\beta
\\
\{\beta_1,\dots,\beta_m\}&:=\{\delta \in \nfP{n} \mid \exists j=1,\dots,r \hbox{ with } \lnot \tup{\dow[\rho]\delta=\dow[\psi]\mu_j} \hbox{ is a conjunct of } \varphi\}
\\
\gamma_i&:=\dow[\psi]\mu_j \mbox{ for $j=1,\dots r$ such that $\tup{\dow[\rho]\beta_i=\dow[\psi]\mu_j}$ is a conjunct of $\varphi$}. 
\end{align*}
Then there exist a tree $\widetilde{\Tt^{\rho}}=(\widetilde{T^{\rho}},\widetilde{\pi^{ \rho}})$ with root $\widetilde{r^{\rho}}$ and a node $y$ such that
\begin{itemize}
\item $\widetilde{\Tt^{\rho}}, \widetilde{r^{\rho}}\models \rho$,
\item $\widetilde{\Tt^{\rho}}, \widetilde{r^{\rho}}, y\models \beta$, 
\item $[y]_{\widetilde{\pi^{\rho}}}\neq [z]_{\widetilde{\pi^{\rho}}}$ for all $z$ such that there is $\delta\in \nfP{n}$ and $j=1,\dots,r$ with $\widetilde{\Tt^\rho},\widetilde{r^\rho}, z\models \delta$ and $\lnot \tup{\dow[\rho]\delta=\dow[\psi]\mu_j}$ is a conjunct of $\varphi$.
\end{itemize}
Define $T^{\bf u}_2$ as $\widetilde{T^\rho}$, $\pi^{\bf u}_2$ as $\widetilde{\pi^\rho}$, $r^{\bf u}_2$ as $\widetilde{r^\rho}$ and $y^{\bf u}=y$. Without loss of generality, we assume that $T^{\bf u}_1$ and $T^{\bf u}_2$ are disjoint.

Now we define a partition $\pi^{\bf u}$ over $T^{\bf u}_1\cup T^{\bf u}_2$ as 
$$
\pi^{\bf u}=\pi^{\bf u}_1\cup\pi^{\bf u}_2\cup\{[x^{\bf u}]_{\pi_1^{\bf u}}\cup[y^{\bf u}]_{\pi_2^{\bf u}}\}\setminus\{[x^{\bf u}]_{\pi_1^{\bf u}},[y^{\bf u}]_{\pi_2^{\bf u}}\}.
$$
In other words, the rooted data tree $(T_1^{\bf u},\pi^{\bf u}\restr{T_1^{\bf u}} ,r_1^{\bf u})$ is just a copy of $(\widetilde{T^\psi},\widetilde{\pi^\psi},\widetilde{r^\psi})$, with a special node named $x^{\bf u}$ which satisfies $T_1^{\bf u},\pi^{\bf u},r_1^{\bf u},x^{\bf u}\models\alpha$. Analogously, the pointed data tree \mbox{$(T_2^{\bf u},\pi^{\bf u}\restr{ T_2^{\bf u}},r^{\bf u}_2)$} is a copy of $(\widetilde{T^\rho},\widetilde{\pi^\rho},\widetilde{r^\rho})$, with a special node named $y^{\bf u}$ which satisfies   $T_2^{\bf u},\pi^{\bf u},r_2^{\bf u},y^{\bf u}\models\beta$.  
Notice that the equivalence class $\sim$ induced by $\pi^{\bf u}$ (defined over the disjoint sets $T_1^{\bf u}$ and $T_2^{\bf u}$) is defined as $z\sim w$ iff $w\in[z]_{\widetilde{\pi^{\psi}}}$ or $w\in[z]_{\widetilde{\pi^{\rho}}}$, or both $w\in[x^{\bf u}]_{\widetilde{\pi^{\psi}}}$ and $z\in[y^{\bf u}]_{\widetilde{\pi^{\rho}}}$ or both $w\in[y^{\bf u}]_{\widetilde{\pi^{\rho}}}$ and $z\in[x^{\bf u}]_{\widetilde{\pi^{\psi}}}$. See Figure~\ref{fig:Tu}.
\begin{figure}[ht]
   \begin{center}
   \includegraphics[scale=0.25]{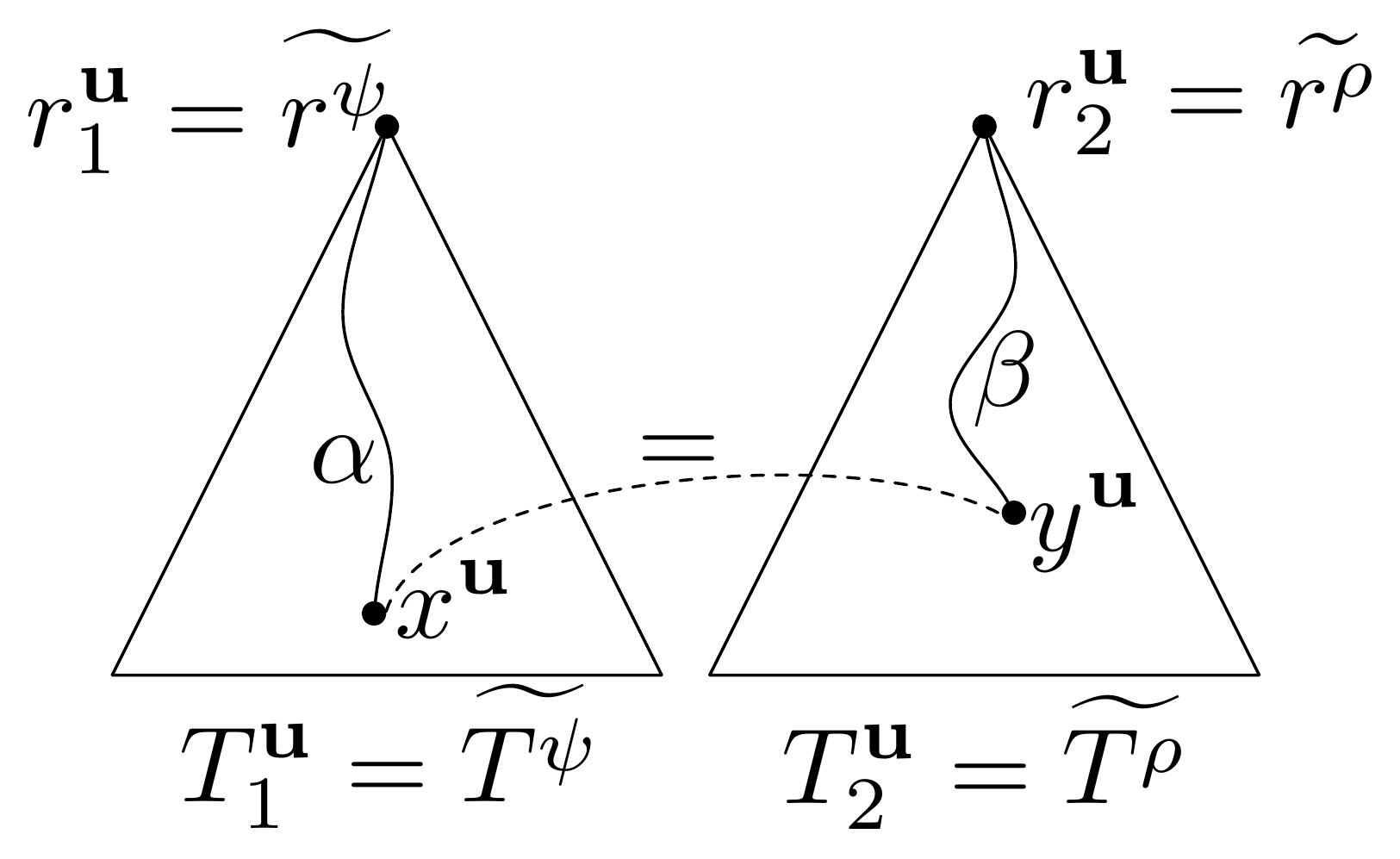}
   \end{center}
   \caption{The data trees $\Tt_1^{\bf u}=(T_1^{\bf u},\pi_1^{\bf u})$ and $\Tt_2^{\bf u}=(T_2^{\bf u},\pi_2^{\bf u})$ for some ${\bf u}\in{\bf U}$. $\pi^{\bf u_1}$ and $\pi^{\bf u_2}$ are disjoint except that the equivalence class of $x^{\bf u}$ is merged with the equivalence class of $y^{\bf u}$.}\label{fig:Tu}
\end{figure}

The following remark will be used later to prove that Rule 2 does not spoil condition~\ref{nopsialphaeqrhobeta} (cf. Figure~\ref{fig:idea2}(c)):

\begin{remark}\label{lemacasosindistinto}
Let $(\psi,\alpha,\rho,\beta)\in {\bf U}$. If $\lnot \tup{\dow[\psi]\mu = \dow[\rho]\delta}$ is a conjunct of $\varphi$, then $[y^{\bf u}]_{\pi^{\bf u}_2}\neq [y]_{\pi^{\bf u}_2}$ for all $y$ such that $\Tt_2^{\bf u}, r_2^{\bf u}, y \models \delta$ or $[x^{\bf u}]_{\pi^{\bf u}_1}\neq [x]_{\pi^{\bf u}_1}$ for all $x$ such that $\Tt_1^{\bf u}, r_1^{\bf u}, x \models \mu$.
%
 
 
 
\end{remark}

\begin{proof}
The result is immediate from Rule 2: If neither of the conditions is satisfied, then $\mu=\mu_j$ for some $j=1,\dots,r$ and so $\tup{\dow[\rho]\delta=\dow[\psi]\mu}$ is a conjunct of $\varphi$ which is a contradiction.
\end{proof} 

\paragraph{The rooted data tree \boldmath{$(T^\varphi,\pi^\varphi,r^\varphi)$}.}

Now, using {Rule 1} and {Rule 2}, we define $T^\varphi$ as the tree which consists of a root $r^\varphi$ with label $a\in\A$ if $a$ is a conjunct of $\varphi$, and with children  
$$
(T^{\bf v})_{{\bf v}\in {\bf V}}\ ,\  (T^{\bf u}_1)_{{\bf u}\in {\bf U}}\ ,\ (T^{\bf u}_2)_{{\bf u}\in {\bf U}}.
$$
We assume that the nodes of all such trees are pairwise disjoint.
Define $\pi^\varphi$ over $T^\varphi$ by 
$$
\pi^\varphi= \left(\bigcup_{{\bf v}\in {\bf V}}\pi^{\bf v}\setminus\{[x^{\bf v}]_{\pi^{\bf{v}}}\mid {\bf v}\in {\bf V}\}\right)\cup\left\{\{r^\varphi\}\cup \bigcup_{{\bf v}\in {\bf V}}[x^{\bf v}]_{\pi^{\bf{v}}}\right\}
 \cup \bigcup_{{\bf u}\in {\bf U}} \pi^{\bf u}.
$$
In other words, $T^\varphi$ has a root, named $r^\varphi$, and children $(r^{\bf v})_{{\bf v}\in {\bf V}}$, $(r_1^{\bf u})_{{\bf u}\in {\bf U}}$, $(r_2^{\bf u})_{{\bf u}\in {\bf U}}$. Each of these children is the root of its corresponding tree inside $T^\varphi$ as defined above, i.e.\ for each ${\bf v\in V}$, $r^{\bf v}$ is the root of $T^{\bf v}$, and for each ${\bf u\in U}$, $r_i^{\bf u}$ is the root of $T_i^{\bf u}$ ($i=1,2$). All these subtrees are disjoint, and $\pi^\varphi$ is defined as the disjoint union of partitions $\pi^{\bf v}$ for ${\bf v}\in {\bf V}$, and    
all $\pi^{\bf u}$ for ${\bf u}\in {\bf U}$, {\em with the exception} that we put into the same class the nodes $r^\varphi$ and $(x^{\bf{v}})_{{\bf v}\in{\bf V}}$. See Figure~\ref{fig:Tcan}.
\begin{figure}[ht]
   \begin{center}
   \includegraphics[scale=0.25]{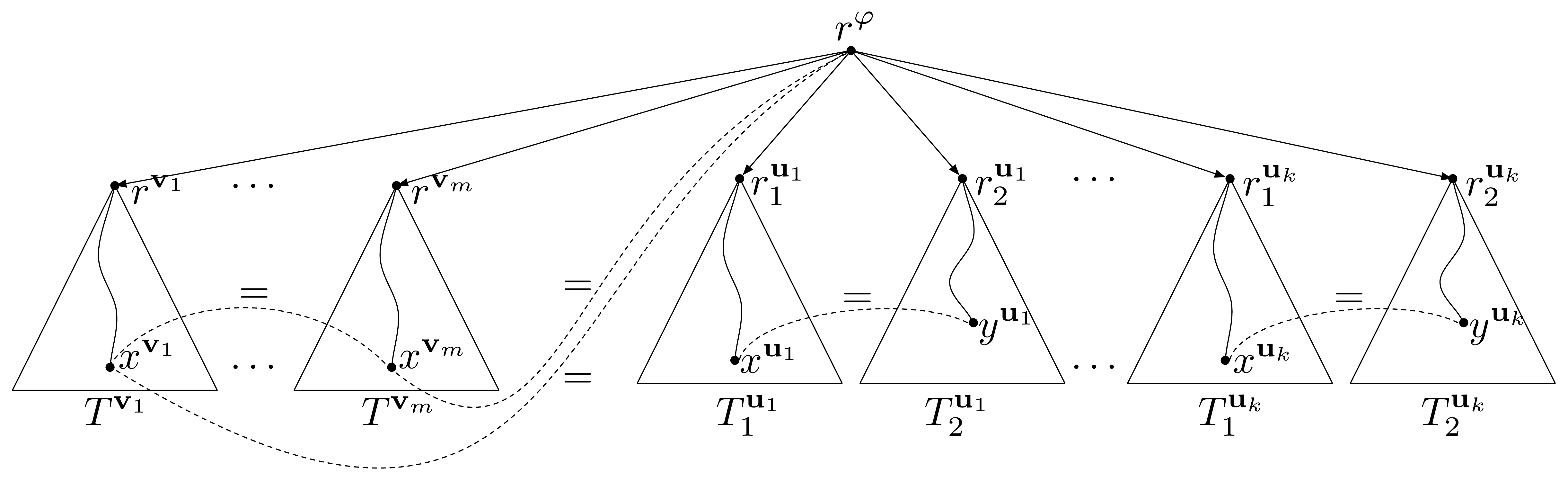}
   \end{center}
   \caption{The data tree $\Tt^\varphi$, with root $r^\varphi$, when ${\bf V}=\{{\bf v}_1,\dots,{\bf v}_m\}$ and ${\bf U}=\{{\bf u}_1,\dots,{\bf u}_k\}$. Nodes $r^\varphi,x^{{\bf v}_1},\dots,x^{{\bf v}_m}$ are in the same equivalence class, and for each $i$ nodes $x^{{\bf u}_i}$ and $y^{{\bf u}_i}$ are in the same equivalence class.}\label{fig:Tcan}
\end{figure}

The following fact follows easily by construction:
\begin{fact}
The partition restricted to the trees $T^{\bf v}$ for ${\bf v\in V}$ and the partition restricted to the trees $T^{\bf u}_1$ and $T^{\bf u}_2$ for ${\bf u\in U}$ remains unchanged. More formally:
\begin{itemize}
\item For each ${\bf v}=(\psi,\alpha)\in {\bf V}$, we have $\pi^\varphi\restr{T^{\bf v}}=\pi^{\bf v}$.

\item For each ${\bf u}=(\psi,\alpha,\rho,\beta)\in {\bf U}$, we have $\pi^\varphi\restr{T^{\bf u}_1}=\pi^{\bf u}_1$, and $\pi^\varphi\restr{T^{\bf u}_2}=\pi^{\bf u}_2$.
\end{itemize}
\end{fact}
%
We conclude from Proposition~\ref{prop:local} and the construction that:
\begin{fact}\label{fac:preserves}
The validity of a formula in a child of $r^\varphi$ is preserved in $\Tt^\varphi$. More formally:
\begin{itemize}
%
%
%
%
\item For each ${\bf v}\in{\bf V}$ and $x,y\in T^{\bf v}$ we have $\Tt^\varphi,x \eqres  \Tt^{\bf v},x$ and $\Tt^\varphi,x,y \eqres   \Tt^{\bf v},x,y$.

\item For each ${\bf u}\in {\bf U}$, $i\in\{1,2\}$ and $x,y\in T^{\bf u}_i$
we have $\Tt^\varphi,x \eqres \Tt^{\bf u}_i,x$ and   $\Tt^\varphi,x,y \eqres \Tt^{\bf u}_i,x,y$.
\end{itemize}
\end{fact}

%% file: verif.tex
It only remains to check that conditions \ref{label}--\ref{nopsialphaeqrhobeta} at the beginning of \S\ref{subsec:completeness} are satisfied: 
%

\paragraph{Verification of \ref{label}.}This condition is trivially satisfied.

\paragraph{Verification of \ref{epsiloneqpsialpha}.}
Suppose $\tup{\epsilon=\dow[\psi]\alpha}$ is a conjunct of $\varphi$.
Then, by Rule 1, there is $x^\bv\in T^\varphi$ such that $[r^\varphi]_{\pi^\varphi}=[x^\bv]_{\pi^\varphi}$, with $\bv=(\psi,\alpha)$.
We also know by construction that $\Tt^\bv,r^\bv\models\psi$ and $\Tt^\bv,r^\bv,x^\bv\models\alpha$.
By Fact \ref{fac:preserves} we conclude $\Tt^\varphi,r^\varphi\models\tup{\epsilon=\dow[\psi]\alpha}$.

\paragraph{Verification of \ref{psialphaeqrhobeta}.}
Suppose $\tup{\dow[\psi]\alpha=\dow[\rho]\beta}$ is a conjunct of $\varphi$.
Then, by Rule 2, there are $x^\bu,y^\bu \in \Tt^{\varphi}$ such that $[x^\bu]_{\pi^\varphi}=[y^\bu]_{\pi^\varphi}$,
with $\bu=(\psi,\alpha,\rho,\beta)$. We also know on the one hand that 
$\Tt_1^\bu,r_1^\bu\models\psi$ and $\Tt_2^\bu,r_2^\bu\models\rho$,
and on the other hand that $\Tt_1^\bu,r_1^\bu,x^\bu\models\alpha$ and 
$\Tt_2^\bu,r_2^\bu,y^\bu\models\beta$. 
By Fact \ref{fac:preserves} we conclude $\Tt^\varphi,r^\varphi\models\tup{\dow[\psi]\alpha=\dow[\rho]\beta}$.

\paragraph{Verification of \ref{noepsiloneqpsialpha}.}
Suppose $\neg\tup{\epsilon=\dow[\psi]\alpha}$ is a conjunct of $\varphi$. Aiming for a contradiction, suppose that  $\Tt^\varphi,r^\varphi\models\tup{\epsilon=\dow[\psi]\alpha}$. 
Then there is a successor $z$ of $r^\varphi$ in which $\psi$ holds, and by construction plus Lemma~\ref{lemma:inconsistent}, $z$ is the root of some copy of a data tree $\widetilde{\Tt^\psi}$. Moreover, there is $x\in \widetilde{T^\psi}$ such that $\Tt^\varphi,z,x\models\alpha$, with $[x]_{\pi^\varphi}=[r^\varphi]_{\pi^\varphi}$. In addition to this, $(\psi, \alpha)\not \in {\bf V}$ and so, by Rule 1, $[x]_{\pi^{\varphi}}\neq [x^{\bf v}]_{\pi^{\varphi}}$ for all ${\bf v} \in {\bf V}$. Then, by construction, $[x]_{\pi^{\varphi}}\neq [r^{\varphi}]_{\pi^{\varphi}}$ which is a contradiction.

\paragraph{Verification of \ref{nopsialphaeqrhobeta}.}
Suppose $\neg\tup{\dow[\psi]\alpha=\dow[\rho]\beta}$ is a conjunct of $\varphi$. Aiming for a contradiction, suppose that $\Tt^\varphi,r^\varphi\models\tup{\dow[\psi]\alpha=\dow[\rho]\beta}$. 
Then there are successors $z_1$ and $z_2$  of $r^\varphi$ in which $\psi$ and $\rho$ holds, respectively. Also, by
construction and Lemma~\ref{lemma:inconsistent}, $z_1$ and $z_2$
 are the roots of some copies of data trees $\widetilde{\Tt^\psi}$ and $\widetilde{\Tt^\rho}$ (note that we are using the notation $\widetilde{\Tt^{\psi}}$ and $\widetilde{\Tt^{\rho}}$ either if the tree is the one obtained by inductive hypothesis or a modified version of it).
 Moreover, there are descendants $w_1$ and $w_2$ such  that $\Tt^\varphi, z_1,w_1 \models \alpha$,  $\Tt^\varphi, z_2,w_2 \models \beta$ and $[w_1]_{\pi^\varphi}=[w_2]_{\pi^\varphi}$.
%
We now have two cases to analyze: 
\begin{itemize}

\item $\widetilde{\Tt^\psi}=\widetilde{\Tt^\rho}$: In this case, because of Lemma~\ref{lemma:inconsistent}, $\psi=\rho$. And we have $\widetilde{\Tt^\psi},\widetilde{r^\psi}\models\tup{\alpha=\beta}$, and
as a consequence $\tup{\alpha=\beta}$ has to be a conjunct of $\psi$ (Lemma~\ref{lem:eq_in_psi}). We prove that in this case $\tup{\dow[\psi]\alpha'=\dow[\psi]\alpha'}$ can not be a conjunct of $\varphi$ for any $\alpha'\in \nfP{n}$: If this were the case, $\tup{\dow[\psi]\alpha'=\dow[\psi]\alpha'} \land \neg\tup{\dow[\psi]\alpha=\dow[\rho]\beta}$ would be consistent, but:
\begin{align*}&
\hspace{-10pt} \tup{\dow[\psi]\alpha'=\dow[\psi]\alpha'} \land \neg\tup{\dow[\psi]\alpha=\dow[\rho]\beta} \\&
    \leq  \tup{\dow[\psi]} \land \neg\tup{\dow[\psi]\alpha=\dow[\rho]\beta} \tag{{\bf Der12} (Fact~\ref{fact boolean})}  \\&
   \equiv \tup{\dow[\psi\land \tup{\alpha=\beta}]} \land \neg\tup{\dow[\psi]\alpha=\dow[\rho]\beta}  \tag{$\tup{\alpha=\beta}$ is a conjunct of $\psi$} \\&
   \equiv \tup{\dow[\psi][\tup{\alpha=\beta}]} \land \neg\tup{\dow[\psi]\alpha=\dow[\rho]\beta}  \tag{{\bf Der21} (Fact~\ref{fact boolean})} \\&
   \leq \tup{\dow[\psi]\alpha=\dow[\psi]\beta} \land \neg\tup{\dow[\psi]\alpha=\dow[\rho]\beta} \tag{\eqax{7}} \\&
   \equiv \botNode\tag{Boolean}
\end{align*}  
which is a contradiction. 

Then, $\tup{\dow[\psi]\alpha'=\dow[\psi]\alpha'}$ is not a conjunct of $\varphi$ and so, it follows easily from the consistency of $\varphi$ that $(\psi,\alpha')\not \in {\bf V}$ for all $\alpha'\in \nfP{n}$. And also $(\psi,\alpha',\rho',\beta')\not \in {\bf U}$ for all $\alpha',\beta'\in \nfP{n}$, $\rho'\in \nfN{n}$. This gives a contradiction by construction because in this case it would not be a copy of a tree $\widetilde{\Tt^{\psi}}$. 

\item $\widetilde{\Tt^\psi}\neq\widetilde{\Tt^\rho}$: In this case, there are two possibilities to consider:

\begin{itemize}

\item One possibility is that $[w_1]_{\pi^\varphi} = [w_2]_{\pi^\varphi}$ because Rule 2 was applied (see Figure~\ref{fig:verif}~(a)). Then there is $\bu=(\psi,\alpha',\rho,\beta')\in {\bf U}$ (the symmetric case is analogous).
In this case we have $\Tt^{\bf u}_1, r^{\bf u}_1, x^{\bf u} \models \alpha'$ and $\Tt^{\bf u}_2, r^{\bf u}_2, y^{\bf u} \models \beta'$. Furthermore, since $[w_1]_{\pi^\varphi} = [w_2]_{\pi^\varphi}$, we have that $[w_1]_{\pi^{\bf u}_1} = [x^{\bf u}]_{\pi^{\bf u}_1}$ and $[w_2]_{\pi^{\bf u}_2} = [y^{\bf u}]_{\pi^{\bf u}_2}$ which is a contradiction by Remark \ref{lemacasosindistinto}.


%

\item
The other possibility is that $[w_1]_{\pi^\varphi} = [w_2]_{\pi^\varphi}$ because Rule 1 was applied twice (see Figure~\ref{fig:verif}(b)). Then there exist $\bv_1=(\psi, \alpha')$, $\bv_2=(\rho, \beta')\in{\bf V}$.
In this case we have $\Tt^{\bv_1}, r^{\bf v_1}, x^{\bv_1} \models \alpha'$ and $\Tt^{\bv_2}, r^{\bv_2}, x^{\bv_2} \models \beta'$. Furthermore, since $[w_1]_{\pi^\varphi} = [w_2]_{\pi^\varphi}$, we have that $[w_1]_{\pi^{\bf v_1}} = [x^{\bv_1}]_{\pi^{\bf v_1}}$ and $[w_2]_{\pi^{\bf v_2}} = [y^{\bv_2}]_{\pi^{\bf v_2}}$. Then, by Rule 1, $(\psi, \alpha)$ and $(\rho,\beta)$ belong to ${\bf V}$ which gives a contradiction because of the consistency of $\varphi$ plus  \eqax{6}.

\end{itemize}

\end{itemize}

\begin{figure}[ht]
   \begin{center}
   \begin{tabular}{c@{\hskip .5in}c}
   \includegraphics[scale=0.25]{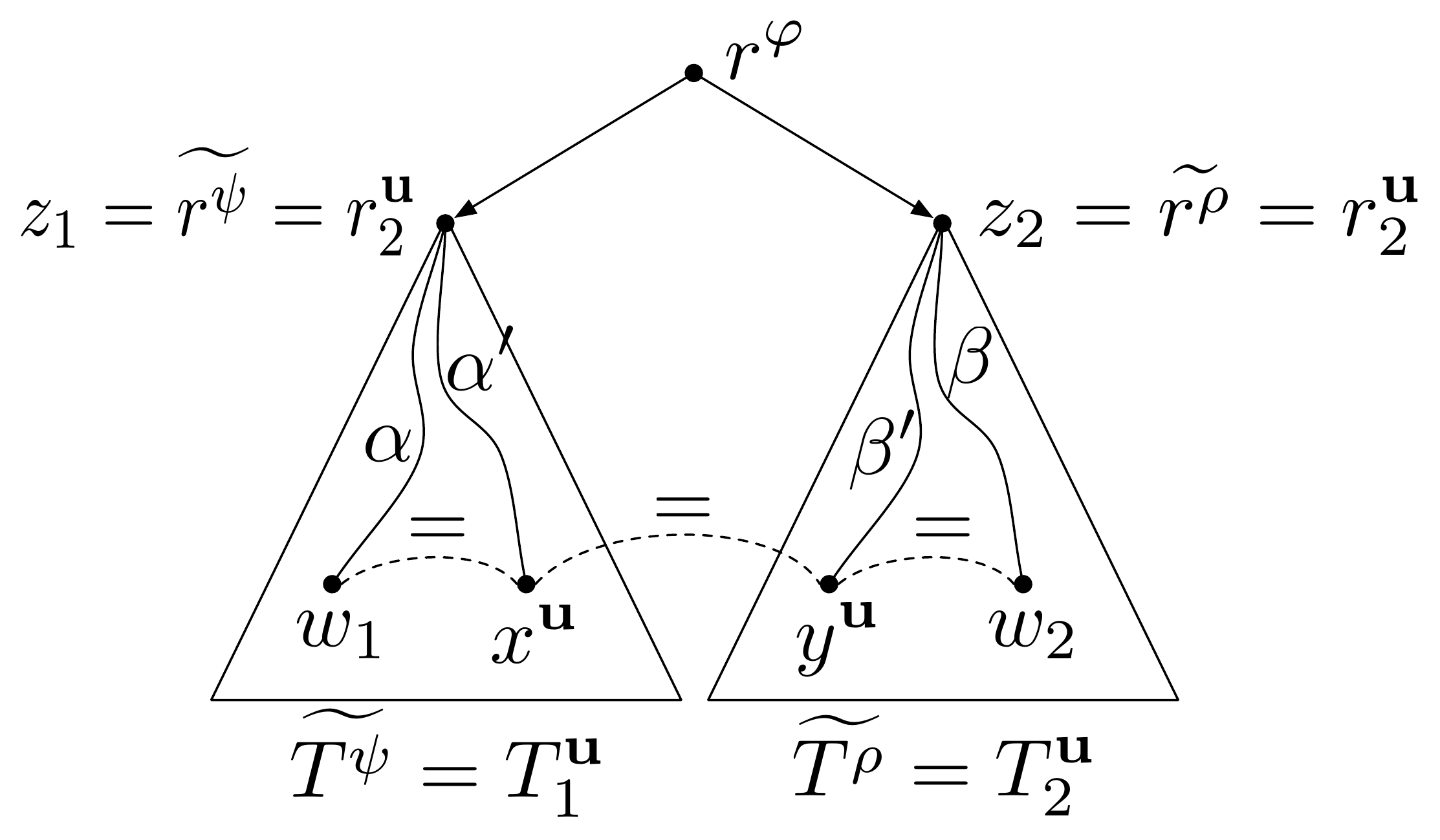}&\includegraphics[scale=0.25]{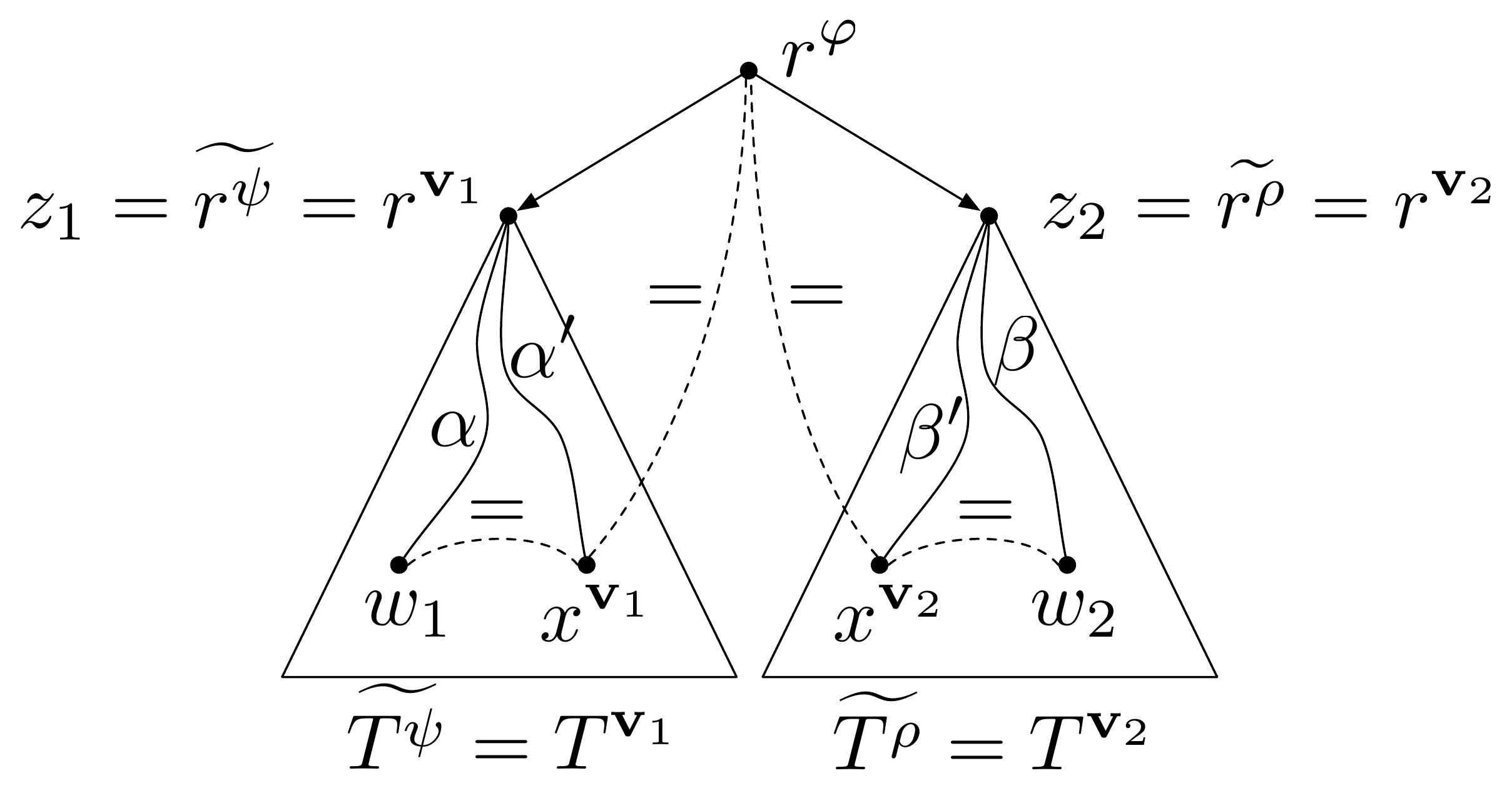}\\
   (a)&(b)
   \end{tabular}
   \end{center}
   \caption{Nodes $w_1$ and $w_2$ are in the same equivalence class because (a) Rule 2 was applied via $\bu=(\psi,\alpha',\rho,\beta')\in{\bf U}$, or (b) Rule 1 was applied twice via $\bv_1=(\psi, \alpha')$, $\bv_2=(\rho, \beta')\in{\bf V}$.
   }   \label{fig:verif}
\end{figure}

%
%
%

%% file: axiom-neq.tex

In this section we introduce additional axiom schemes to handle inequalities. Axioms schemes in Table~\ref{tab:axiomsneq}
extend those from Table~\ref{tab:axiomseq} to form a complete axiomatic system for the full logic $\xpd$. Observe that \neqaxone\ -- \neqaxfifteen are analogous to \eqax{2}\ -- \eqax{7}.

\begin{table}[h!]
\centering
\colorbox{black!10}{
\begin{tabular}{lrcll}
\hline
\multicolumn{5}{l}{\bf Node axiom schemes for inequality \vspace{.05in}}
\\
\hypertarget{neqaxone}{\neqaxonen} &$\tup{\alpha \neq \beta}$ & $\equiv$ & $\tup{\beta \neq \alpha}$  
& \!\!\!\!\!\!\!\rdelim\}{5}{-7mm}[\small\begin{tabular}{l}Analogous to\\\eqax{2}\ -- \eqax{7}\\but with symbol\\$\neq$ instead of $=$\end{tabular}]
\\
\hypertarget{neqaxsixteen}{\neqaxsixteenn} &$\tup{\alpha \cup \beta \neq \gamma}$&$\equiv$&$\tup{\alpha \neq \gamma} \lor \tup{\beta \neq \gamma}$  
&
\\
\hypertarget{neqaxseventeen}{\neqaxseventeenn}  &$\varphi \land \tup {\alpha \neq \beta}$&$\equiv$&$\tup{[\varphi] \alpha \neq \beta}$ 
&
\\
\hypertarget{neqaxtwo}{\neqaxtwon}  & $\tup{\alpha \neq \beta}$& $\leq$ & $\tup{\alpha}$
&
\\ 
\hypertarget{neqaxfifteen}{\neqaxfifteenn} & $\tup{\gamma[\tup{\alpha\neq\beta}]}$ & $\leq$ & $\tup{\gamma\alpha\neq\gamma\beta}$ 
& 
\\
\hypertarget{neqaxfive}{\neqaxfiven} & $\tup{\alpha=\gamma} \wedge \tup{\beta=\eta} $ & $\leq$ & $\tup{\alpha=\beta} \lor \tup{\gamma\neq\eta}$ \\
\hypertarget{neqaxeight}{\neqaxeightn} & $\tup{\alpha\neq\gamma} \wedge \tup{\beta=\eta} $ & $\leq$ & $\tup{\alpha\neq\beta} \lor \tup{\gamma\neq\eta}$
&
\\
\hypertarget{neqaxnine}{\neqaxninen} & $\tup{\gamma=\eta[\neg\tup{\alpha=\beta}\land\tup{\alpha}]\beta}$ & $\leq$ & $\tup{\gamma\neq\eta\alpha}$
&
\\
\hypertarget{neqaxten}{\neqaxtenn} & $\tup{\gamma\neq\eta[\neg\tup{\alpha\neq\beta}\land\tup{\alpha}]\beta}$ & $\leq$ & $\tup{\gamma\neq\eta\alpha}$ 
&
\\
\hypertarget{neqaxeighteen}{\neqaxeighteenn} & $\tup{\gamma = \eta[\lnot \tup{\alpha\neq\alpha} \land \tup{\alpha=\beta}]\alpha}$ & $\leq$ & $\tup{\gamma = \eta\beta}$
&\qquad\qquad\qquad
\vspace{.1in}\\\hline
\end{tabular}
}
\caption{Additional axiom schemes to allow for data inequality tests. The axiomatic system $\axiomRestrNeq$ consists of all the instantiations of this table, plus the ones of Table~\ref{tab:axiomseq}.} \label{tab:axiomsneq}
\end{table}
%
Let $\axiomRestrNeq$ be the set of all instantiations of the axiom schemes from Table~\ref{tab:axiomseq} plus the ones from Table~\ref{tab:axiomsneq}. In the scope of this section we will often say that a node expression is {\em consistent} meaning that it is $\axiomRestrNeq$-consistent (as in Definition \ref{def:con-nodeexp}).
%
%
%
\bigskip

Sometimes we use \neqaxone and \neqaxtwo without explicitly mentioning them.
We omit such steps in order to make the proofs more readable. 
We also note that \neqaxsixteen and \neqaxseventeen are necessary for the proof of Theorem~\ref{thm:normal-form-neq}, which is omitted; they have to be used in the same way as \eqax{3} and \eqax{4} in the proof of Theorem~\ref{thm:normal-form}.


It is not difficult to see that the axioms $\axiomRestrNeq$ are sound for $\xpd$:

\begin{proposition}[Soundness of $\xpd$]\label{prop:corr}
~
\thmsoundness{$\xpd$}{\axiomRestrNeq}
\end{proposition} 

%% file: nf-neq.tex

We define the sets $P_n$ and $N_n$, that contain the path and node expressions of $\xpd$, respectively,  in normal form at level $n$:

\begin{definition}[Normal form for $\xpd$]\label{def:cons-neneq} 
%
\begin{eqnarray*}
P_0 &=& \left\{\eps\right\}\\
N_0 &=& \{a\wedge \tup{\eps=\eps}\wedge\lnot\tup{\eps\neq\eps}\mid a\in\A\}   \\
P_{n+1} &=& \left\{\eps\} \cup \{ \dow[\psi]\beta\mid \psi\in N_{n},\beta\in P_n \right\}\\
D_{n+1} &=& \left\{\tup{\alpha=\beta}\mid \alpha,\beta\in P_{n+1}\} \cup \{\tup{\alpha\neq\beta}\mid \alpha,\beta\in P_{n+1}\right\}\\ 
N_{n+1} &=& \left\{a \wedge \bigwedge_{\varphi\in C} \varphi  \wedge \bigwedge_{\varphi\in D_{n+1}\setminus C} \neg\varphi \mid C\subseteq D_{n+1}, a\in \A\right\} \cap \con_\axiomRestrNeq.
\end{eqnarray*}
\end{definition} 
Normal forms are built using the same idea from \S \ref{subsec:nf}, but considering also 
data-aware diamonds with inequalities. 
Again, let us remark that it would suffice that $N_0$ contains formulas of the form $a$, for $a\in\A$, but we include instead formulas of the form $a\wedge\tup{\eps=\eps}\wedge\lnot\tup{\eps\neq\eps}$ (containing the tautologies $\tup{\eps=\eps}$
and $\lnot\tup{\eps\neq\eps}$) for technical reasons.
For instance, considering again two labels $a$ and $b$, the node expressions of $N_0$ are
$$
\psi=a \wedge\tup{\eps=\eps}\wedge\lnot\tup{\eps\neq\eps} \mbox{\quad and\quad}\theta= b \wedge\tup{\eps=\eps}\wedge\lnot\tup{\eps\neq\eps}.
$$
The sets $P_1$ and $D_1$ are as follows:%
\begin{eqnarray*}
	P_1&=&\{\dow[\psi]\eps,\dow[\theta]\eps,\eps\}
	\\
	D_1&=&\{\tup{\eps=\eps}, \tup{\dow[\psi]\eps=\dow[\theta]\eps},\tup{\eps=\dow[\psi]\eps},
	\tup{\eps=\dow[\theta]\eps},\tup{\dow[\psi]\eps=\dow[\psi]\eps},\tup{\dow[\theta]\eps=\dow[\theta]\eps},\\
	& & \tup{\eps\neq\eps}, \tup{\dow[\psi]\eps\neq\dow[\theta]\eps}, \tup{\eps\neq\dow[\psi]\eps}, \tup{\eps\neq\dow[\theta]\eps},\tup{\dow[\psi]\eps\neq\dow[\psi]\eps}, \tup{\dow[\theta]\eps\neq\dow[\theta]\eps}\}
\end{eqnarray*}
An example of a node expression in normal form at level 1, i.e.\ a node expression in $N_1$, is
\begin{eqnarray*}
\varphi &= &a \wedge \tup{\eps=\eps}  \wedge\neg\tup{\eps\neq\eps} \wedge \tup{\dow[\psi]\eps=\dow[\theta]\eps} \wedge \tup{\dow[\psi]\eps= \dow[\psi]\eps} \wedge \tup{\dow[\theta]\eps = \dow[\theta]\eps} \wedge\\
& &  \wedge  \tup{\eps\neq\dow[\psi]\eps}  \wedge  \tup{\eps\neq\dow[\theta]\eps}  \wedge\tup{\dow[\psi]\eps\neq\dow[\theta]\eps}\wedge \neg\tup{\eps=\dow[\psi]\eps} \wedge \neg\tup{\eps=\dow[\theta]\eps}\wedge\\
& & \wedge \tup{\dow[\theta]\eps\neq\dow[\theta]\eps} \wedge \tup{\dow[\psi]\eps\neq\dow[\psi]\eps}.
\end{eqnarray*}
Analogs of Lemmas \ref{lem:eq_in_psi}, \ref{nextPathIsConjunct} and \ref{lemma:inconsistent} hold in this case, with the same proofs as those given for the case of $\xpdeq$:

\begin{lemma}
\label{lem:eq_in_psineq} Let $*\in\{=,\neq\}$,  $\psi\in N_n$ and $\alpha,\alpha'\in P_n$. Let $\Tt,u$ be a pointed data tree,
such that $\Tt,u\models\psi$ and $\Tt,u\models\tup{\alpha * \alpha'}$. Then $\tup{\alpha*\alpha'}$
is a conjunct of $\psi$
\end{lemma}

\begin{lemma} \label{nextPathIsConjunctneq}
Let $\psi\in N_{n}$ and $\alpha\in P_{n}$.  If $[\psi]\alpha$ is consistent then $\tup{\alpha=\alpha}$ is a conjunct of~$\psi$. As an immediate consequence, if $\tup{\dow[\psi]\alpha}$ is consistent then $\tup{\alpha=\alpha}$ is a conjunct of~$\psi$.
\end{lemma}

\begin{lemma}
\label{lemma:inconsistentneq}
For every pair of distinct elements $\varphi,\psi \in N_n$, $\varphi\wedge\psi$ is
inconsistent.
\end{lemma}


We omit the proof of the following theorem, since it is analogous to the one for $\axiomRestr$ (Theorem \ref{thm:normal-form}):

\begin{theorem}[Normal form for $\xpd$]\label{thm:normal-form-neq}
\thmnormalform{$\xpd$}{\axiomRestrNeq}{N_n}{P_n}
\end{theorem}

The following two technical lemmas, whose proofs are deferred to Appendix~\ref{app}, will be needed for the construction of the canonical model:

\newcommand{\lemaA}
{
Let $*\in\{=,\neq\}$, 
 $\gamma \in P_n$, $\psi_i \in N_{n-i}$ for $i=1,\dots,i_0$, $\alpha, \beta \in P_{n-i_0}$ such that 
 $$\tup{\gamma * \dow[\psi_1]\dots\dow[\psi_{i_0}]\alpha} \land \lnot \tup{\gamma * \dow[\psi_1]\dots\dow[\psi_{i_0}]\beta}$$
 is consistent and $\lnot \tup{\alpha \neq \alpha}$ is a conjunct of $\psi_{i_0}$. Then $\lnot \tup{\alpha = \beta}$ is a conjunct of $\psi_{i_0}$.
}
\begin{lemma}\label{lema A} 
\lemaA
\end{lemma}

\newcommand{\paraverif}
{ Let $\psi \in N_n$, $\alpha,\beta \in P_n$ such that $\tup{\dow[\psi]\alpha \neq \dow[\psi]\alpha} \land \lnot \tup{\dow[\psi]\gamma \neq \dow[\psi]\gamma}$ is consistent and $\lnot \tup{\alpha \neq \alpha}$ is a conjunct of $\psi$. Then $\lnot \tup{\alpha = \gamma}$ is a conjunct of $\psi$.}

\begin{lemma}\label{paraverif}
  \paraverif
\end{lemma}

%% file: compl-neq.tex

In this section we show that for node expressions $\varphi$ and $\psi$ of $\xpd$, the equivalence $\varphi\equiv\psi$ is derivable from the axiom schemes of Table~\ref{tab:axiomseq} plus Table~\ref{tab:axiomsneq} if and only if $\varphi$ is $\xpd$-semantically equivalent to $\psi$. We also show the corresponding result for path expressions of $\xpd$. 
\begin{theorem}[Completeness of $\xpd$]\label{thm:sat}
\thmcompletenessnode{$\xpd$}{\axiomRestrNeq}
\end{theorem}
%
%
The proof of the above theorem is analogous to that of Theorem~\ref{thm:sat-restr}. The critical part of the argumentation is the analog of Lemma~\ref{lem:construction-restr} for the more expressive logic $\xpd$:

\begin{lemma}\label{lem:construction}
Any node expression $\varphi\in N_n$ is satisfiable.
\end{lemma}

The rest of this section, namely \S\ref{canonical model}, is devoted to the proof of Lemma~\ref{lem:construction}.

\subsubsection{Canonical model}\label{canonical model}
\input{construction-neq}

\input{verif-neq}

%% file: construction-neq.tex

We construct, recursively in $n$ and for every $\varphi\in N_n$, a data tree $\Tt^\varphi=(T^\varphi,\pi^\varphi)$ such that $\varphi$ is satisfiable in $\Tt^\varphi$. 

For the base case, if $\varphi\in N_{0}$ and $\varphi=a \wedge \tup{\eps=\eps} \land \lnot \tup{\eps \neq \eps}$ with $a \in\A$, we define  the data tree $\Tt^\varphi=(T^\varphi,\pi^\varphi)$ where $T^\varphi$ is a tree which consists of the single node $x$ with label $a$, and $\pi^\varphi = \{\{x\}\}$. 

Now, let $\varphi\in N_{n+1}$. Since $\varphi$ is a conjunction as in Definition~\ref{def:cons-neneq}, it is enough to guarantee that the following conditions hold (observe that we are using \eqax{2} and \neqaxone but we usually avoid these observations of symmetry):
\begin{enumerate}[label=(C\arabic*)]
 \item\label{labelneq} If $a\in \A$ is a conjunct of $\varphi$, then the {\rm root} $r^\varphi$ of $\Tt^\varphi$ has label $a$. 
 
 \item\label{epsiloneqpsialphaneq} If $\tup{\eps = \dow [\psi] \alpha }$ is a conjunct of $\varphi$, then there is a child $r^{\bf v}$ of the root $r^\varphi$ of $\Tt^\varphi$ at which $\psi$ is satisfied, and a node $x^{\bv}$ with the same data value as $r^\varphi$
 such that $\Tt^{\varphi}, r^\bv, x^{\bv} \models \alpha$. 

 \item\label{epsilonneqpsialphaneq} If $\tup{\eps \neq \dow [\psi] \alpha }$ is a conjunct of $\varphi$, then there is a child $r^{\bf v}$ of the root $r^\varphi$ of $\Tt^\varphi$ at which $\psi$ is satisfied, and a node $x^{\bv}$ with different data value than $r^\varphi$
 such that $\Tt^{\varphi}, r^\bv, x^{\bv} \models \alpha$.   
 
 \item\label{psialphaeqrhobetaneq} If $\tup{\dow [\psi]\alpha = \dow [\rho] \beta}$ is a conjunct of $\varphi$, then there are two children $r^\bu_1$, $r^\bu_2$ of the root $r^\varphi$ of $\Tt^\varphi$ at which $\psi$ and $\rho$ are satisfied respectively, and there are nodes $x^\bu$ and $y^\bu$ with the same data value such that $\Tt^{\varphi}, r^\bu_1, x^{\bu} \models \alpha$ and $\Tt^{\varphi}, r^\bu_2, y^{\bu} \models \beta$.

 \item\label{psialphaneqrhobetaneq} If $\tup{\dow [\psi]\alpha \neq \dow [\rho] \beta}$ is a conjunct of $\varphi$, then there are two children $r^\bu_1$, $r^\bu_2$ of the root $r^\varphi$ of $\Tt^\varphi$ at which $\psi$ and $\rho$ are satisfied respectively, and there are nodes $x^\bu$ and $y^\bu$ with different data value such that $\Tt^{\varphi}, r^\bu_1, x^{\bu} \models \alpha$ and $\Tt^{\varphi}, r^\bu_2, y^{\bu} \models \beta$.
 
 \item\label{noepsiloneqpsialphaneq} If $\lnot \tup{\eps = \dow [\psi] \alpha }$ is a conjunct of $\varphi$, then for each child $z$ of the root $r^\varphi$ of $\Tt^\varphi$ at which $\psi$ is satisfied, if $x$ is a node such that $\Tt^\varphi, z, x \models \alpha$, then the data value of $x$ is different than the one of $r^\varphi$.
 
 \item\label{noepsilonneqpsialphaneq}  If $\lnot \tup{\eps \neq \dow [\psi] \alpha }$ is a conjunct of $\varphi$, then for each child $z$ of the root $r^\varphi$ of $\Tt^\varphi$ at which $\psi$ is satisfied, if $x$ is a node such that $\Tt^\varphi, z, x \models \alpha$, then the data value of $x$ is the same as the one of $r^\varphi$.  
 
 \item\label{nopsialphaeqrhobetaneq} If $\lnot \tup{\dow [\psi]\alpha = \dow [\rho] \beta}$ is a conjunct of $\varphi$, then for each children $z_1, z_2$ of the root $r^\varphi$ of $\Tt^\varphi$ at which $\psi$ and $\rho$ are satisfied respectively, if $w_1, w_2$ are nodes such that $\Tt^\varphi, z_1, w_1 \models \alpha$ and $\Tt^\varphi, z_2, w_2 \models \beta$, then the data values of $w_1$ and $w_2$ are different.

 \item\label{nopsialphaneqrhobetaneq} If $\lnot \tup{\dow [\psi]\alpha \neq \dow [\rho] \beta}$ is a conjunct of $\varphi$, then for each children $z_1, z_2$ of the root $r^\varphi$ of $\Tt^\varphi$ at which $\psi$ and $\rho$ are satisfied respectively, if $w_1, w_2$ are nodes such that $\Tt^\varphi, z_1, w_1 \models \alpha$ and $\Tt^\varphi, z_2, w_2 \models \beta$, then $w_1$ and $w_2$ have the same data value.
 
\end{enumerate}

As in \S\ref{construccion}, we first give an intuitive description of the construction of the model, and then proceed to formalize it:

\subsubsection*{Insight into the construction}
 
The construction given in \S\ref{construccion} has some similarities with the one we are about to present. As before, we will hang, from a common {\rm root}, copies of trees given by inductive hypothesis to guarantee the satisfaction of some conjuncts of $\varphi$. Like in the previous case, we may need to introduce some changes on those trees in order to avoid spoiling the satisfaction of other conjuncts. 

However, this construction is far more complex than the one for $\xpdeq$. In the previous case, when adding new witnesses with fresh data values, one only needed to be careful enough to avoid putting in the same class nodes that should  be in different classes. Now, in addition to that (which is also harder to achieve, as witnessed by the differences between Lemmas~\ref{Lemma:ClavePlusMinus} and~\ref{lema clave plus} explained at the end of the latter), one also needs to guarantee conditions of the form $\lnot \tup{\mu \neq \delta}$ with $\mu, \delta \in P_{n+1}$ which force the merging of classes of {\em every} witness of the kind of paths involved that could appear along the construction. 

Unlike the case of $\xpdeq$, each pair of path expressions $\mu, \delta$ in $P_{n+1}$ will occur in two conjuncts of $\varphi$ instead of one (we do not  care about symmetric repetitions). Indeed, in the case of $\xpdeq$, for $\mu, \delta$ in $\nfP{n+1}$, we either have $\tup{\mu =\delta}$ or $\lnot \tup{\mu = \delta}$ as a conjunct of a node expression in $\nfN{n+1}$. Now we have four choices because we also have either $\tup{\mu \neq \delta}$ or $\lnot \tup{\mu \neq \delta}$, and hence {\em two} conjuncts containing $\mu$ and $\delta$ will occur in node expressions of $N_{n+1}$. We cannot treat as separate from each other those two conjuncts in which the same pair $\mu, \delta$ 
appear, so we first split $P_{n+1}$ into four subsets to deal with diamonds that compare against the constant empty path:
%
%
%
%
%
%
%
%
%
\begin{align*}
\sisi&=\{(\psi,\alpha) \mid \psi \in N_n, \alpha \in P_n, \tup{\eps=\dow[\psi]\alpha} \hbox{ and } \tup{\eps\neq\dow[\psi]\alpha} \hbox{ are conjuncts of } \varphi\}
\\
\sino&=\{(\psi,\alpha) \mid \psi \in N_n, \alpha \in P_n, \tup{\eps=\dow[\psi]\alpha} \hbox{ and } \lnot \tup{\eps\neq\dow[\psi]\alpha} \hbox{ are conjuncts of } \varphi\}
\\
\nosi&=\{(\psi,\alpha) \mid \psi \in N_n, \alpha \in P_n, \lnot \tup{\eps=\dow[\psi]\alpha} \hbox{ and } \tup{\eps\neq\dow[\psi]\alpha} \hbox{ are conjuncts of } \varphi\}
\\
\nono&=\{(\psi,\alpha) \mid \psi \in N_n, \alpha \in P_n, \lnot \tup{\eps=\dow[\psi]\alpha} \hbox{ and } \lnot \tup{\eps\neq\dow[\psi]\alpha} \hbox{ are conjuncts of } \varphi\}
\end{align*}
%
%
We make the following observations regarding the above definitions:
%

\newcommand{\lemacinco}{
Let $\psi\in N_n$, $\alpha \in P_n$, $\gamma \in P_{n+1}$. If $\lnot \tup{\gamma = \dow[\psi]\alpha} \land \lnot \tup{\gamma \neq \dow[\psi]\alpha} \land \tup{\gamma} \land \tup{\dow[\psi]}$ is consistent, then $\lnot \tup{\alpha = \alpha}$ is a conjunct of $\psi$. 
}
 
 \begin{observation}
 For $(\psi, \alpha) \in \nono$, our axioms should tell us that either $\tup{\alpha}$ is not a conjunct of $\psi$ or $\dow [\psi] \beta$ does not appear in any other positive conjunct of $\varphi$. If this is not the case, then $\varphi$ would be clearly unsatisfiable and thus our axiomatic system would not be complete. This assertion is a consequence of the following lemma plus {\bf Der12} of Fact~\ref{fact boolean}. It is important to remark that the axioms required for the proof can be easily proven sound.

 \begin{lemma}\label{lema 5}
\lemacinco
\end{lemma}

\begin{proof}
 See \S\ref{app}.
\end{proof}

Then, by Lemma~\ref{lemma:inconsistentneq} plus the fact that we will construct our model by hanging from the {\rm root} the trees given by inductive hypothesis, we should not be worried about the satisfaction of either $\lnot \tup{\eps = \dow [\psi] \alpha}$ nor $\lnot \tup{\eps \neq \dow [\psi] \alpha}$ because we will never create a pair of nodes witnessing the path $\dow[\psi] \alpha$.
\end{observation}

 \newcommand{\lemauno}{Let $*\in\{=,\neq\}$, 
 $\psi\in N_n$, $\alpha,\beta\in P_n$, $\gamma \in P_{n+1}$. If $\tup{\gamma = \dow[\psi]\alpha} \land \lnot \tup{\gamma \neq \dow[\psi]\alpha} \land \lnot \tup{\gamma * \dow[\psi]\beta}$ is consistent, then $\lnot \tup{\alpha * \beta}$ is a conjunct of $\psi$.}

\begin{observation}\label{los sino todos iguales} For $(\psi,\alpha) \in \sino$, our axioms should tell us that in a tree $\Tt^{\psi}$, any pair of nodes satisfying $\alpha$ ends in a node in the same equivalence class, since we want to put any such node in the class of the root $r^{\varphi}$. The following lemma has this property as an immediate consequence.

 \begin{lemma}\label{lema 1} 
\lemauno
 \end{lemma}
 
 \begin{proof}
 See \S\ref{app}.
\end{proof}
 \end{observation}
 
 \begin{observation}\label{los sino y los nosi distintos} For $(\psi,\alpha)\in \sino$ and $(\psi,\beta)\in \nosi$, Lemma~\ref{lema 1} also tells us that in a tree $\Tt^{\psi}$, any pairs of nodes satisfying $\alpha$ and $\beta$ end in points in different equivalence classes; which is also necessary to be able to satisfy $\varphi$. 
 \end{observation}
 
 \newcommand{\lemados}
 {Let $*\in\{=,\neq\}$, 
  $\psi\in N_n$, $\alpha,\beta\in P_n$, $\gamma \in P_{n+1}$. If $\tup{\gamma = \dow[\psi]\alpha} \land \lnot \tup{\gamma \neq \dow[\psi]\alpha}  \land  \tup{\gamma * \dow[\psi]\beta}$ is consistent, then $\tup{\alpha * \beta}$ is a conjunct of $\psi$. 
  }

 \begin{observation}\label{testigo para epsilonneqpsialpha} For $(\psi,\alpha)\in \sisi$ and $(\psi,\beta)\in \sino$, in order to obtain a witness for $\tup{\eps \neq \dow[\psi]\alpha}$, our axioms should tell us that in a tree $\Tt^{\psi}$ we can find a pair of nodes satisfying $\alpha$ starting from the root, and such that its ending node is in a different class from that of the ending node of any pair of nodes satisfying $\beta$ and beginning at the root of that tree. The following lemma combined with Observation~\ref{los sino todos iguales} has this as an immediate consequence.

 \begin{lemma}\label{lema 2} 
 \lemados
\end{lemma}
 
  \begin{proof}
 See \S\ref{app}.
\end{proof}
\end{observation}

\begin{observation}\label{testigo para epsiloneqpsialpha} For $(\psi,\alpha) \in \sisi$, in order to obtain a witness for $\tup{\eps = \dow[\psi]\alpha}$, we need a tree in which $\psi$ is satisfied and a pair of nodes (beginning at the root of that tree) satisfying $\alpha$ and ending in a node such that: it is in the class of the ending nodes of pairs of nodes satisfying $\beta$ for $(\psi,\beta) \in \sino$,  but it is not in the class of any ending node of a pair of nodes satisfying $\gamma$ for $(\psi,\gamma) \in \nosi$. In case there exists $\beta \in P_n$ such that $(\psi,\beta) \in \sino$, any tree at which $\psi$ is satisfied will work by the previous observations and lemmas. But in case $(\psi,\beta) \not \in \sino$ for all $\beta \in P_n$, we will have to make use of Lemma~\ref{lema clave plus} (the analogous of Lemma~\ref{Lemma:ClavePlusMinus} for this case).
\end{observation}

\paragraph{\em Processing data-aware diamonds of the form \boldmath{$(\lnot)\tup{\eps*\dow[\psi]\alpha}$}.} Having all these observations at hand, we begin by analyzing the following (non-disjoint) cases to construct our tree $\Tt^{\varphi}$:

\begin{enumerate}[label=(Case \arabic*),leftmargin=1.6cm]
 \item\label{Rule1} For $(\psi, \alpha) \in \sisi$, we add two witnesses. One for $\tup{\eps = \dow[\psi]\alpha}$ from which we  merge the class of the ending point $x^{\bf v_1}$ of a pair of nodes satisfying $\alpha$ as in Observation~\ref{testigo para epsiloneqpsialpha} with the class of $r^\varphi$. We add another witness for $\tup{\eps \neq \dow[\psi]\alpha}$ (remember Observation~\ref{testigo para epsilonneqpsialpha}). See Figure~\ref{intuicionneq}(a).   
 
 \item\label{Rule2} For $(\psi, \alpha) \in \sino$, we  add one witness for $\tup{\eps =\dow[\psi]\alpha}$ (see Figure~\ref{intuicionneq}(b)) and, at the end of the construction, we will merge the class of any node $x$ such that $r^\varphi, x \models \dow[\psi] \alpha$ with the class of $r^\varphi$ (remember Observation~\ref{los sino todos iguales}).
 
 \item\label{Rule3} For $(\psi, \alpha) \in \nosi$, we  add one witness for $\tup{\eps \neq \dow[\psi]\alpha}$ (See Figure~\ref{intuicionneq}(c)). Note that $\tup{\eps\neq \dow[\psi]\alpha} \land \lnot \tup{\eps= \dow[\psi]\alpha}$ will be satisfied by Observations~\ref{los sino y los nosi distintos} and~\ref{testigo para epsiloneqpsialpha}.
\end{enumerate}

\begin{figure}[ht]
   \begin{center}
   \begin{tabular}{c@{\hskip .5in}c@{\hskip .5in}c}
   \includegraphics[scale=0.25]{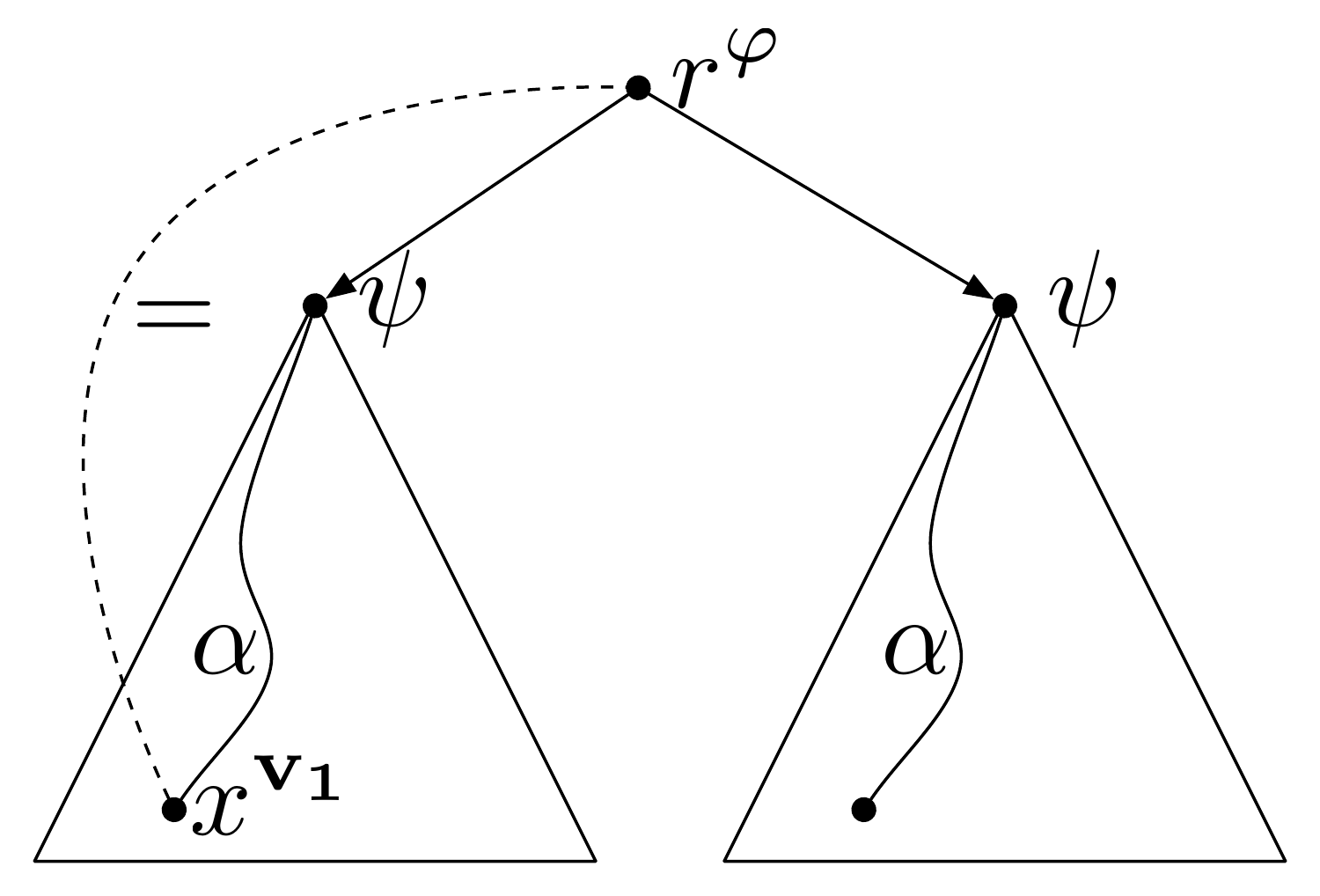}&\includegraphics[scale=0.25]{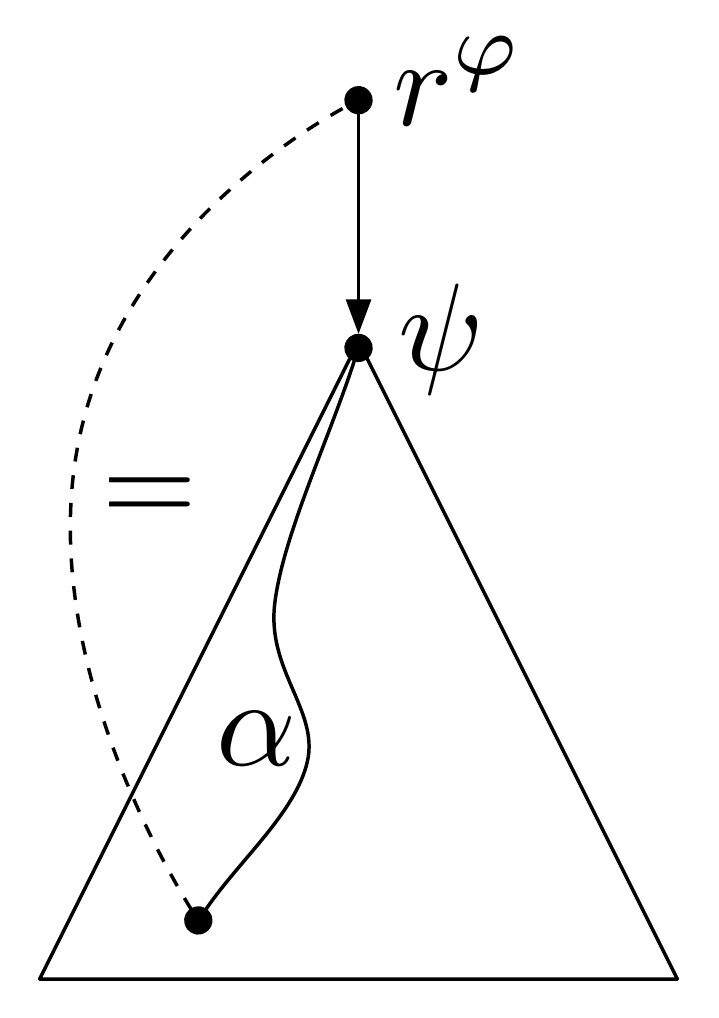}&\includegraphics[scale=0.25]{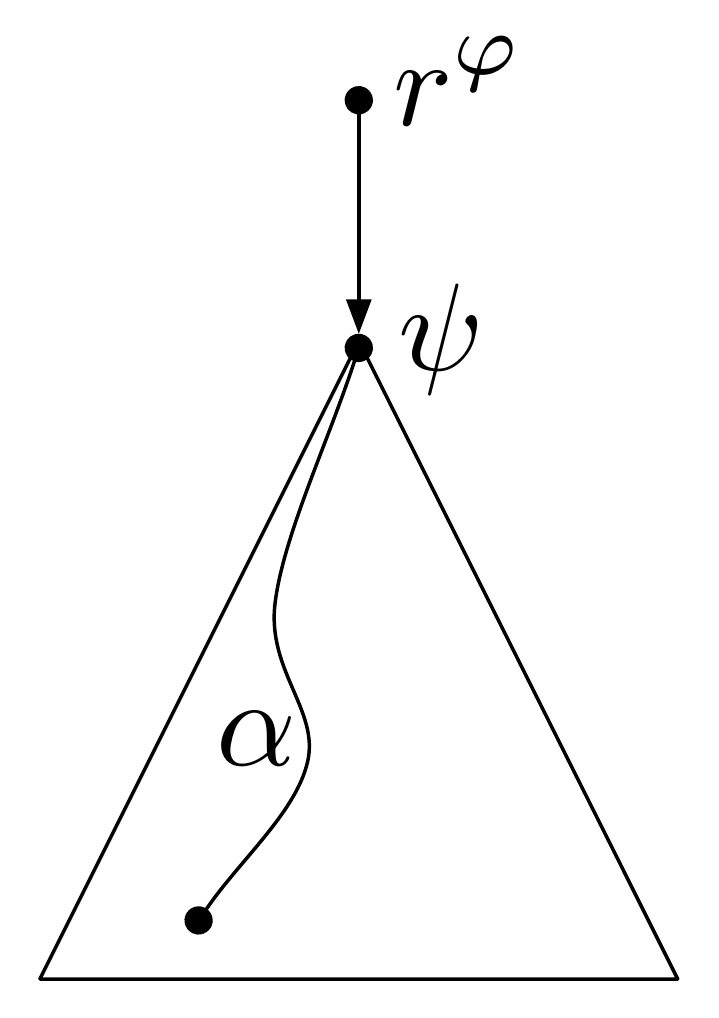}\\
   (a)&(b)&(c)
   \end{tabular}
   \end{center}
   \caption{(a) Witnesses for $\tup{\eps=\dow[\psi]\alpha}$ and $\tup{\eps \neq\dow[\psi]\alpha}$ for $(\psi,\alpha)\in \sisi$; (b) A witness for $\tup{\eps=\dow[\psi]\alpha}$ for $(\psi,\alpha)\in \sino$; (c) A witness for $\tup{\eps\neq\dow[\psi]\alpha}$ for $(\psi,\alpha)\in \nosi$.}   \label{intuicionneq}
\end{figure}

\paragraph{\em Processing data-aware diamonds of the form \boldmath{$(\lnot)\tup{\dow[\psi] \alpha * \dow[\rho] \beta}$}.}
For conjuncts of $\varphi$ the form $(\lnot) \tup{\dow[\psi] \alpha * \dow[\rho] \beta}$ that do not involve comparison with the constant path $\eps$, we have that, depending on which of the sets $\sisi, \sino, \nosi, \nono$  do $(\psi, \alpha)$ and $(\rho,\beta)$ belong to, many of the four possible combinations ($\tup{\dow[\psi]\alpha = \dow[\rho]\beta}$ and $\tup{\dow[\psi]\alpha \neq \dow[\rho]\beta}$, $\tup{\dow[\psi]\alpha = \dow[\rho]\beta}$ and $\lnot \tup{\dow[\psi]\alpha \neq \dow[\rho]\beta}$, etc.) are not possible as conjuncts for a consistent $\varphi$. More specifically: 
%
\begin{enumerate}[label=(Case \arabic*),leftmargin=1.6cm]
\setcounter{enumi}{3}
 \item \label{conjuntoZ} If we have $\tup{\dow[\psi]\alpha =\dow[\rho]\beta}$ and $\lnot \tup{\dow[\psi]\alpha \neq \dow[\rho]\beta}$ as conjuncts of $\varphi$, then all the following cases should be impossible since, in that case, $\varphi$ would be clearly unsatisfiable and thus it should be inconsistent: $(\psi, \alpha)$ or $(\rho, \beta)$ in $\sisi$, $(\psi, \alpha)$ or $(\rho, \beta)$ in $\nono$, one in $\sino$ and the other in $\nosi$. Besides, if both belong to $\sino$, since we merge the class of any node $x$ such that $r^\varphi, x \models \dow[\psi]\alpha$ or $r^\varphi, x \models \dow[\rho]\beta$, those conjuncts $\tup{\dow[\psi]\alpha =\dow[\rho]\beta}$ and $\lnot \tup{\dow[\psi]\alpha \neq \dow[\rho]\beta}$ will be satisfied. If both belong to $\nosi$, we need to force these conjuncts by merging the class of any node $x$ such that $r^\varphi, x \models \dow[\psi]\alpha$ or $r^\varphi, x \models \dow[\rho]\beta$ (note that we have such nodes by~\ref{Rule3}). It is 
important to notice that this process does not  add nodes to the 
class of the root since such nodes $x$ are never in the same equivalence class than any $x^{\bf v_1}$ from~\ref{Rule1} nor in the same equivalence class of a witness of $\tup{\dow[\mu]\delta}$ for $(\mu,\delta)\in \sino$.
 
 \item  \label{conjuntoU} If we have $\tup{\dow[\psi]\alpha =\dow[\rho]\beta}$ and $\tup{\dow[\psi]\alpha \neq \dow[\rho]\beta}$ as conjuncts of $\varphi$, then it cannot be the case that $(\psi,\alpha)$ or $(\rho,\beta)$ belong to $\nono$. Neither is possible that both of them belong to $\sino$ or one to $\sino$ and the other to $\nosi$. Besides, if $(\psi, \alpha), (\rho, \beta)$ belong to $\sisi$, then $\tup{\dow[\psi]\alpha =\dow[\rho]\beta}$ and $\tup{\dow[\psi]\alpha \neq \dow[\rho]\beta}$ are already satisfied: $\tup{\dow[\psi]\alpha =\dow[\rho]\beta}$ by the witnesses for $\tup{\eps = \dow[\psi]\alpha}$ and $\tup{\eps = \dow[\rho]\beta}$ (see Figure~\ref{intuicionneq2}(a)), $\tup{
\dow[\psi]\alpha \neq \dow[\rho]\beta}$ by the witnesses for $\tup{\eps = \dow[\psi]\alpha}$ and $\tup{\eps \neq \dow[\rho]\beta}$ (see Figure~\ref{intuicionneq2}(b) and remember Observation~\ref{testigo para epsilonneqpsialpha}). In case one belongs to $\sisi$ and the other to $\sino$, the argument is similar. If one belongs to $\sisi$ and the other to $\nosi$ or both to $\nosi$, $\tup{\dow[\psi]\alpha\neq \dow[\rho]\beta}$ will be satisfied using arguments similar to the previous ones; but we need to add witnesses to guarantee the satisfaction of $\tup{\dow[\psi]\alpha= \dow[\rho]\beta}$ (see Figure~\ref{intuicionneq2}(c)). In some cases, the merging performed in~\ref{conjuntoZ}, would have already merged the classes of a witness for $\tup{\dow[\psi]\alpha}$ and a witness for $\tup{\dow[\rho]\beta}$, in the remaining cases, we will need to force that merging carefully enough not to spoil conditions~\ref{noepsiloneqpsialphaneq} and~\ref{nopsialphaeqrhobetaneq} (we will use Lemma~\ref{lema clave plus} 
to achieve that).   

\begin{figure}[ht]
   \begin{center}
   \begin{tabular}{c@{\hskip .5in}c@{\hskip .5in}c}
   \includegraphics[scale=0.25]{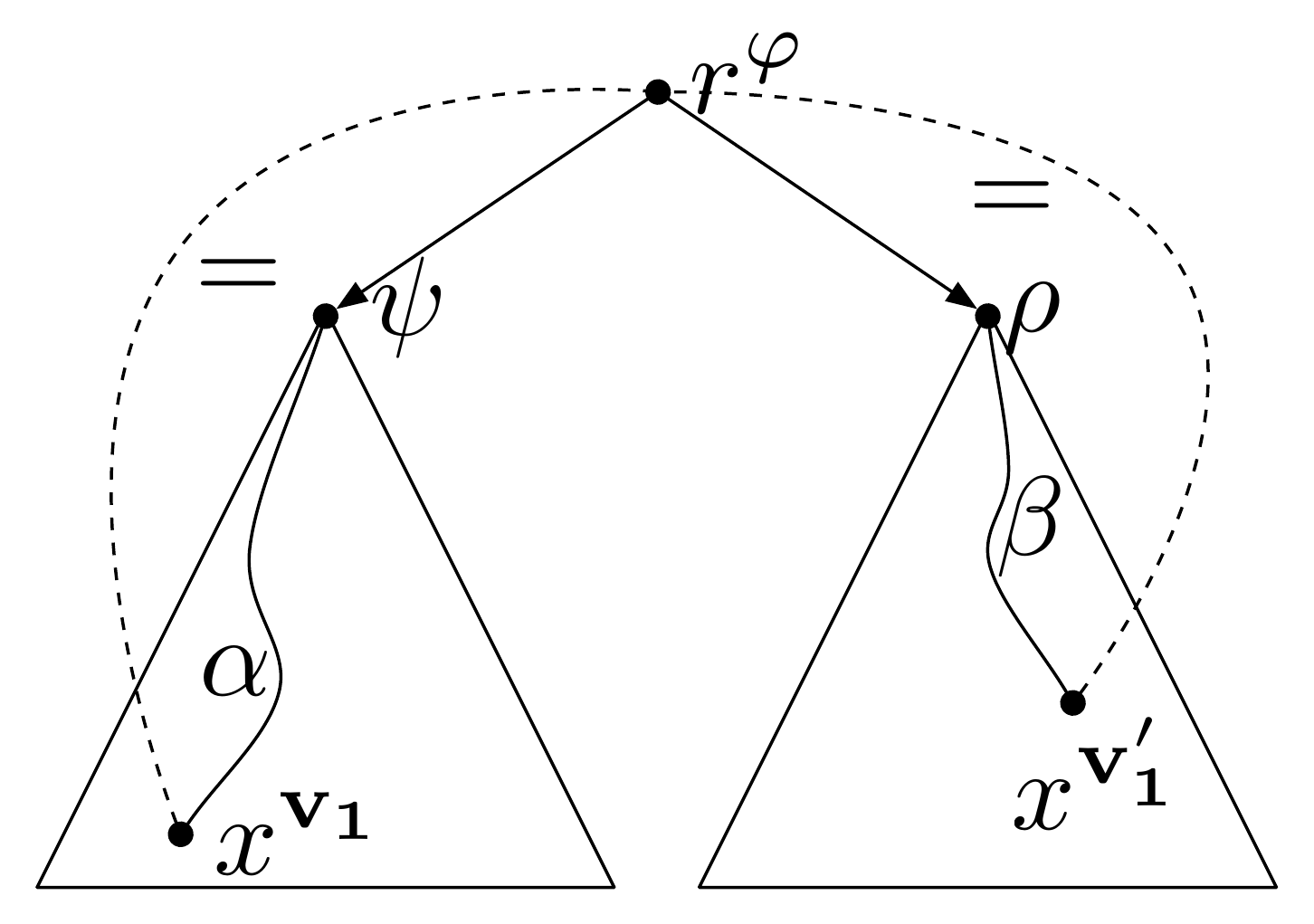}&\includegraphics[scale=0.25]{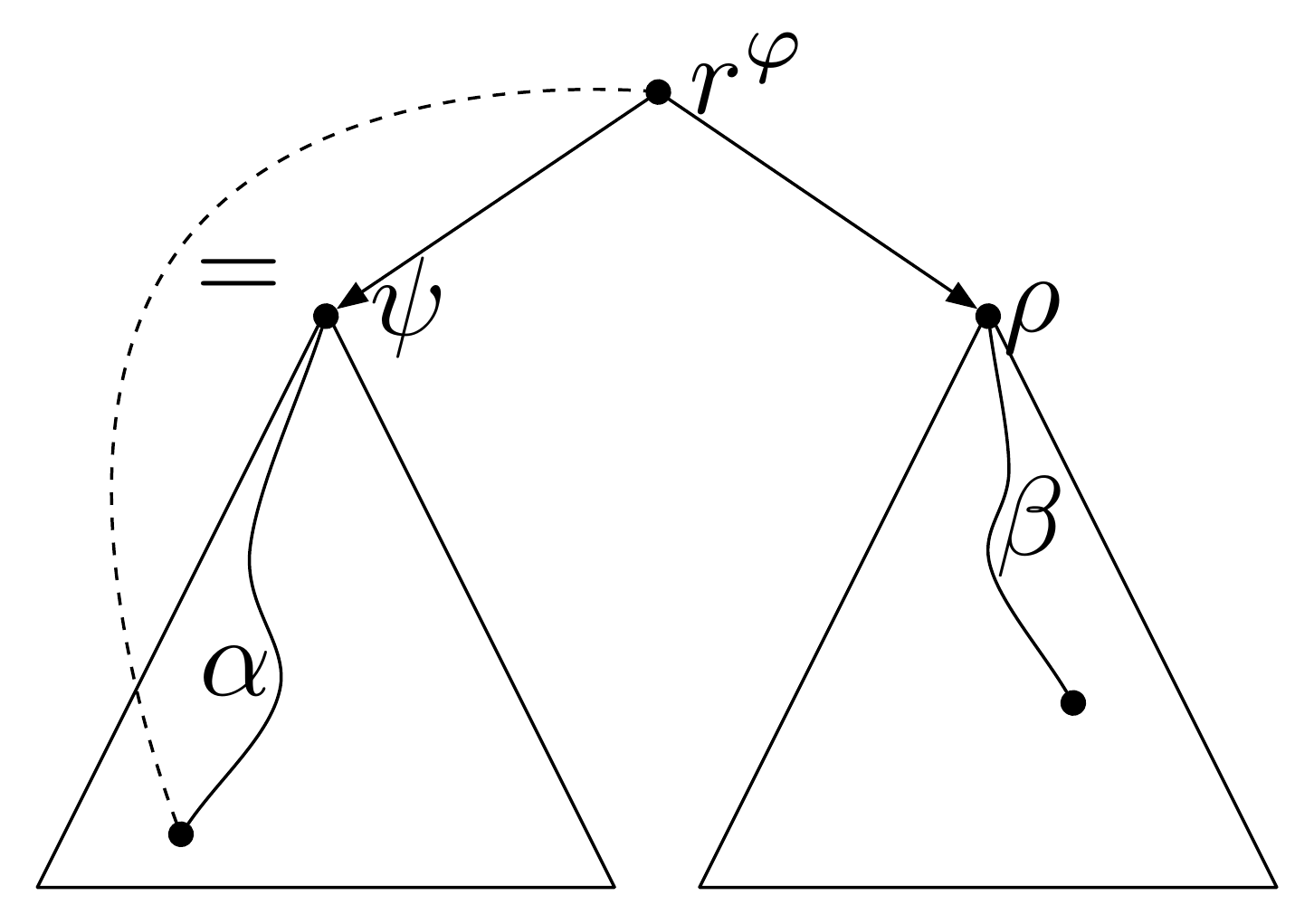}&\includegraphics[scale=0.25]{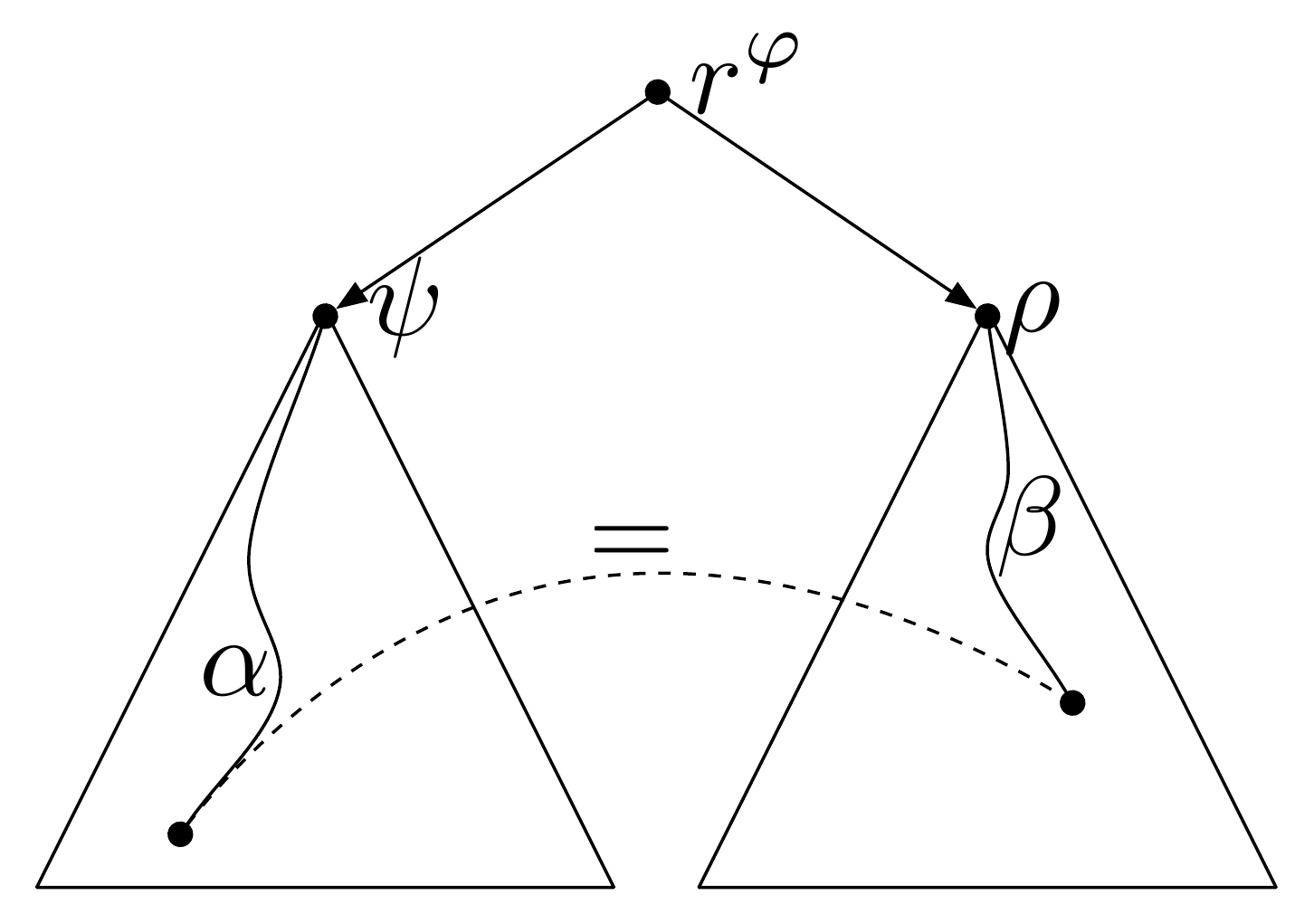}\\
   (a)&(b)&(c)
   \end{tabular}
   \end{center}
   \caption{(a) Witnesses for $\tup{\eps=\dow[\psi]\alpha}$ and $\tup{\eps = \dow[\rho]\beta}$ for $(\psi,\alpha)\in \sisi, (\rho,\beta)\in \sino$ end up in the same equivalence class; (b) Witnesses for $\tup{\eps=\dow[\psi]\alpha}$ for $(\psi,\alpha)\in \sino$ and $\tup{\eps\neq\dow[\rho]\beta}$ for $(\rho,\beta)\in \nosi$ end up in different equivalence classes; (c) Witnesses for $\tup{\dow[\psi]\alpha=\dow[\rho]\beta}$ for $(\psi,\alpha)\in \sisi, (\rho,\beta)\in \nosi$ or $(\psi,\alpha), (\rho,\beta)\in \nosi$.}   \label{intuicionneq2}
\end{figure}
\end{enumerate}

Finally, these last two cases are satisfied automatically:
\begin{enumerate}[label=(Case \arabic*),leftmargin=1.6cm]
\setcounter{enumi}{5}

 \item If we have $\lnot \tup{\dow[\psi]\alpha =\dow[\rho]\beta}$ and $\tup{\dow[\psi]\alpha \neq \dow[\rho]\beta}$ as conjuncts of $\varphi$, then all the following cases should be impossible: $(\psi, \alpha)$ or $(\rho, \beta)$ in $\nono$, both in $\sisi$ or $\sino$. Besides, if one belongs to $\sino$ and the other to $\nosi$ or if one belongs to $\sisi$ and the other to $\nosi$ or if they both belong to $\nosi$, $\lnot \tup{\dow[\psi]\alpha =\dow[\rho]\beta}$ and $\tup{\dow[\psi]\alpha \neq \dow[\rho]\beta}$ will be satisfied automatically ---the last two cases may not be as intuitive as others but are also true and we will give a detailed proof in time. 
  
 \item If we have $\lnot \tup{\dow[\psi]\alpha =\dow[\rho]\beta}$ and $\lnot \tup{\dow[\psi]\alpha \neq \dow[\rho]\beta}$ as conjuncts of $\varphi$, then the only case that should not lead to an inconsistency is when at least one of $(\psi, \alpha)$ and $(\rho, \beta)$ is in $\nono$ and, in this case, $\lnot \tup{\dow[\psi]\alpha =\dow[\rho]\beta}$ and $\lnot \tup{\dow[\psi]\alpha \neq \dow[\rho]\beta}$ will be satisfied automatically. 
  
\end{enumerate}



\bigskip
\subsubsection*{Formalization}

In order to formalize the construction described above, we introduce the following lemma, which is key to guarantee conditions~\ref{epsiloneqpsialphaneq} and~\ref{psialphaeqrhobetaneq} without spoiling conditions~\ref{noepsiloneqpsialphaneq} and~\ref{nopsialphaeqrhobetaneq}:

\begin{lemma}\label{lema clave plus}
 Let $\psi_0\in N_n$, $\alpha, \beta_1,\dots,\beta_m\in P_n$. Suppose that there exists a tree $\Tt^{\psi_0}=(T^{\psi_0},\pi^{\psi_0})$ with {\rm root} $r^{\psi_0}$ such that $\Tt^{\psi_0}, r^{\psi_0}\models \psi_0$ and for all $i=1,\dots,m $ there exists $\gamma_i\in P_{n+1}$ such that $\tup{\gamma_i=\dow[\psi_0]\alpha} \land \lnot \tup{\gamma_i=\dow[\psi_0]\beta_i}$ is consistent. Then there exists a tree $\widetilde{\Tt^{\psi_0}}=(\widetilde{T^{\psi_0}},\widetilde{\pi^{\psi_0}})$ with {\rm root} $\widetilde{r^{\psi_0}}$ and a node $x$ such that:
\begin{itemize} 
 \item $\widetilde{\Tt^{\psi_0}}, \widetilde{r^{\psi_0}}\models \psi_0$, 
 \item $\widetilde{\Tt^{\psi_0}}, \widetilde{r^{\psi_0}}, x\models \alpha$, and 
 \item $[x]_{\widetilde{\pi^{\psi_0}}}\neq [y]_{\widetilde{\pi^{\psi_0}}}$ for all $y$ such that $\widetilde{\Tt^{\psi_0}}, \widetilde{r^{\psi_0}}, y\models \beta_i$ for some $i=1,\dots,m$.
\end{itemize}
\end{lemma}

\begin{proof}
Suppose $\alpha=\dow[\psi_1]\dots\dow[\psi_{j_0}]\eps$ where $\psi_k \in N_{n-k} $ for all $k=1,\dots,j_0$ and let 
$$
k_0=\min_{0\leq k\leq j_0} \left\{k \mid \lnot \tup{\dow[\psi_{k+1}]\dots\dow[\psi_{j_0}]\eps \neq \dow[\psi_{k+1}]\dots\dow[\psi_{j_0}]\eps} \mbox{ is a conjunct of } \psi_k\right\}.
$$
In case $k_0=0$ (i.e. $\lnot \tup{\alpha \neq \alpha}$), by Lemma~\ref{lema A} 
, $\lnot \tup{\alpha = \beta_i}$ is a conjunct of $\psi_0 $ for all $i=1,\dots,m$. Then $\widetilde{\Tt^{\psi_0}}=(T^{\psi_0},\pi^{\psi_0})$ satisfies the desired properties. The intuitive idea behind this application of Lemma~\ref{lema A} is that in case every ending point of a pair of nodes satisfying $\alpha$ is in the same 
equivalence class, then there cannot be pairs of nodes satisfying $\alpha$ and $\beta_i$ ending in points with the same data value, because in that case $\tup{\gamma_i=\dow[\psi_0]\alpha} \land \lnot \tup{\gamma_i=\dow[\psi_0]\beta_i}$ would be unsatisfiable and thus inconsistent, which is a contradiction with our hypothesis.

In case $k_0\neq 0$, by consistency, there are $z',x'\in T^{\psi_0}$ such that $\Tt^{\psi_0},r^{\psi_0}, z'\models \dow[\psi_1]\dots\dow[\psi_{k_0}]$ and $\Tt^{\psi_0},z',x'\models \dow[\psi_{k_0+1}]\dots\dow[\psi_{j_0}]$. Before proceeding to complete the proof of this case, we give an intuitive idea.
We prove that we cannot have a witness for $\beta_i$ with the same data value than $x'$ in the subtree $T^{\psi_0} \restr{z'}$. Intuitively this is because, in that case, $\alpha$ and $\beta_i$ would have a common prefix. Let us say that
$$
\beta=\dow[\psi_1]\dots\dow[\psi_{k_0}]\dow [\rho_{k_0+1}]\dots \dow[\rho_{l_0}]\eps \mbox{\quad and\quad}
\tup{\dow[\psi_{k_0+1}]\dots\dow[\psi_{j_0}]\eps=\dow [\rho_{k_0+1}]\dots \dow[\rho_{l_0}]\eps}
$$
is a conjunct of $\psi_{k_0}$. Then, since $\lnot \tup{\dow[\psi_{k_0+1}]\dots\dow[\psi_{j_0}]\eps\neq \dow[\psi_{k_0+1}]\dots\dow[\psi_{j_0}]\eps}$ is also a conjunct of $\psi_{k_0}$, $\tup{\gamma_i=\dow[\psi_0]\alpha} \land \lnot \tup{\gamma_i=\dow[\psi_0]\beta_i}$ would be unsatisfiable (and thus inconsistent) for any choice of $\gamma_i$, which is a contradiction. But our hypotheses do not guarantee that we would not have a witness for $\beta_i$ in the class of $x'$ outside $T^{\psi_0} \restr{z'}$,
 and therefore we need to change the tree in order to achieve the desired properties. We replicate the subtree $T^{\psi_0} \restr{z'}$ but using a fresh data value (different from any other data value already present in $\Tt^{\psi_0}$) for the class of the companion of $x'$ that we call $x$; see Figure~\ref{fig:dibu-oneneq}. It is clear that in this way, the second and the third conditions will be satisfied by $x$. The first condition will also remain true because, 
intuitively, the positive conjuncts will 
remain valid since we are not suppressing any nodes, and
the negative ones that compare by equality will not be affected because every new node has either the same data value than its companion or a fresh data value. The argument for negative conjuncts that compare by inequality is based on the way in which we have chosen $k_0$ (see a detailed proof below). 

Now we formalize the previous intuition.
Let $p$ be the parent of $z'$ ($k_0>0$). 
As we did in the proof of Lemma~\ref{Lemma:ClavePlusMinus}, we define $\widetilde{\+T^{\psi_0}}$ by adding a new child $z$ of $p$ and a data tree $\Tt=(T,\pi)$ hanging from $z$. This tree $T$ is a copy of $T^{\psi_0}\restr{z}$, and we call $x$ to the companion of $x'$. $\widetilde{\pi^{\psi_0}}$ is defined as $\pi^{\psi_0}$ with the exception that the class of $x$ is new (the classes of the other nodes of $T$ are merged with the classes of their companions) (see Figure~\ref{fig:dibu-oneneq}).

\begin{figure}[ht]
   \begin{center}
   \includegraphics[scale=0.25]{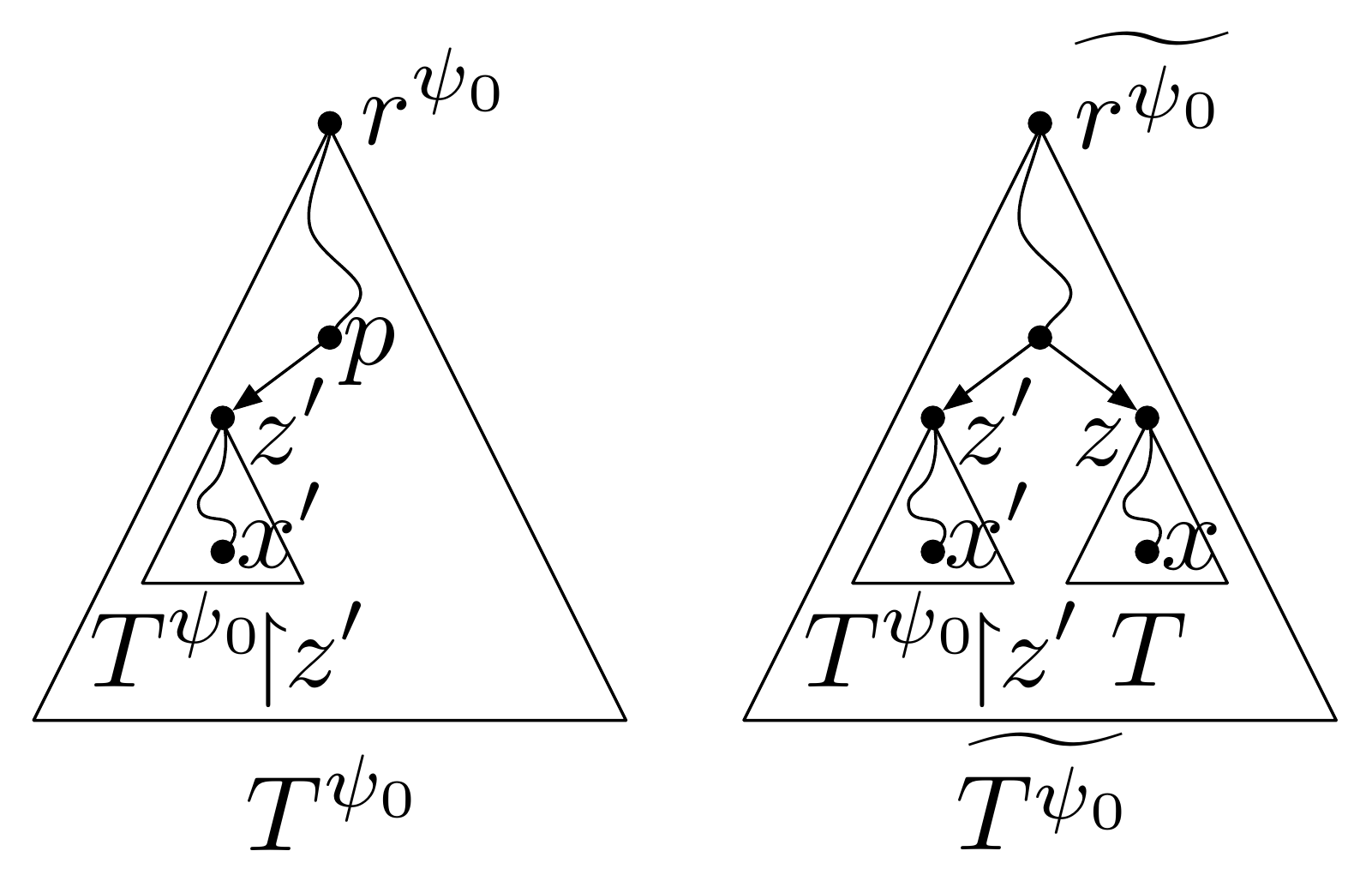}
   \end{center}
   \caption{$T= T^{\psi_0} \restr{z}$ is a new subtree with a special node $x$ such that its class of data values is disjoint to the rest of $\widetilde{\Tt^{\psi_0}}$ and $\widetilde{\Tt^{\psi_0}}, \widetilde{r^{\psi_0}}, x \models \alpha$.}   \label{fig:dibu-oneneq}
\end{figure}

We first prove by induction that $z_j$, the $j$-th ancestor of $z$ (namely $z_j\childp{j}z$, and we let $z_0:=z$), satisfies $\widetilde{\Tt^{\psi_0}}, z_j \models \psi_{k_0-j}$. This will prove both that $\widetilde{\Tt^{\psi_0}}, \widetilde{r^{\psi_0}}\models \psi_0$ and that $\widetilde{\Tt^{\psi_0}}, \widetilde{r^{\psi_0}}, x\models \alpha$.
By Proposition~\ref{prop:local}, it is straightforward from the construction that $\psi_{k_0}$ is satisfied at $z$ (the companion of $z'$) which proves the base case. 
For the inductive case, assume the result holds for $z_0, \dots, z_j$. We want to see that it holds for $z_{j+1}$. 
To do this, we check that every conjunct of $\psi_{k_0-j-1}$ is satisfied at $z_{j+1}$: 

\begin{itemize}
 \item If the conjunct is a label, it is clear that $z_{j+1}$ has that label in $\widetilde{\Tt^{\psi_0}}$, as it has not been changed by the construction.
 
 \item If the conjunct is of the form $\tup{\mu_1 = \mu_2}$ or $\tup{\mu_1 \neq \mu_2}$, then it must still hold in $\widetilde{\Tt^{\psi_0}}$ by inductive hypothesis and the fact that our construction did not remove nodes.
 
 \item If the conjunct is of the form $\lnot \tup{\mu_1 = \mu_2}$, we observe that, by inductive hypothesis plus the way in which we have constructed $\widetilde{\Tt^{\psi_0}}$, we have that: If $\widetilde{\Tt^{\psi_0}}, z_{j+1} \models \tup{\mu_1 = \mu_2}$ then $\Tt^{\psi_0}, z_{j+1} \models \tup{\mu_1 = \mu_2}$ (for a complete proof of this assertion, one can use arguments similar to the ones used in Lemma~\ref{Lemma:ClavePlusMinus}) which is a contradiction with the fact that \  $\Tt^{\psi_0}, z_{j+1} \models \psi_{k_0-(j+1)}$. Then $\widetilde{\Tt^{\psi_0}}, z_{j+1} \models \lnot \tup{\mu_1 = \mu_2}$. 
 

\item  If the conjunct is of the form $\lnot \tup{\mu_1 \neq \mu_2}$, by inductive hypothesis plus the way in which we have constructed $\widetilde{\Tt^{\psi_0}}$, $\tup{\mu_1 \neq \mu_2}$ can only be true in $z_{j+1}$ if there are witnesses $y_1, y_2$ in distinct equivalence classes such that $\widetilde{\Tt^{\psi_0}}, z_{j+1}, y_1 \models \mu_1$, $\widetilde{\Tt^{\psi_0}}, z_{j+1}, y_2 \models \mu_2$ and at least one of them is in the new subtree $T$. In that case, without loss of generality, we have that $\mu_1=\dow[\psi_{k_0-j}]\dots\dow[\psi_{k_0}]\hat{\mu_1}$. Then, by definition of $k_0$, $\tup{\dow[\psi_{k_0-j}]\dots\dow[\psi_{j_0}]\eps \neq \dow[\psi_{k_0-j}]\dots\dow[\psi_{j_0}]\eps}$ is a conjunct of $\psi_{k_0-j-1}$. Therefore, by consistency and  \neqaxeight, $\tup{\dow[\psi_{k_0-j}]\dots\dow[\psi_{j_0}]\eps \neq \mu_2}$ or $\lnot \tup{\mu_2=\mu_2}$ is also a conjunct of $\psi_{k_0-j-1}$. If the latter occurs, we have a contradiction by the previous item. If $\tup{\dow[\psi_{k_0-
j}]\dots\dow[\psi_{j_0}]\eps \neq \mu_2}$ is a conjunct of $\psi_{k_0-j-1}$, by Lemma~\ref{lema A} $\lnot \tup{\dow[\psi_{k_0+1}]\dots \dow[\psi_{j_0}]\eps = \hat{\mu_1}}$ is a conjunct of $\psi_{k_0}$. Then, by construction, the class of $y_1$ in $\widetilde{\Tt^{\psi_0}}$ is equal to the class of its companion and so we can assume that $y_1\not \in T$. Analogously we can assume that $y_2\not \in T$ but, as we have already said, by inductive hypothesis plus the way in which we have constructed $\widetilde{\Tt^{\psi_0}}$, $\tup{\mu_1 \neq \mu_2}$ cannot be satisfied at $z_{j+1}$ by witnesses $y_1, y_2$ if neither of them is in the new subtree $T$.  
\end{itemize}

To conclude the proof, we only need to check that $[x]_{\widetilde{\pi^{\psi_0}}}\neq [y]_{\widetilde{\pi^{\psi_0}}}$ for all $y$ such that \  $\widetilde{\Tt^{\psi_0}}, \widetilde{r^{\psi_0}}, y\models \beta_i$ for some $i=1,\dots,m$. Suppose that $\beta_i=\dow[\rho_1]\dots\dow[\rho_{l_0}]\eps$. If $l_0<k_0$ or $\rho_l\neq \psi_l$ for some $l=1,\dots,k_0$, then the result follows immediately from construction. If not, by hypothesis, there exists $\gamma_i \in P_{n+1}$ such that $\tup{\gamma_i= \dow[\psi_{0}]\dots\dow[\psi_{j_0}]\eps}  \land \lnot \tup{\gamma_i = \dow[\psi_{0}]\dots\dow[\psi_{k_0}]\dow[\rho_{k_0+1}]\dots\dow[\rho_{l_0}]\eps}$ is consistent and $\lnot \tup{\dow[\psi_{k_0+1}]\dots\dow[\psi_{j_0}]\eps \neq \dow[\psi_{k_0+1}]\dots\dow[\psi_{j_0}]\eps}$ is a conjunct of $\psi_{k_0}$. Then, by Lemma~\ref{lema A}, $\lnot \tup{\dow[\psi_{k_0+1}]\dots\dow[\psi_{j_0}]\eps = \dow[\rho_{k_0+1}]\dots\dow[\rho_{l_0}]\
\eps}$ 
is a conjunct of $\psi_{k_0}$.
 This together with the fact that the class of $x$ is disjoint with the part of $\widetilde{T^{\psi_0}}$ outside of $T$, shows that $[x]_{\widetilde{\pi^{\psi_0}}}\neq [y]_{\widetilde{\pi^{\psi_0}}}$ if $y$  is such that $\widetilde{\Tt^{\psi_0}}, \widetilde{r^{\psi_0}}, y\models \beta_i$, which concludes the proof.
\end{proof}

%

It might be useful for the reader to note the differences between Lemmas~\ref{lema clave plus} and~\ref{Lemma:ClavePlusMinus}, since this is one of the reasons why the completeness result for $\xpd$ is more complicated than for $\xpdeq$. The main differences between those two lemmas are:

\begin{itemize}
 
 \item In Lemma~\ref{lema clave plus}, if we would replicate the subtree hanging from a witness of $\tup{\alpha}$ then, due to the fact that we are working with the complete fragment (with inequality tests also), we would not be able to prove that each ancestor of that node satisfies the desired formulas. So we are forced to find that minimum $k_0$ that tells us which subtree we should replicate.
 
 \item In Lemma~\ref{Lemma:ClavePlusMinus}, we can use new data for every new node since, again, we are not working with inequality tests. But when it comes to the complete fragment, we need to be more careful in the way we define the partition in $\widetilde{\Tt^{\psi_0}}$ changing only the class of the new witness of $\tup{\alpha}$. 
\end{itemize}
Now that we have this key lemma, we proceed to the formal construction of $\Tt^{\varphi}$.
%
We define some special sets of quadruples $(\psi, \alpha, \rho, \beta)$ with $\psi, \rho \in N_n$, $\alpha, \beta \in P_n$:
\begin{itemize}
\item ${\bf U}$ is the set of quadruples $(\psi, \alpha, \rho, \beta)$ such that one of the following holds:
	\begin{itemize}
	\item $(\psi, \alpha), (\rho, \beta) \in \nosi$, and $\tup{\dow[\psi]\alpha=\dow[\rho]\beta}$, $\tup{\dow[\psi]\alpha\neq\dow[\rho]\beta}$ are conjuncts of $\varphi$, or
	\item $(\psi, \alpha) \in \sisi$, $(\rho, \beta) \in \nosi$, and $\tup{\dow[\psi]\alpha=\dow[\rho]\beta}$, $\tup{\dow[\psi]\alpha\neq\dow[\rho]\beta}$ are conjuncts of~$\varphi$.
	\end{itemize}
cf.~\ref{conjuntoU}.

\item ${\bf Z}$ is the set of all quadruples $(\psi, \alpha, \rho, \beta)$ such that $(\psi, \alpha), (\rho, \beta) \in \nosi$, and \  $\tup{\dow[\psi]\alpha=\dow[\rho]\beta}$, $\lnot \tup{\dow[\psi]\alpha\neq\dow[\rho]\beta}$ are conjuncts of $\varphi$.

\noindent cf.~\ref{conjuntoZ}.
\end{itemize}

The following lemma states that the relation between the elements of $\nosi$ defined by the set ${\bf Z}$ is transitive, a fact which will be needed to prove that $\varphi$ is indeed satisfied in the constructed tree:

\begin{lemma}\label{lema c}
If $(\psi,\alpha,\rho,\beta), (\rho,\beta,\theta,\gamma)\in {\bf Z}$, then $(\psi,\alpha,\theta,\gamma)\in {\bf Z}$.
\end{lemma}

\begin{proof}
 By  \neqaxeight and the consistency of $\varphi$, $\lnot \tup{\dow[\psi]\alpha \neq \dow[\theta]\gamma}$ is a conjunct of $\varphi$. Then, by consistency of $\varphi$ plus  \neqaxfour, $\tup{\dow[\psi]\alpha = \dow[\theta]\gamma}$ is also a conjunct of $\varphi$ which concludes the proof.
\end{proof}

Now that we have these lemmas, we proceed to construct $\Tt^{\varphi}$ as follows:

\paragraph {\bf Rule 1. Witnesses for \boldmath{${\bf v_1}=(\psi,\alpha)\in \sisi$}.} (cf.~\ref{Rule1}) We define data trees $\Tt_1^{\bf v_1}=(T_1^{\bf v_1},\pi_1^{\bf v_1})$ and $\Tt_2^{\bf v_1}=(T_2^{\bf v_1},\pi_2^{\bf v_1})$ with {\rm root}s $r^{\bf v_1}_1$ and $r^{\bf v_1}_2$ respectively.  In order to choose appropriate witnesses for $\tup{\eps=\dow[\psi]\alpha}$ and $\tup{\eps \neq \dow[\psi]\alpha}$, we need the following lemma:

\begin{lemma}\label{lema a}
 Let ${\bf v_1}=(\psi, \alpha) \in {\bf V_{=,\neq}}$. Then there exist $\widetilde{\Tt^\psi}=(\widetilde{T^\psi},\widetilde{\pi^\psi})$ with {\rm root} $\widetilde{r^\psi}$ and a node $x$ such that:
\begin{itemize} 
 \item $\widetilde{\Tt^\psi},\widetilde{r^\psi}\models \psi$, 
 \item $\widetilde{\Tt^\psi},\widetilde{r^\psi}, x\models \alpha$, 
\item $[x]_{\widetilde{\pi^\psi}}= [y]_{\widetilde{\pi^\psi}}$ for all $y$ such that there is $\beta \in P_n$ with $(\psi, \beta) \in \sino$ and $\widetilde{\Tt^\psi},\widetilde{r^\psi}, y\models \beta$, 
\item  
$[x]_{\widetilde{\pi^\psi}}\neq [z]_{\widetilde{\pi^\psi}}$ for all $z$ such that there is $\gamma \in P_n$ with $(\psi, \gamma) \in \nosi$ and $\widetilde{\Tt^\psi},\widetilde{r^\psi}, z\models \gamma$.
\end{itemize} 
\end{lemma}


\begin{proof}
We first analyze the case that there exists $\beta \in P_n$ such that $(\psi,\beta)\in {\bf V_{=,\lnot \neq}}$. Then, by Lemmas ~\ref{lema 1} and~\ref{lema 2}, the result is immediate from the fact that we are assuming there is a tree $\Tt^{\psi}$ satisfying $\psi$ at its {\rm root}. The idea is that by inductive hypothesis, there exists $\Tt^ {\psi}=(T^{\psi},\pi^{\psi})$ satisfying $\psi$ at is {\rm root}. Then, Lemma~\ref{lema 1} guarantees that every witness of some $\beta$ as described before belongs to the same class in $\pi^{\psi}$ and that every witness of some $\gamma$ as described before does not belong to this class. Finally, Lemma~\ref{lema 2} shows the existence of the desired node $x$. 

To conclude the proof, suppose that $(\psi,\beta)\not \in {\bf V_{=,\lnot \neq}}$ for all $\beta \in P_n$. Then the result follows from Lemma~\ref{lema clave plus}.
\end{proof}
%

Using Lemma~\ref{lema a}, define $T^{\bf v_1}_1$ as $\widetilde{T^\psi}$, $\pi^{\bf v_1}_1$ as $\widetilde{\pi^\psi}$, $r^{\bf v_1}_1$ as $\widetilde{r^\psi}$ and $x^{\bf v_1}=x \in T^{\bf v_1}_1$. Also, by inductive hypothesis, there exists a tree $\Tt^{\psi}=(T^{\psi}, \pi^{\psi})$ with {\rm root} $r^{\psi}$ such that $\Tt^\psi,r^\psi \models \psi$. Define $T^{\bf v_1}_2$ as $T^{\psi}$, $\pi^{\bf v_1}_2$ as $\pi^{\psi}$ and $r^{\bf v_1}_2$ as $r^{\psi}$. Without loss of generality, we assume that $T^{\bf v_1}_1$ and $T^{\bf v_1}_2$ are disjoint. In other words, the {\rm root}ed data tree $(T^{\bf v_1}_1,\pi^{\bf v_1}_1,r^{\bf v_1}_1)$ is just a copy of $(\widetilde{T^\psi},\widetilde{\pi^\psi},\widetilde{r^\psi})$ with a special node named 
$x^{\bf v_1}$ and $(T^{\bf v_1}_2,\pi^{\bf v_1}_2,r^{\bf v_1}_2)$ is just a copy of $(T^{\psi}, \pi^{\psi})$ disjoint with $(T^{\bf v_1}_1,\pi^{\bf v_1}_1,r^{\bf v_1}_1)$. See Figure~\ref{fig:constr-neq}(a). 


\begin{figure}[ht]
   \begin{center}
   \begin{tabular}{c@{\hskip .2in}c@{\hskip .2in}c}
   \includegraphics[scale=0.25]{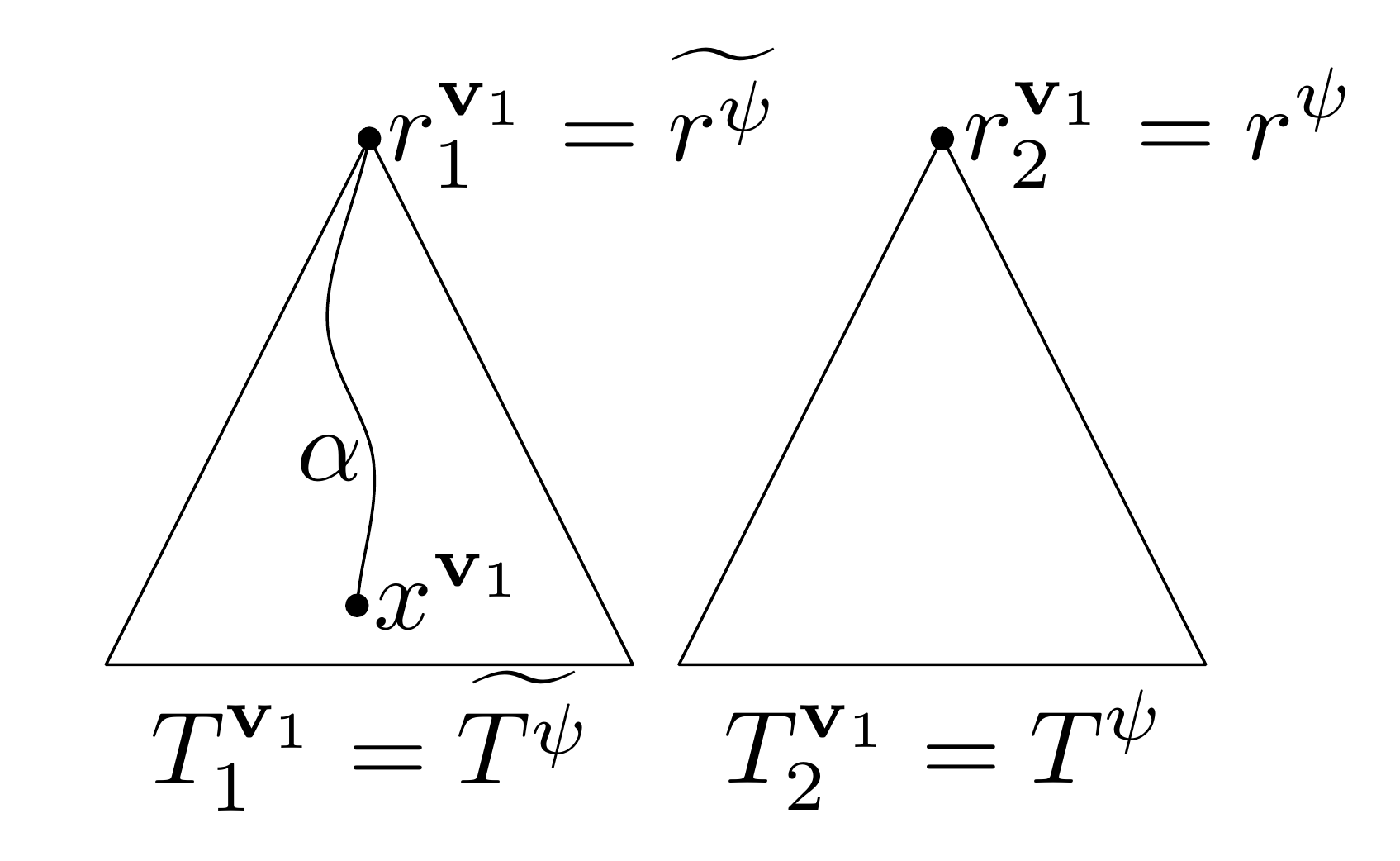}&
   \includegraphics[scale=0.25]{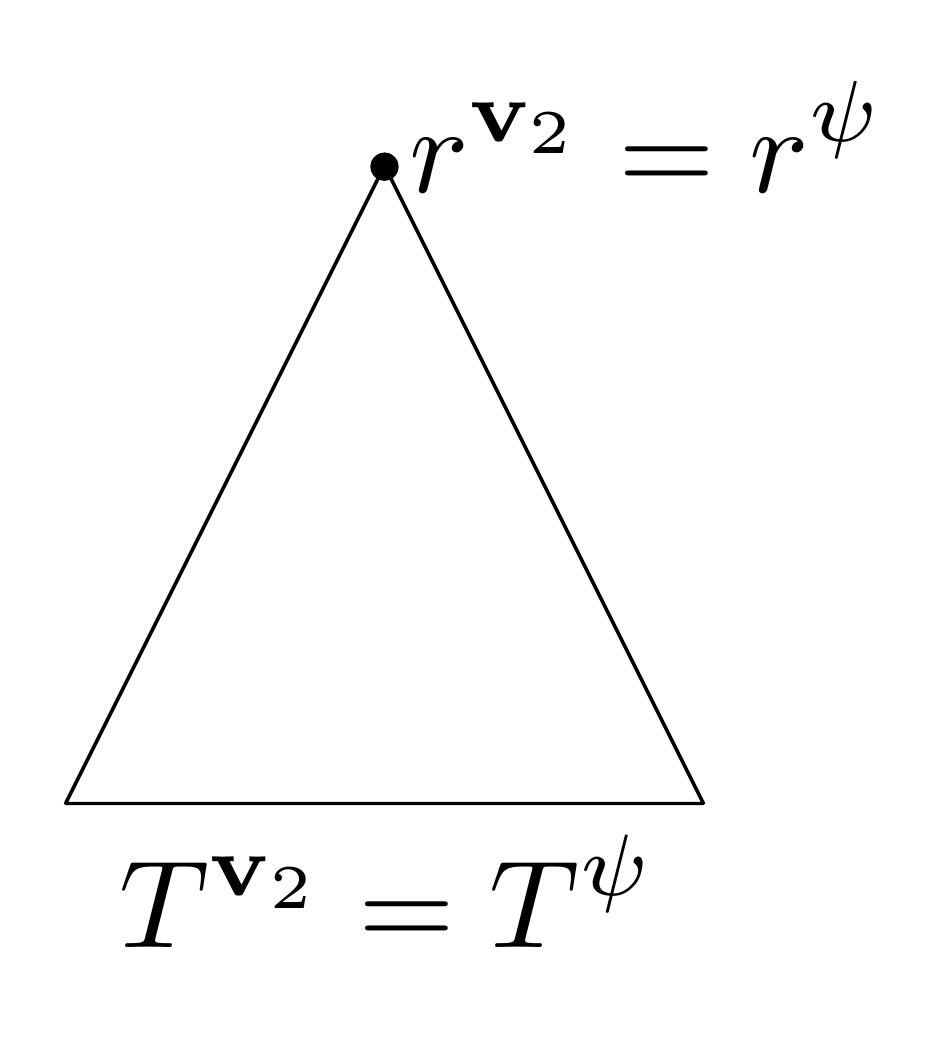}&
   \includegraphics[scale=0.25]{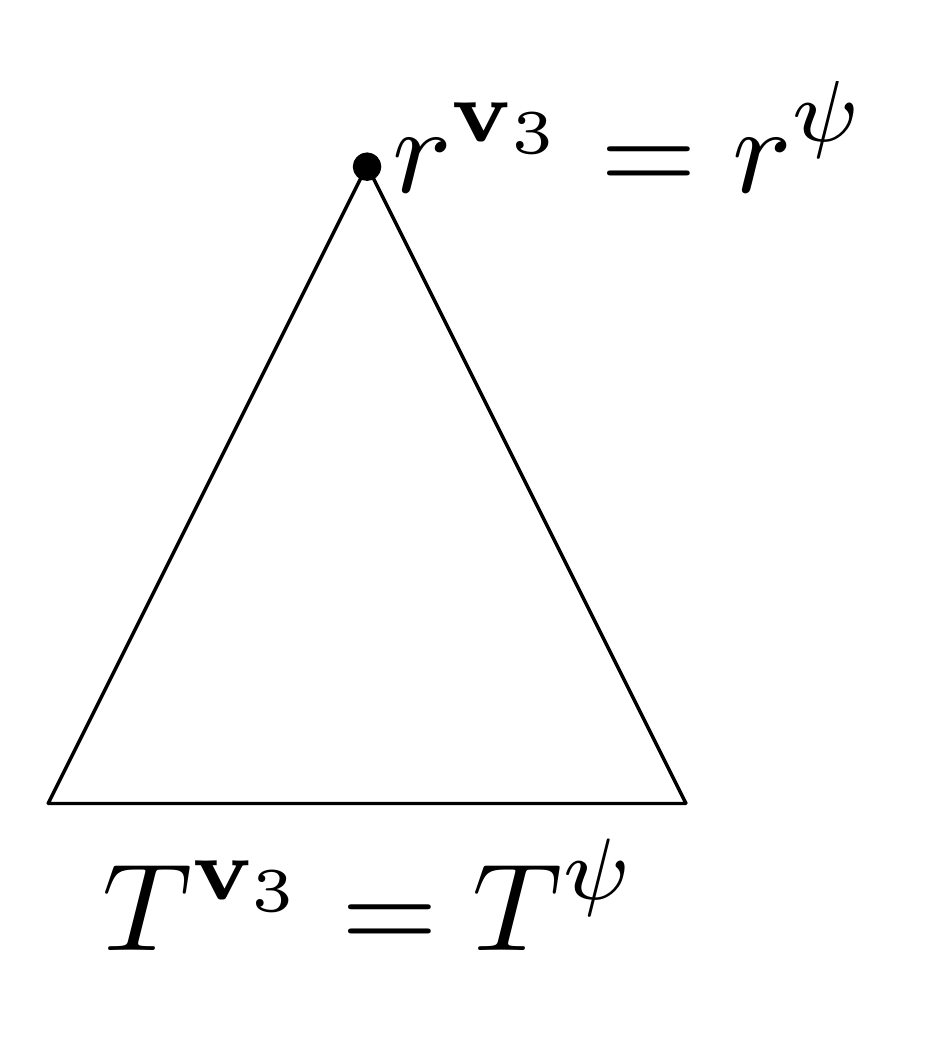}
   \\
	   (a)&(b)&(c) 
   \end{tabular}
   \begin{tabular}{c@{\hskip .2in}c@{\hskip .2in}c}
   \includegraphics[scale=0.25]{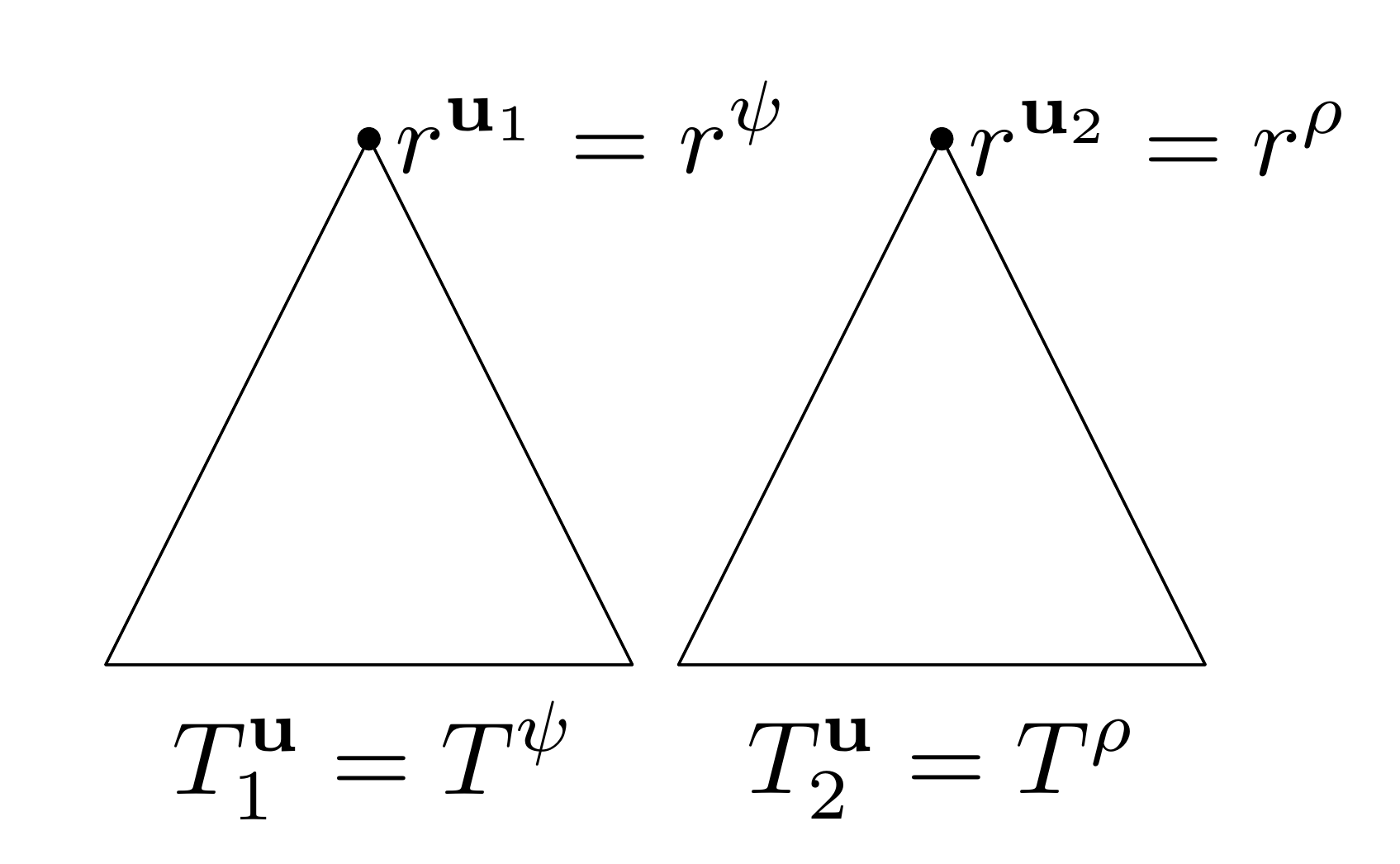}&
   \includegraphics[scale=0.25]{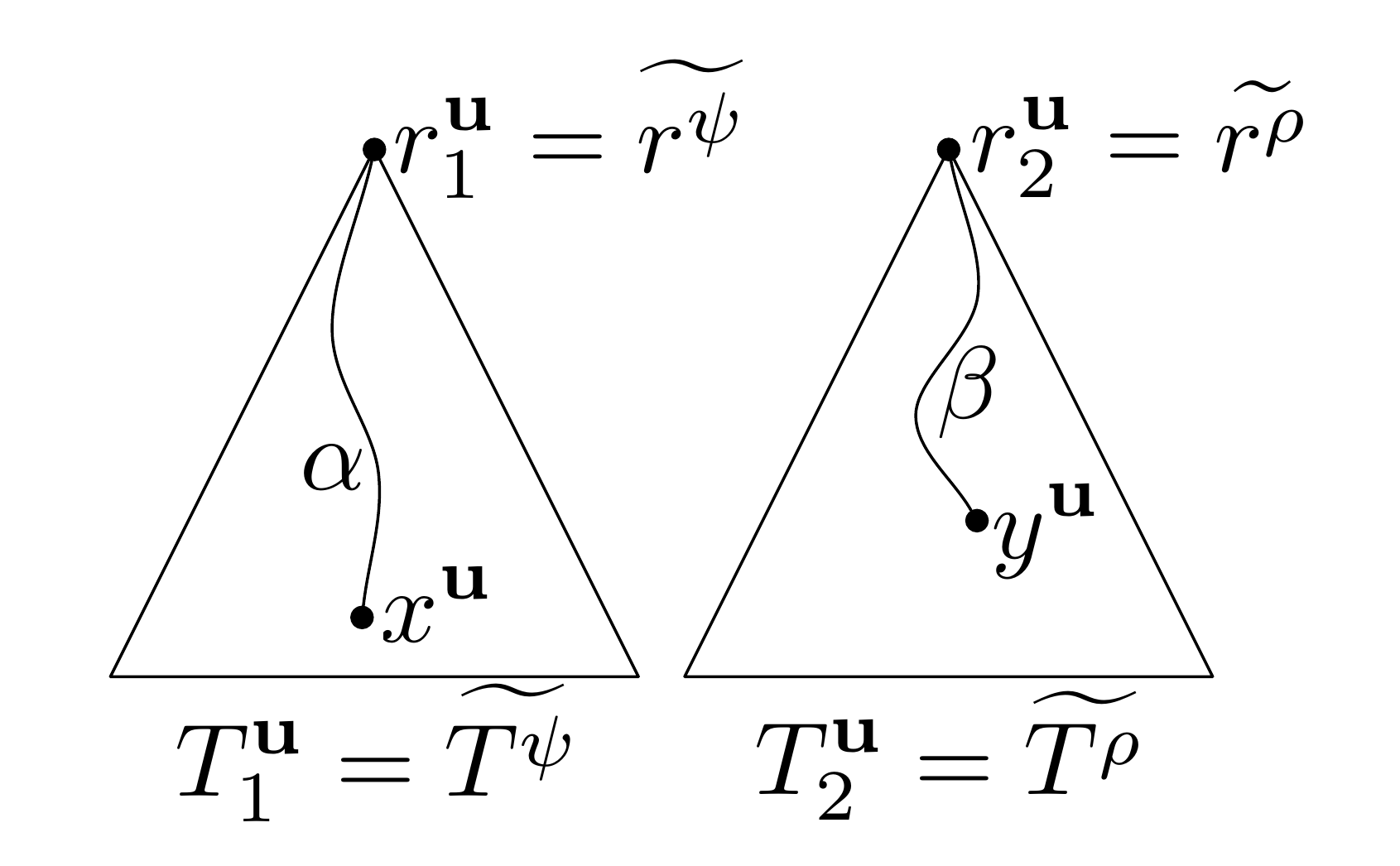}&    
\\(d)&(e)&
   \end{tabular}
   \end{center}
   \caption{Witnesses for (a) ${\bf v_1}=(\psi,\alpha)\in \sisi$; (b) ${\bf v_2}=(\psi,\alpha)\in \sino$; \mbox{(c) ${\bf v_3}=(\psi,\alpha)\in \nosi$}; (d) ${\bf u}=(\psi,\alpha,\rho,\beta)\in {\bf U_1}$; (e) ${\bf u}=(\psi,\alpha,\rho,\beta)\in {\bf U_2}$.
   }   \label{fig:constr-neq}
\end{figure}

\paragraph {\bf Rule 2. Witnesses for \boldmath{${\bf v_2}=(\psi,\alpha)\in \sino$}.} (c.f~\ref{Rule2}) We define a data tree  $\Tt^{\bf v_2}=(T^{\bf v_2},\pi^{\bf v_2})$ with {\rm root} $r^{\bf v_2}$.
By inductive hypothesis, there exists $\Tt^\psi=(T^\psi,\pi^\psi)$, with {\rm root} $r^\psi$ such that $\Tt^\psi,r^\psi\models \psi$.
Define $T^{\bf v_2}$ as $T^\psi$, $\pi^{\bf v_2}$ as $\pi^\psi$, and $r^{\bf v_2}$ as $r^\psi$.
In other words, the {\rm root}ed data tree $(T^{\bf v_2},\pi^{\bf v_2},r^{\bf v_2})$ is just a copy of $(T^\psi,\pi^\psi,r^\psi)$. 
See Figure~\ref{fig:constr-neq}(b).

\paragraph {\bf Rule 3. Witnesses for \boldmath{${\bf v_3}=(\psi,\alpha)\in \nosi$}.} (cf.~\ref{Rule3}) We define a data tree  $\Tt^{\bf v_3}=(T^{\bf v_3},\pi^{\bf v_3})$ with {\rm root} $r^{\bf v_3}$.
By inductive hypothesis, there exists $\Tt^\psi=(T^\psi,\pi^\psi)$, with {\rm root} $r^\psi$ such that $\Tt^\psi,r^\psi\models \psi$.
Define $T^{\bf v_3}$ as $T^\psi$, $\pi^{\bf v_3}$ as $\pi^\psi$, and $r^{\bf v_3}$ as $r^\psi$.
In other words, the {\rm root}ed data tree $(T^{\bf v_3},\pi^{\bf v_3},r^{\bf v_3})$ is just a copy of $(T^\psi,\pi^\psi,r^\psi)$.
See Figure~\ref{fig:constr-neq}(c).

\paragraph {\bf Rule 4. Witnesses for \boldmath{$u=(\psi,\alpha,\rho,\beta)\in {\bf U}$}.} (cf.~\ref{conjuntoU}) \label{NeqRule4} We define data trees $\Tt_1^{\bf u}=(T_1^{\bf u},\pi_1^{\bf u})$ and $\Tt_2^{\bf u}=(T_2^{\bf u},\pi_2^{\bf u})$ with {\rm root}s $r^{\bf u}_1$, $r^{\bf u}_2$ respectively.
%
%
%
%
%

By inductive hypothesis, there exist trees $\Tt^\psi=(T^\psi,\pi^\psi)$ (with {\rm root} $r^\psi$) and $\Tt^\rho=(T^\rho,\pi^\rho)$ (with {\rm root} $r^\rho$) such that $\Tt^\psi,r^\psi\models \psi$ and $\Tt^\rho,r^\rho\models \rho$. 

Now, in order to consider the information given by ${\bf U}$ and its interaction with ${\bf Z}$, we  split ${\bf U}$ into two different subsets:

\begin{itemize}
\item ${\bf U_1}$ is the set of $(\psi,\alpha,\rho,\beta) \in {\bf U}$ for which there are $\gamma, \delta\in P_n$ such that:
\begin{itemize}
\item $(\psi,\gamma,\rho,\delta) \in {\bf Z}$, 
\item $\tup{\gamma = \alpha}$ is a conjunct of $\psi$, 
\item $\tup{\delta = \beta}$ is a conjunct of $\rho$.
\end{itemize}

\item ${\bf U_2}={\bf U} \setminus {\bf U_1}$
\end{itemize}

For ${\bf u}=(\psi,\alpha,\rho,\beta) \in {\bf U_1}$, define $T^{\bf u}_1$ as $T^\psi$, $\pi^{\bf u}_1$ as $\pi^\psi$, $r^{\bf u}_1$ as $r^\psi$ and define $T^{\bf u}_2$ as $T^\rho$, $\pi^{\bf u}_2$ as $\pi^\rho$, $r^{\bf u}_2$ as $r^\rho$. Without loss of generality, we assume that $T^{\bf u}_1$ and $T^{\bf u}_2$ are disjoint. 


In other words, the {\rm root}ed data tree $(T_1^{\bf u},\pi_1^{\bf u}, r_1^{\bf u})$ is just a copy of $(T^\psi,\pi^\psi,r^\psi)$ and the pointed data tree $(T_2^{\bf u},\pi_2^{\bf u},r^{\bf u}_2)$ is a copy of $(T^\rho,\pi^\rho,r^\rho)$. See Figure~\ref{fig:constr-neq}(d). Note that these are the cases in which the satisfaction of $\tup{\dow[\psi]\alpha=\dow[\rho]\beta}$ will be guaranteed by the merging described in~\ref{conjuntoZ}.

For ${\bf u}=(\psi,\alpha,\rho,\beta) \in {\bf U_2}$, in Lemma~\ref{lema clave plus} consider
\begin{align*}
\psi_0&:=\psi
\\
\Tt^{\psi_0}&:=\Tt^{\psi}
\\
\alpha&:=\alpha
\\
\{\beta_1,\dots ,\beta_m\}&:=\{\gamma \in P_{n} \mid \lnot \tup{\dow[\rho]\beta=\dow[\psi]\gamma} \hbox{ is a conjunct of } \varphi\}
\\
\gamma_i&:=\dow[\rho]\beta\mbox{ for all $i=1,\dots ,m$}
\end{align*}

\noindent Then there exist $\widetilde{\Tt^\psi}=(\widetilde{T^\psi},\widetilde{\pi^\psi})$ with {\rm root} $\widetilde{r^\psi}$ and a node $x$ such that:
\begin{itemize}
\item $\widetilde{\Tt^\psi},\widetilde{r^\psi}\models \psi$, 
\item $\widetilde{\Tt^\psi},\widetilde{r^\psi}, x\models \alpha$ 
\item $[x]_{\widetilde{\pi^\psi}}\neq [y]_{\widetilde{\pi^\psi}}$ for all $y$ such that there is $\gamma\in P_n$ with $\widetilde{\Tt^\psi},\widetilde{r^\psi}, y\models \gamma$  and $\lnot \tup{\dow[\rho]\beta = \dow[\psi]\gamma }$ is a conjunct of $\varphi$. 
\end{itemize}
Define $T^{\bf u}_1$ as $\widetilde{T^\psi}$, $\pi^{\bf u}_1$ as $\widetilde{\pi^\psi}$, $r^{\bf u}_1$ as $\widetilde{r^\psi}$ and $x^{\bf u}=x \in T^{\bf u}_1$. Now let  
$$
\left\{\mu_1,\dots,\mu_r \right\}=\left\{\mu \in P_n \mid \mbox{ there exists } \ y\in T^{\bf u}_1 \mbox{ such that } \Tt^{\bf u}_1,r^{\bf u}_1, y\models \mu \mbox{ and } [y]_{\pi^{\bf u}_1}=[x^{\bf u}]_{\pi^{\bf u}_1}\right\}.
$$
Then it follows that $\tup{\dow[\rho]\beta=\dow[\psi]\mu_j}$ is a conjunct of $\varphi$ for all $j=1,\dots,r$. 

In Lemma~\ref{lema clave plus}, consider
\begin{align*}
\psi_0&:=\rho
\\
\Tt^{\psi_0}&:=\Tt^{\rho}
\\
\alpha&:=\beta
\\
\{\beta_1,\dots ,\beta_m\}&:=\{\delta \in P_{n} \mid \exists j=1,\dots ,r \hbox{ with } \lnot \tup{\dow[\rho]\delta=\dow[\psi]\mu_j} \hbox{ is a conjunct of } \varphi\}
\\
\gamma_i&:=\dow[\psi]\mu_j \mbox{ for $j=1,\dots r$ such that $\tup{\dow[\rho]\beta_i=\dow[\psi]\mu_j}$ is a conjunct of $\varphi$}
\end{align*}

\noindent Then there exist a tree $\widetilde{\Tt^{\rho}}=(\widetilde{T^{\rho}},\widetilde{\pi^{ \rho}})$ with {\rm root} $\widetilde{r^{\rho}}$ and a node $y$ such that
\begin{itemize}
\item $\widetilde{\Tt^{\rho}}, \widetilde{r^{\rho}}\models \rho$,
\item $\widetilde{\Tt^{\rho}}, \widetilde{r^{\rho}}, y\models \beta$, 
\item $[y]_{\widetilde{\pi^{\rho}}}\neq [z]_{\widetilde{\pi^{\rho}}}$ for all $z$ such that there is $\delta\in P_n$ and $j=1,\dots,r$ with $\widetilde{\Tt^\rho},\widetilde{r^\rho}, z\models \delta$ and $\lnot \tup{\dow[\rho]\delta=\dow[\psi]\mu_j}$ is a conjunct of $\varphi$.
\end{itemize}
Define $T^{\bf u}_2$ as $\widetilde{T^\rho}$, $\pi^{\bf u}_2$ as $\widetilde{\pi^\rho}$, $r^{\bf u}_2$ as $\widetilde{r^\rho}$ and $y^{\bf u}=y$. Without loss of generality, we assume that $T^{\bf u}_1$ and $T^{\bf u}_2$ are disjoint.

%
In other words, the {\rm root}ed data tree $(T_1^{\bf u},\pi^{\bf u}\restr{T_1^{\bf u}} ,r_1^{\bf u})$ is just a copy of $(\widetilde{T^\psi},\widetilde{\pi^\psi},\widetilde{r^\psi})$, with a special node named $x^{\bf u}$ which satisfies $\Tt_1^{\bf u},r_1^{\bf u},x^{\bf u}\models\alpha$. Analogously, the pointed data tree \mbox{$(T_2^{\bf u},\pi^{\bf u}\restr{ T_2^{\bf u}},r^{\bf u}_2)$} is a copy of $(\widetilde{T^\rho},\widetilde{\pi^\rho},\widetilde{r^\rho})$, with a special node named $y^{\bf u}$ which satisfies \  $\Tt_2^{\bf u},r_2^{\bf u},y^{\bf u}\models\beta$.  
See Figure~\ref{fig:constr-neq}(e).

Notice that this rule differs from Rule 2 of~\S\ref{construccion} in the fact that we do not merge the classes of $x^{\bf u}$ and $y^{\bf u}$ yet. We will perform that merging only at the end of the construction. This is not really important and we could have merged the classes at this step; the reason for doing it at the end is only a technical issue. The proof of Fact~\ref{factnocambialaparticion} will be easier to understand this way.    

The following remark will be used later to prove that $\varphi$ is indeed satisfied in the constructed tree. Its proof is omitted since it is analogous to the proof of Remark~\ref{lemacasosindistinto}:

\begin{remark}\label{lema}
Let $(\psi,\alpha,\rho,\beta)\in {\bf U_2}$. If $\lnot \tup{\dow[\psi]\mu = \dow[\rho]\delta}$ is a conjunct of $\varphi$, then $[y^{\bf u}]_{\pi^{\bf u}_2}\neq [y]_{\pi^{\bf u}_2}$ for all $y$ such that $\Tt_2^{\bf u}, r_2^{\bf u}, y \models \delta$ or $[x^{\bf u}]_{\pi^{\bf u}_1}\neq [x]_{\pi^{\bf u}_1}$ for all $x$ such that $\Tt_1^{\bf u}, r_1^{\bf u}, x \models \mu$.
%
%

\qed
\end{remark}


\paragraph{The rooted data tree \boldmath{$(T^\varphi,\pi^\varphi,r^\varphi)$}.} As shown in Figure~\ref{fig:constr-neq-2}, now we define $T^\varphi$, using our Rules, as the tree which consists of a {\rm root} $r^\varphi$ with label $a\in\A$ if $a$ is a conjunct of $\varphi$, and with children 
$$
(T^{\bf v_1}_1)_{{\bf v_1}\in \sisi},
(T^{\bf v_1}_2)_{{\bf v_1}\in \sisi},
(T^{\bf v_2})_{{\bf v_2}\in \sino},
(T^{\bf v_3})_{{\bf v_3}\in \nosi},
(T^{\bf u}_1)_{{\bf u}\in {\bf U}},
(T^{\bf u}_2)_{{\bf u}\in {\bf U}}.
$$
%
%
As a first step we provisionally define $\widetilde{\pi^\varphi}$ over $T^\varphi$ by
$$
\widetilde{\pi^\varphi}= \{\{ r^\varphi\}\} \cup \bigcup_{{\bf v_1}\in \sisi}(\pi_1^{\bf v_1}\cup \pi_2^{\bf v_1}) \cup
\bigcup_{{\bf v_2}\in \sino}\pi^{\bf v_2}
\cup \bigcup_{{\bf v_3}\in \nosi}\pi^{\bf v_3}
\cup \bigcup_{{\bf u}\in {\bf U}} (\pi_1^{\bf u}\cup \pi_2^{\bf u})
$$
\begin{figure}[ht]
   \begin{center}
   \includegraphics[scale=0.25]{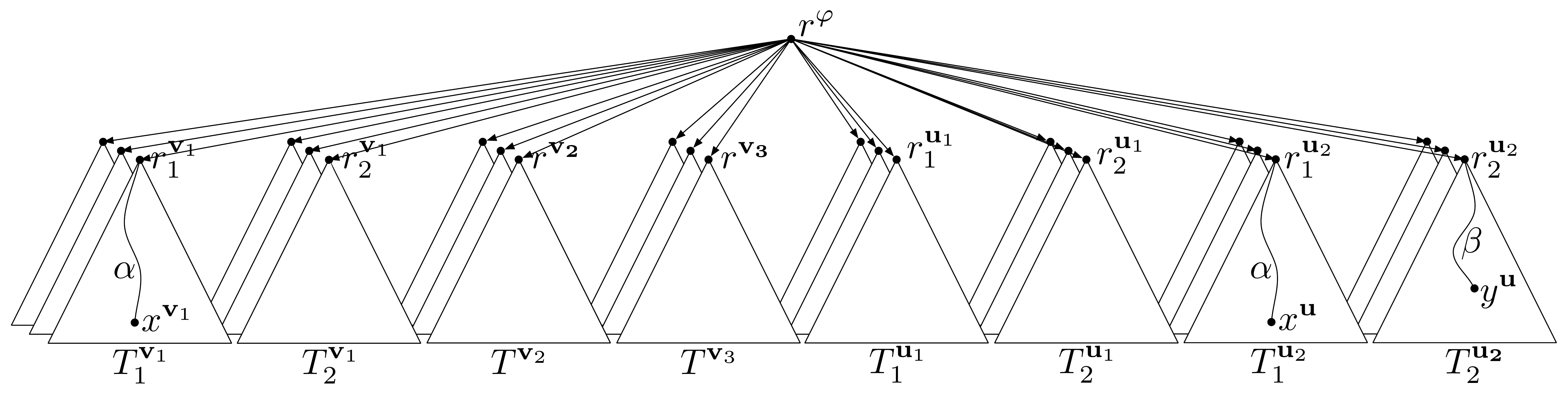}
   \end{center}
   \caption{The tree $T^\varphi$ (without any partition yet).}
   \label{fig:constr-neq-2}
\end{figure}
It is important to notice that, up to this point in the construction, the tree hanging from each child of the {\rm root} preserves its original partition. 

In order to consider the information given by ${\bf Z}$ (cf.~\ref{conjuntoZ}), we split $\nosi$ into two subsets: 
\begin{align*}
{\bf V_{\lnot =, \neq}'}&=\left\{(\psi,\alpha)\in \nosi \mid \mbox{ for all }\ (\rho,\beta)\in \nosi,\  (\psi,\alpha,\rho,\beta)\not\in {\bf Z} \right\},
\\
{\bf V_{\lnot =, \neq}''}&=\nosi \setminus {\bf V_{\lnot =, \neq}'}.
\end{align*}
The following property of the set ${\bf V''_{\lnot =, \neq}}$ will be used to prove that $\varphi$ is indeed satisfied at the constructed tree:

\begin{lemma}\label{lema 0'}
Let $(\theta,\delta), (\theta,\delta')\in {\bf V_{\lnot =,\neq}}$. Suppose that $(\theta,\delta)\in {\bf V''_{\lnot =, \neq}}$ and $\lnot \tup{\delta' \neq \delta'}$ and $\tup{\delta=\delta'}$ are conjuncts of $\theta$. Then $(\theta,\delta,\theta,\delta')\in {\bf Z}$. 
\end{lemma}

\begin{proof}
By  \neqaxeight, $\lnot \tup{\dow[\theta]\delta \neq \dow[\theta]\delta}$ is a conjunct of $\varphi$. By  \eqax{7} plus {\bf Der21} of Fact~\ref{fact boolean}, $\tup{\dow[\theta]\delta = \dow[\theta]\delta'}$  is a conjunct of $\varphi$. If we suppose that $\tup{\dow[\theta]\delta' \neq \dow[\theta]\delta'}$ is a conjunct of $\varphi$, by Lemma~\ref{paraverif} 
, we have that $\lnot \tup{\delta=\delta'}$ is a conjunct of $\theta$ which is a contradiction. Then we can assume that $\lnot \tup{\dow[\theta]\delta' \neq \dow[\theta]\delta'}$ is a conjunct of $\varphi$ and so we can conclude from  \neqaxeight that $\lnot \tup{\dow[\theta]\delta \neq \dow[\theta]\delta'}$ is a conjunct of $\varphi$. Then we have that $(\theta,\delta,\theta,\delta')\in {\bf Z}$. 
\end{proof}

As a particular case of Lemma~\ref{lema 0'}, we have: 
 
\begin{remark}
Let $(\theta,\delta), (\theta,\delta')\in {\bf V_{\lnot =,\neq}}$. Suppose that $(\theta,\delta), (\theta,\delta')\in {\bf V''_{\lnot =, \neq}}$ and $\tup{\delta=\delta'}$ is a conjunct of $\theta$. Then $(\theta,\delta,\theta,\delta')\in {\bf Z}$. 
\end{remark}

\begin{proof}
 Use  \neqaxfifteen plus {\bf Der21} of Fact~\ref{fact boolean} and  \neqaxeight.
\end{proof}

We classify the elements of ${\bf V_{\lnot =, \neq}''}$ according to the following equivalence relation:
$$
[(\psi,\alpha)]=[(\rho,\beta)] \mbox{\quad iff\quad} (\psi,\alpha,\rho,\beta) \in {\bf Z}.
$$
Observe that this relation is reflexive by  \neqaxeight, it is clearly symmetric and it is transitive by Lemma~\ref{lema c}.
We name the equivalence classes $A_1,\dots,A_m$. 
We define $\widehat{\pi^{\varphi}}$ over $T^{\varphi}$ taking into account the information given by $\sisi, \sino$ and ${\bf Z}$.
%
%
$\widehat{\pi^{\varphi}}$ is the smallest equivalence relation containing $\widetilde{\pi^{\varphi}}$ such that:

\begin{itemize} 
\item $[x^{\bf v_1}]_{\widehat{\pi^{\varphi}}}=[r^{\varphi}]_{\widehat{\pi^{\varphi}}}$ for all ${\bf v_1}\in \sisi$,

\item $[x]_{\widehat{\pi^{\varphi}}}=[r^{\varphi}]_{\widehat{\pi^{\varphi}}}$ for all $x\in M$,

\item For all $i=1,\dots ,m$ $[x]_{\widehat{\pi^{\varphi}}}=[y]_{\widehat{\pi^{\varphi}}}$ for all $x,y \in L_i$
\end{itemize}
where 
\begin{align*}
M &= \{x \mid \hbox{ there exists } (\psi, \alpha) \in \sino \hbox{ and a child } z \hbox{ of } r^{\varphi} \hbox{ such that} 
\\
&\qquad T^\varphi, \widetilde{\pi^\varphi}, z \models \psi \hbox{ and } T^\varphi, \widetilde{\pi^\varphi},z,x\models\alpha \}
\\
L_i &= \{x \mid \hbox{ there exists } (\psi, \alpha)\in A_i \hbox{ and a child } z \hbox{ of } r^{\varphi} \hbox{ such that}
\\
&\qquad T^\varphi, \widetilde{\pi^\varphi}, z \models \psi \hbox{ and } T^\varphi, \widetilde{\pi^\varphi},z,x\models\alpha \} 
\end{align*}
 for all $i=1,\dots ,m$.
  
In the previous ``gluing'', we forced our model to satisfy all diamonds of the form \mbox{$\tup{\eps=\dow[\psi]\alpha}$}, $\lnot\tup{\eps\neq\dow[\psi]\alpha}$ and $\lnot \tup{\dow[\psi]\alpha\neq\dow[\rho]\beta}$ that need to be forced.

\bigskip

It is important to notice that, up to here, the tree hanging from each child of the {\rm root} still preserves its partition:

\begin{fact}\label{factnocambialaparticion1}
The partition restricted to the trees $T_1^{\bf v_1}$, $T_2^{\bf v_1}$ for ${\bf v_1\in \sisi}$, the partition restricted to the trees $T^{\bf v_2}$ for ${\bf v_2\in \sino}$, the partition restricted to the trees $T^{\bf v_3}$ for ${\bf v_3\in \nosi}$ and the partition restricted to the trees $T^{\bf u}_1$ and $T^{\bf u}_2$ for ${\bf u\in U}$ remain unchanged. More formally:

\begin{itemize}
\item For each ${\bf v_1}=(\psi,\alpha)\in {\bf \sisi}$ and $i\in\{1,2\}$, we have $\widehat{\pi^\varphi}\restr{T_i^{\bf v_1}}=\pi_i^{\bf v_1}$.

\item For each ${\bf v_2}=(\psi,\alpha)\in {\bf \sino}$, we have $\widehat{\pi^\varphi}\restr{T^{\bf v_2}}=\pi^{\bf v_2}$.

\item For each ${\bf v_3}=(\psi,\alpha)\in {\bf \nosi}$, we have $\widehat{\pi^\varphi}\restr{T^{\bf v_3}}=\pi^{\bf v_3}$.

\item For each ${\bf u}=(\psi,\alpha,\rho,\beta)\in {\bf U}$ and $i\in\{1,2\}$, we have $\widehat{\pi^\varphi}\restr{T^{\bf u}_i}=\pi^{\bf u}_i$.
\end{itemize}

\end{fact}

\begin{proof}
 We give a sketch of the proof and leave the details to the reader.
 If we think we have three kinds of ``gluings'', \emph{${\rm root}_{=,\neq}$}-kind, \emph{${\rm root}_{=,\lnot \neq}$}-kind and \emph{$Z$}-kind, then the way in which two equivalence classes in the same subtree can (hypothetically) be glued together is by a sequence of these gluings. The examples displayed in Figure~\ref{figurefact1} shows that in (a), the classes of nodes $x$ and $y$ were glued together by a sequence of the form \emph{${\rm root}_{=,\neq}$}-\emph{${\rm root}_{=,\lnot \neq}$}; in (b), the classes of nodes $x$ and $y$ were glued together by a sequence of the form \emph{${\rm root}_{=,\neq}$}-\emph{${\rm root}_{=,\lnot \neq}$}-\emph{$Z$}.

\begin{figure}[ht]
   \begin{center}
   \begin{tabular}{c@{\hskip .2in}c@{\hskip .2in}c}
   \includegraphics[scale=0.25]{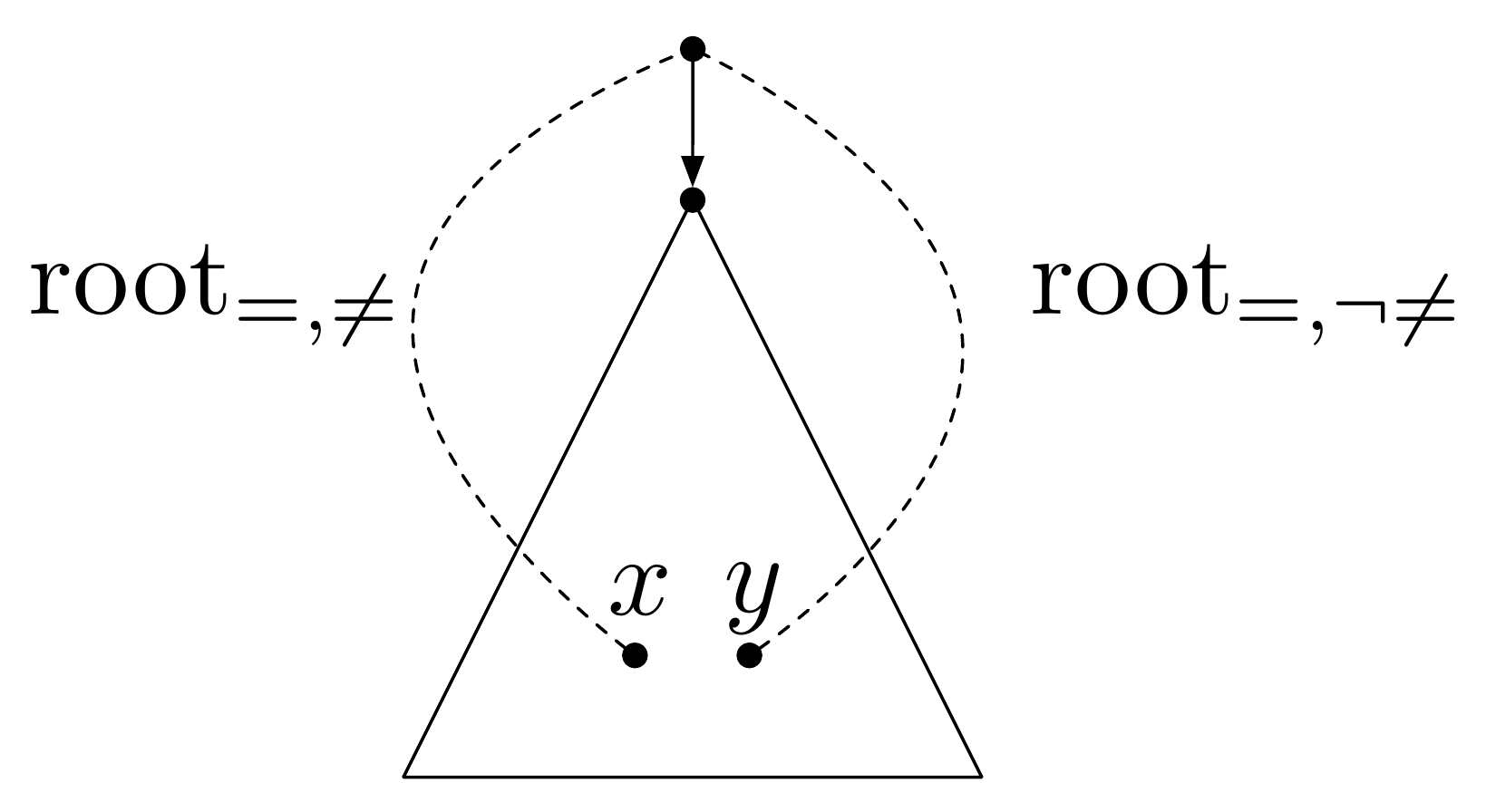} &
   \includegraphics[scale=0.25]{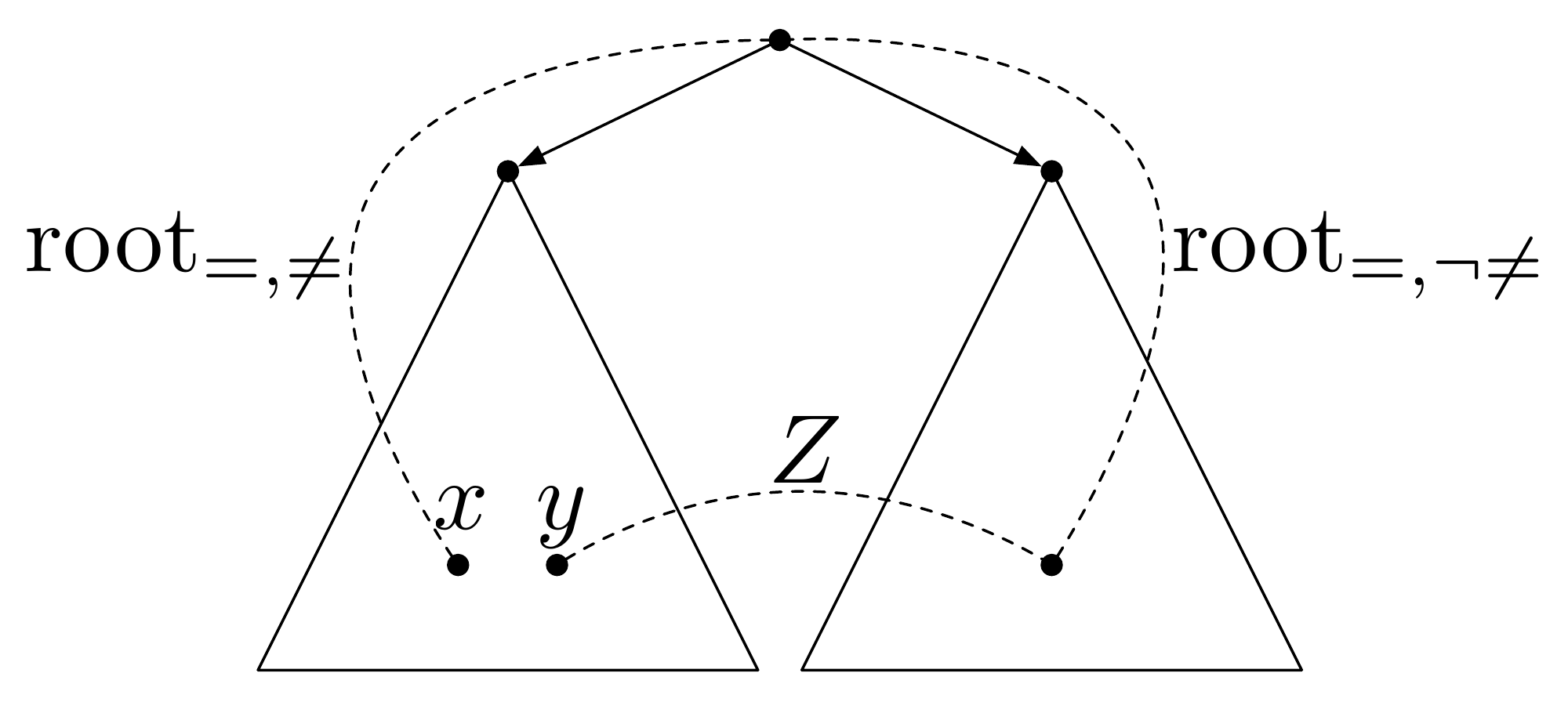}
   \\
	   (a)&(b)
   \end{tabular}
   \end{center}
   \caption{Examples of (hypothetical) ``gluings''.}\label{figurefact1}
\end{figure}

We give a list of the ingredients for the complete proof.
 
 \begin{itemize}
 
 \item By Rule 1, every witness for $\tup{\dow[\psi]\alpha}$ with $(\psi,\alpha)\in \nosi$ in $\Tt^{\bf v_1}_1$ is in a different class (according to $\widetilde{\pi^{\varphi}}$) than $x^{\bf v_1}$ for all ${\bf v_1} \in \sisi$. Thus we do not have sequences containing \emph{${\rm root}_{=,\neq}$}-\emph{$Z$} or \emph{$Z$}-\emph{${\rm root}_{=,\neq}$}.   
 
  \item Lemma~\ref{lema 1} implies that every witness for $\tup{\dow[\psi]\alpha}$ with $(\psi,\alpha)\in \sino$ and every witness for $\tup{\dow[\psi]\beta}$ with $(\psi,\beta)\in \nosi$ in the same subtree belong to different classes in that subtree. As a particular case, every $x\in M$ and $y\in L_i$ in the same subtree belong to different classes. Thus we do not have sequences containing \emph{${\rm root}_{=,\lnot \neq}$}-\emph{$Z$} or \emph{$Z$}-\emph{${\rm root}_{=,\lnot \neq}$}. 
 
  \item Since we use a different copy at each application of Rule 1, we do not have sequences starting and ending with \emph{${\rm root}_{=,\neq}$}.
 
 \item Lemma~\ref{lema 1} implies that every witness for $\tup{\dow[\psi]\alpha}$ with $(\psi,\alpha)\in \sino$ in the same subtree belong to the same equivalence class in that subtree. Thus we do not have to worry about sequences starting and ending with \emph{${\rm root}_{=,\lnot \neq}$} because this kind of sequences do not glue different classes.
 
  \item By Rule 1, every $x^{\bf v_1}$ is in the same class that every witness in the same subtree of $\tup{\dow[\psi]\alpha}$ with $(\psi,\alpha)\in \sino$. Thus we do not  have to worry about sequences starting with \emph{${\rm root}_{=, \neq}$} and ending with \emph{${\rm root}_{=,\lnot \neq}$} (or vice versa) because this kind of sequences do not  glue different classes.
   
 Combining the previous items, it only remains to consider sequences of only \emph{$Z$}-kind gluings.  
   
  \item If $x,x'\in L_i$ in the same subtree, a very simple derivation involving  \neqaxfifteen, shows that $[x]_{\widetilde{\pi^{\varphi}}}=[x']_{\widetilde{\pi^{\varphi}}}$. Thus we do not  have to worry about sequences of the form \emph{$Z$} (just one \emph{$Z$}-kind gluing).   
  
  \item By Lemma~\ref{lema 1'} below, we do not  have to worry about longer sequences of all \emph{$Z$}-kind gluings. 
   \end{itemize}
This concludes the proof of the Fact.
\end{proof}

\begin{lemma}\label{lema 1'}
 Let $\psi,\theta_0,\dots,\theta_m \in N_n$, $\alpha, \beta, \delta_0, \delta_0',\dots,\delta_m,\delta_m'\in P_n$, $x,x',y,y'\in T^\psi,   x_0,y_0 \in T^{\theta_0},\dots,\ x_m,y_m \in T^{\theta_m}$ such that (see Figure~\ref{fig:lemma 1'}): 
\begin{itemize}
\item $[x]_{\pi^\psi}=[x']_{\pi^\psi}$, $[y]_{\pi^\psi}=[y']_{\pi^\psi}$, $T^{\psi}, r^{\psi}, x' \models \alpha$, $T^{\psi}, r^{\psi}, y' \models \beta$,

\item $[x_i]_{\pi^{\theta_i}}=[y_i]_{\pi^{\theta_i}}$, $T^{\theta_i}, r^{\theta_i}, x_i \models \delta_i$, $T^{\theta_i}, r^{\theta_i}, y_i \models \delta_i'$ for $i=0\dots m$, and

\item $(\theta_0,\delta_0,\psi,\alpha) \in {\bf Z}$, $(\theta_i,\delta_i,\theta_{i-1},\delta_{i-1}') \in {\bf Z}$ for $i=1\dots m$,  $(\theta_m,\delta_m',\psi,\beta) \in {\bf Z}$.
\end{itemize}
Then $[x]_{\pi^{\psi}}=[y]_{\pi^{\psi}}$.


(Notation: For $\rho \in N_n$, we use $\Tt^{\rho}=(T^{\rho},\pi^{\rho})$ with root $r^{\rho}$ to denote {\em any} tree in which $\rho$ is satisfiable, namely the one given by inductive hypothesis, or the modified one $\widetilde{\Tt^{\rho}}$.)
\end{lemma}

\begin{figure}[ht]
   \begin{center}
   \includegraphics[scale=0.25]{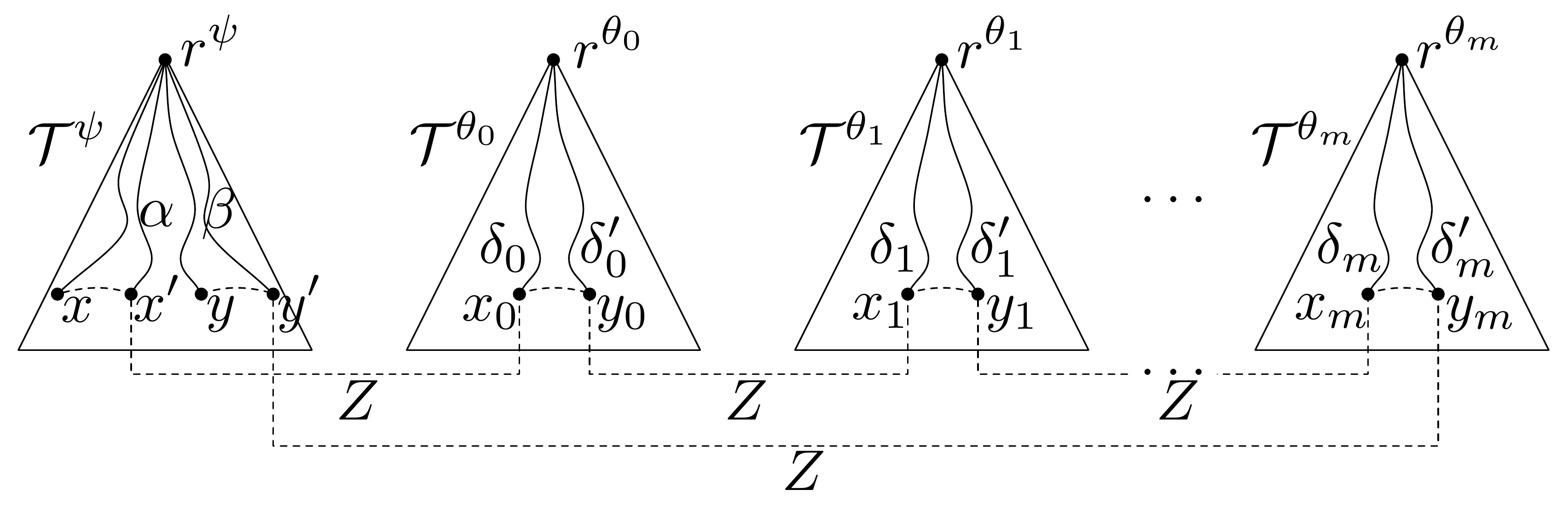}
   \end{center}
   \caption{The hypothesis of Lemma~\ref{lema 1'}}   \label{fig:lemma 1'}
\end{figure}

\begin{proof}
 Observe that, by Lemma~\ref{lema 0'} plus Lemma~\ref{lema c}, $(\psi, \alpha,\psi,\beta) \in {\bf Z}$. Then, by  \neqaxeight plus Lemma~\ref{lema 1}, $\lnot \tup{\alpha \neq \beta}$ is a conjunct of $\psi$ and so $[x]_{\pi^{\psi}}=[y]_{\pi^{\psi}}$. 
\end{proof}

%
%
%
%
%


Finally, define $\pi^\varphi$ over $T^\varphi$ by
\begin{eqnarray*}
\pi^\varphi&=& \left(\widehat{\pi^{\varphi}}\setminus \left(\{[x^{\bf u}]_{\widehat{\pi^{\varphi}}}\}_{{\bf u}\in {\bf U_2}}
\cup
\{[y^{\bf u}]_{\widehat{\pi^{\varphi}}}\}_{{\bf u}\in {\bf U_2}}\right)\right) \cup 
\bigcup_{{\bf u}\in {\bf U_2}}\{[x^{\bf u}]_{\widehat{\pi^{\varphi}}} \cup [y^{\bf u}]_{\widehat{\pi^{\varphi}}}\}.
\end{eqnarray*}
%
%
In other words, $T^\varphi$ has a {\rm root}, named $r^\varphi$, and children 
$$
(T_1^{\bf v_1})_{{\bf v_1}\in \sisi}, (T_2^{\bf v_1})_{{\bf v_1}\in \sisi}, (T^{\bf v_2})_{{\bf v_2}\in \sino}, (T^{\bf v_3})_{{\bf v_3}\in \nosi}, (T_1^{\bf u})_{{\bf u}\in {\bf U}}, (T_2^{\bf u})_{{\bf u}\in {\bf U}}.
$$
Each of these children is the {\rm root} of its corresponding tree inside $T^\varphi$ as defined above. All these subtrees are disjoint, and $\pi^\varphi$ is defined as the disjoint union of the partitions {\em with the exception} that we put into the same class:

\begin{itemize}
\item the nodes $r^\varphi$, $(x^{\bf{v_1}})_{{\bf v_1}\in\sisi}$ and every witness of $\tup{\dow[\psi]\alpha}$ with $(\psi,\alpha)\in \sino$, 
\item a witness for $\tup{\dow[\psi]\alpha}$ and a witness for $\tup{\dow[\rho]\beta}$ if $(\psi, \alpha, \rho,\beta) \in {\bf U_2}$,
\item every pair of witnesses of $\tup{\dow[\psi]\alpha}$ and $\tup{\dow[\rho]\beta}$ respectively with $(\psi, \alpha, \rho, \beta) \in {\bf Z}$.
\end{itemize}
In the previous gluing, we forced our model to satisfy all diamonds of the form $\tup{\dow[\psi]\alpha = \dow[\rho]\beta}$ that need to be forced.
%

\bigskip

The following Fact is key to prove that $\varphi$ is satisfied in $\Tt^\varphi$:

\begin{fact}\label{factnocambialaparticion}
The partition restricted to the trees $T_1^{\bf v_1}$, $T_2^{\bf v_1}$ for ${\bf v_1\in \sisi}$, the partition restricted to the trees $T^{\bf v_2}$ for ${\bf v_2\in \sino}$, the partition restricted to the trees $T^{\bf v_3}$ for ${\bf v_3\in \nosi}$ and the partition restricted to the trees $T^{\bf u}_1$ and $T^{\bf u}_2$ for ${\bf u\in U}$ remain unchanged. More formally:
\begin{itemize}
\item For each ${\bf v_1}=(\psi,\alpha)\in {\bf \sisi}$ and $i\in\{1,2\}$, we have $\pi^\varphi\restr{T_i^{\bf v_1}}=\pi_i^{\bf v_1}$.

\item For each ${\bf v_2}=(\psi,\alpha)\in {\bf \sino}$, we have $\pi^\varphi\restr{T^{\bf v_2}}=\pi^{\bf v_2}$.

\item For each ${\bf v_3}=(\psi,\alpha)\in {\bf \nosi}$, we have $\pi^\varphi\restr{T^{\bf v_3}}=\pi^{\bf v_3}$.

\item For each ${\bf u}=(\psi,\alpha,\rho,\beta)\in {\bf U}$ and $i\in\{1,2\}$, we have $\pi^\varphi\restr{T^{\bf u}_i}=\pi^{\bf u}_i$.
\end{itemize}
\end{fact}
\begin{proof}
We give a guide for the proof and we leave the details to the reader.

Now think that we have four kinds of ``gluings'', \emph{${\rm root}_{=,\neq}$}-kind, \emph{${\rm root}_{=,\lnot \neq}$}-kind, \emph{$Z$}-kind and \emph{$U_2$}-kind, then the way in which two equivalence classes in the same subtree can (hypothetically) be glued together is by a sequence of these gluings. In the example displayed in Figure~\ref{figurefact} the classes of nodes $x$ and $y$ were glued together by a sequence of the form \emph{$Z$}-\emph{$Z$}-\emph{$U_2$}.

\begin{figure}[ht]
   \begin{center}
   \includegraphics[scale=0.25]{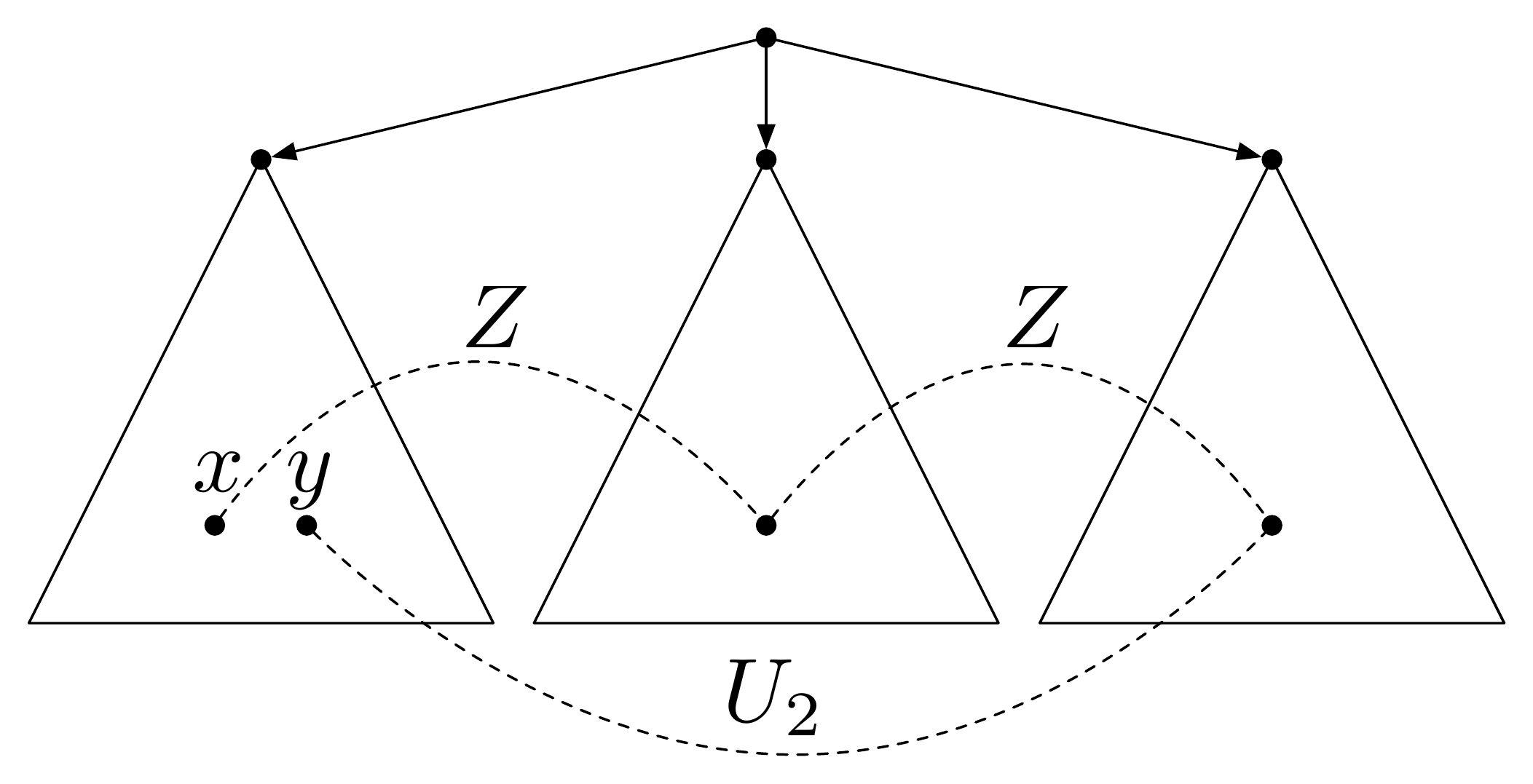}
   \end{center}
   \caption{Example of (hypothetical) ``gluing''.}\label{figurefact}
\end{figure}

\medskip
We give a list of the ingredients for the complete proof.

\begin{itemize}

\item We have already observed that the same assertions hold if we change $\pi^{\varphi}$ for $\widehat{\pi^{\varphi}}$ so we are only interested in sequences that involve some gluing of kind \emph{$U_2$}. Moreover, we can assume all the observations made in the proof of Fact~\ref{factnocambialaparticion1}.





\item The fact that $x^{\bf v_1}$ and $x^{\bf u}$ (or $y^{\bf u}$) are always in different subtrees tells us that we do not  have sequences containing \emph{${\rm root}_{=,\neq}$}-\emph{$U_2$} or \emph{$U_2$}-\emph{${\rm root}_{=,\neq}$}.

 \item Lemma~\ref{lema 1} implies that every witness for $\tup{\dow[\psi]\alpha}$ with $(\psi,\alpha)\in \sino$ and every witness for $\tup{\dow[\psi]\beta}$ with $(\psi,\beta)\in \nosi$ in the same subtree belong to different classes in that subtree. Thus we do not  have sequences containing \emph{${\rm root}_{=,\lnot \neq}$}-\emph{$U_2$} or \emph{$U_2$}-\emph{${\rm root}_{=,\lnot \neq}$} coming from ${\bf u}=(\psi,\alpha,\rho,\beta)\in U_2$ with $(\psi,\alpha), (\rho,\beta)\in \nosi$. Besides, suppose that we have one of those sequences coming from ${\bf u}=(\psi,\alpha,\rho,\beta)\in U_2$ with $(\psi,\alpha) \in \sisi, (\rho,\beta)\in \nosi$ (the symmetric case is analogous) and $(\psi,\mu) \in \sino$. Then, by the consistency of $\varphi$ plus \neqaxfour, we can conclude that $\lnot\tup{\dow[\psi]\mu =\dow[\rho]\beta}$ is a conjunct of $\varphi$. This gives us a contradiction by Remark~\ref{lema}. Thus we do not  have sequences containing \emph{${\rm root}_{=,\lnot \neq}$}-\emph{$U_2$} or \emph{$U_2$}-\
emph{${\rm root}_{=,\lnot \neq}$} at all.  

\item By Lemma~\ref{lema 0'} plus Lemma~\ref{lema c}, we can reduce sequences with two consecutive \emph{$Z$}-kind gluings to sequences not having two consecutive \emph{$Z$}-kind gluings.

\item Since we use new subtrees for each ${\bf u} \in {\bf U_2}$, we cannot have sequences containing \emph{$U_2$}-\emph{$U_2$} neither sequences starting and ending with \emph{$U_2$}.  

\item By Lemma~\ref{no cambia la particion} below, we cannot have sequences that alternate \emph{$Z$}-kind gluings with \emph{$U_2$}-kind gluings. 

\item One can think that the gluing of the classes $[x^{\bf u}]_{\widehat{\pi^{\varphi}}}$ and $[y^{\bf u}]_{\widehat{\pi^{\varphi}}}$ is made one at a time since they are finite.  
\end{itemize}
This concludes the proof of the Fact.
\end{proof}

\begin{lemma}\label{no cambia la particion}
Let $\psi, \theta_0,\dots,\theta_m \in N_n$, $\alpha, \beta, \delta_0, \delta_0',\dots,\delta_m,\delta_m'\in P_n$, $x,x',y,y'\in T^\psi,   x_0,y_0 \in T^{\theta_0},\dots, \ x_m,y_m \in T^{\theta_m}$. The following conditions (see Figure~\ref{fig:no cambia la particion}) cannot be satisfied all at the same time: 
\begin{itemize}
\item $[x]_{\pi^\psi}=[x']_{\pi^\psi}$, $[y]_{\pi^\psi}=[y']_{\pi^\psi}$, $T^{\psi}, r^{\psi}, x' \models \alpha$, $T^{\psi}, r^{\psi}, y' \models \beta$,

\item $[x_i]_{\pi^{\theta_i}}=[y_i]_{\pi^{\theta_i}}$, $T^{\theta_i}, r^{\theta_i}, x_i \models \delta_i$, $T^{\theta_i}, r^{\theta_i}, y_i \models \delta_i'$ for $i=0\dots m$, 

\item $(\theta_0,\delta_0,\psi,\alpha) \in {\bf Z}$, 

\item 
for $i=1\dots m$, $(\theta_i,\delta_i,\theta_{i-1},\delta_{i-1}') \in\begin{cases}{\bf U_2}&\mbox{if $i$ is odd}\\{\bf Z}&\mbox{otherwise}\end{cases}$

\item $
(\theta_m,\delta'_m,\psi,\beta) 
\in \begin{cases}{\bf Z}&\mbox{if $m$ is odd}\\{\bf U_2}&\mbox{otherwise}\end{cases}
$
\end{itemize}

(Notation: For $\rho \in N_n$, we use $\Tt^{\rho}=(T^{\rho},\pi^{\rho})$ with root $r^{\rho}$ to denote {\em any} tree in which $\rho$ is satisfiable, namely the one given by inductive hypothesis, or the modified one $\widetilde{\Tt^{\rho}}$.)\end{lemma}
 
\begin{figure}[ht]
   \begin{center}
   \includegraphics[scale=0.25]{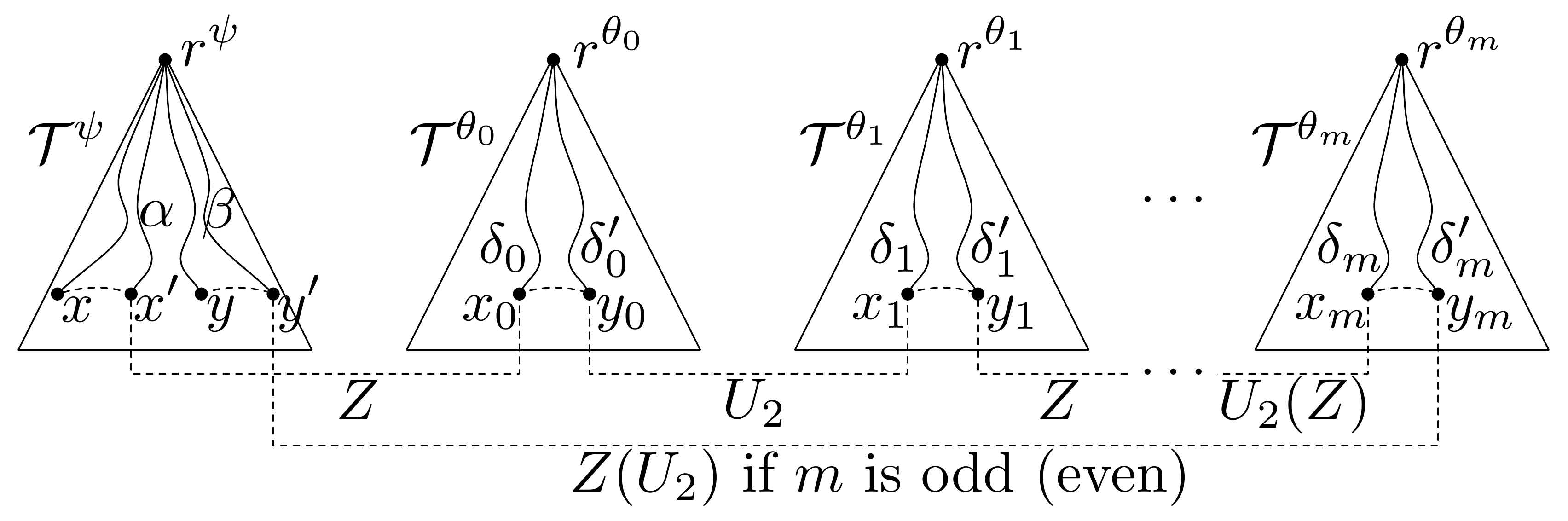}
   \end{center}
   \caption{The hypothesis of Lemma~\ref{no cambia la particion}}   \label{fig:no cambia la particion}
\end{figure}

 \begin{proof}
  We proceed by induction on $m$:
 \begin{itemize}
  \item Case $m=0$ (see Figure~\ref{fig:8and9}(a)):

Since $(\psi,\beta,\theta_0,\delta'_0)\in {\bf U_2}$, $(\psi,\alpha,\theta_0,\delta_0)\in {\bf Z}$ and $\tup{\delta_0=\delta'_0}$ is a conjunct of $\theta_0$, we have that $\lnot \tup{\alpha = \beta}$ is a conjunct of $\psi$. But, on the other hand, by Remark~\ref{lema}, we know that $\tup{\dow[\psi]\beta =\dow[\theta_0]\delta_0}$ is a conjunct of $\varphi$ which implies, by Lemma~\ref{lema 2}, that $\tup{\alpha = \beta}$ is a conjunct of $\psi$, a contradiction. 

\begin{figure}[ht]
   \begin{center}
   \begin{tabular}{c@{\hskip .5in}c}
   \includegraphics[scale=0.25]{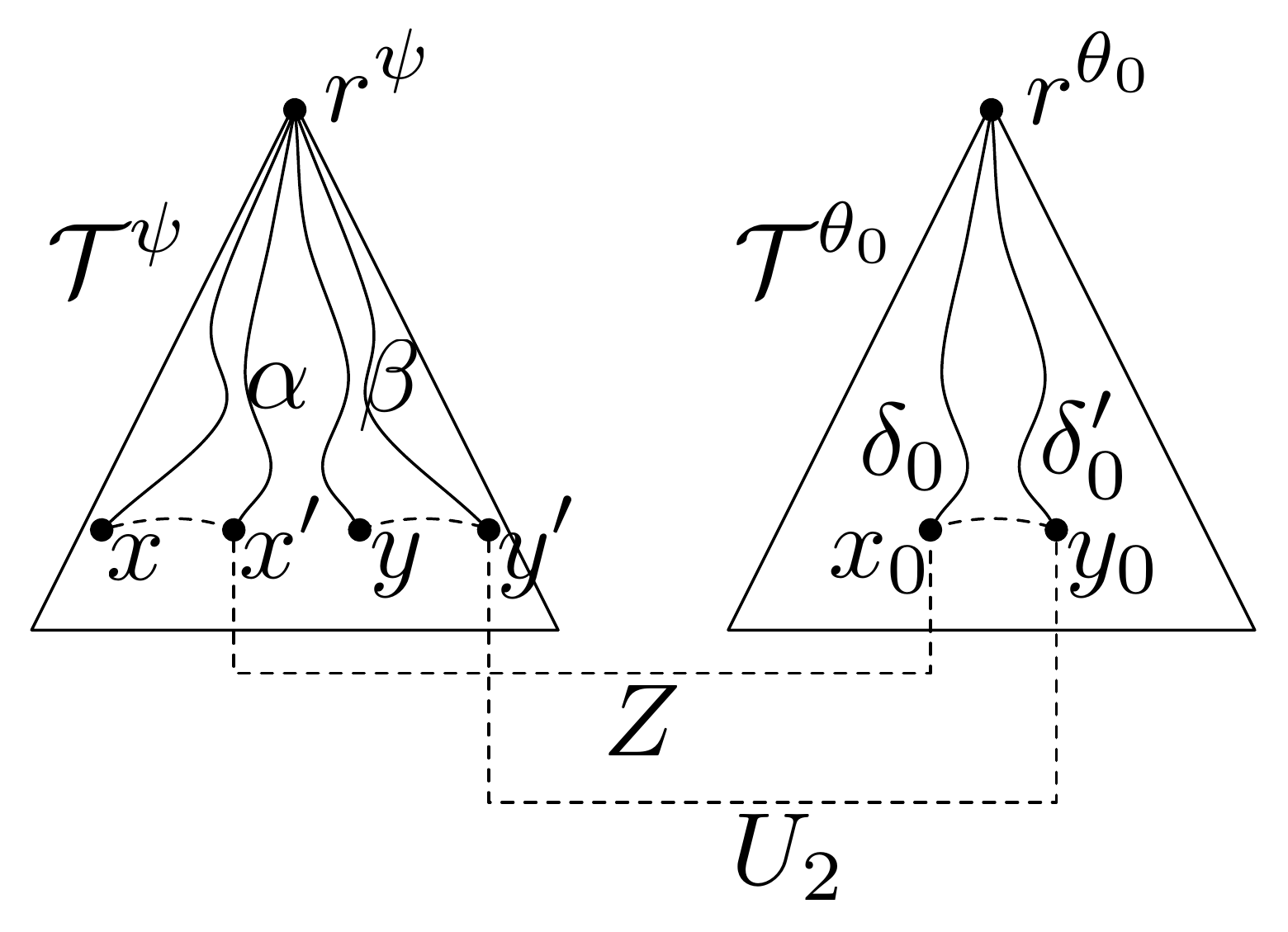} & \includegraphics[scale=0.25]{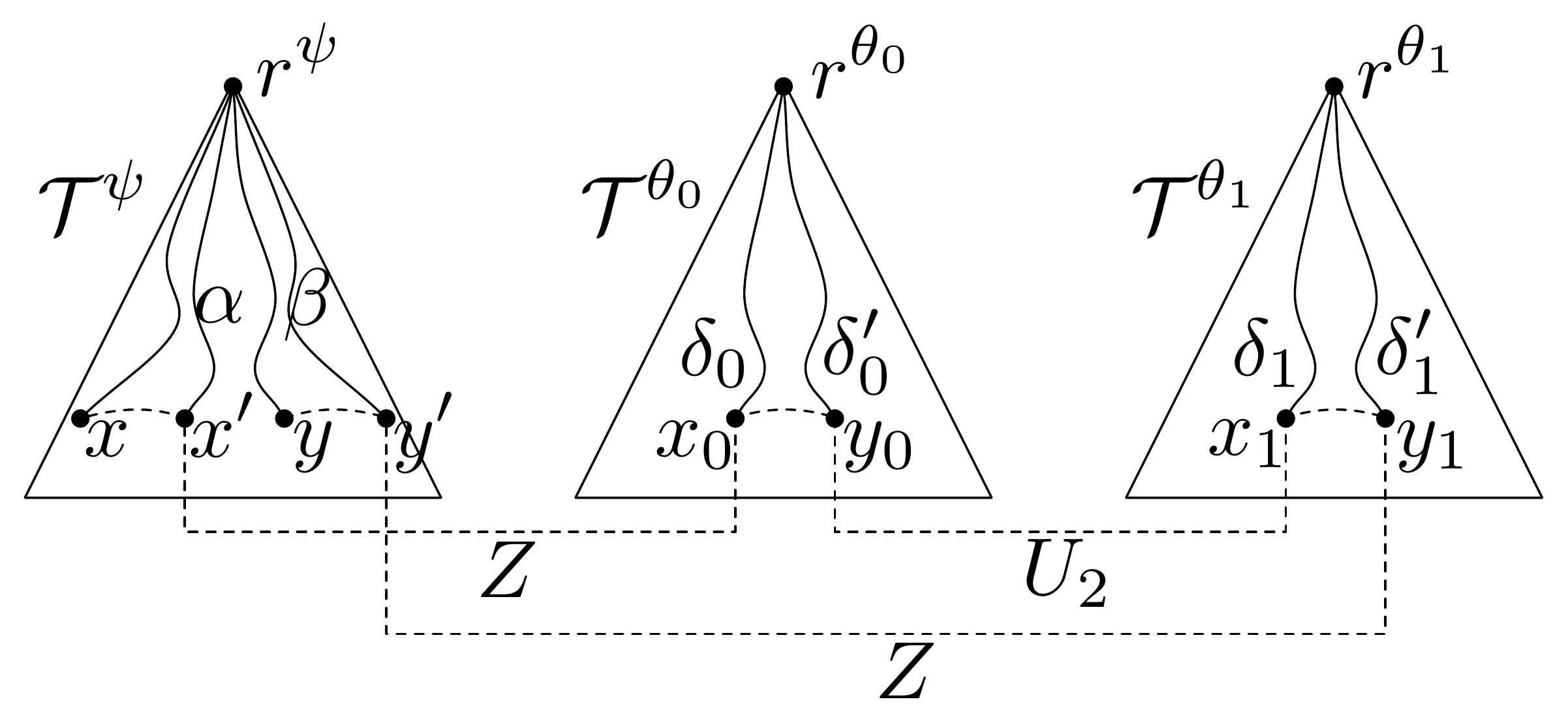} \\
   (a)&(b)
   \end{tabular}
   \end{center}
   \caption{Proof of Lemma~\ref{no cambia la particion}. (a) case $m=0$. (b) case $m=1$.}   \label{fig:8and9}
\end{figure}

  \item If $m=1$ (see Figure~\ref{fig:8and9}(b)):
   
   By Remark~\ref{lema}, $\tup{\dow[\theta_0]\delta_0 = \dow[\theta_1]\delta_1'}$ is a conjunct of $\varphi$ and then, by  \neqaxeight, $(\theta_0,\delta_0, \theta_1,\delta_1') \in {\bf Z}$. This gives a contradiction with the fact that $(\theta_0,\delta'_0,\theta_1,\delta_1)\in {\bf U_2}$ plus the fact that $\tup{\delta_0=\delta'_0}$ is a conjunct of $\theta_0$ and $\tup{\delta_1=\delta'_1}$ is a conjunct of $\theta_1$.
   
%
  \item For the induction, suppose $m\geq 2$:
  
  
  In case $m$ is odd, by Remark~\ref{lema}, $\tup{\dow[\theta_{m-1}]\delta_{m-1} = \dow[\theta_m]\delta_m'}$ is a conjunct of $\varphi$ and then, by  \neqaxeight, $(\theta_{m-1},\delta_{m-1},\theta_m, \delta_m') \in {\bf Z}$. By Lemma~\ref{lema c},   $(\psi, \beta, \theta_{m-2},\delta'_{m-2}) \in {\bf Z}$ and the result follows from inductive hypothesis for $m-2$.
  
 In case $m$ is even, by Remark~\ref{lema}, $\tup{\dow[\theta_{0}]\delta_{0} = \dow[\theta_1]\delta_1'}$ is a conjunct of $\varphi$ and then, by  \neqaxeight plus  \neqaxeight, $(\theta_{0},\delta_{0},\theta_1, \delta_1') \in {\bf Z}$. By Lemma~\ref{lema c}, $(\psi, \alpha, \theta_{2},\delta_{2}) \in {\bf Z}$ and the result follows from inductive hypothesis for $m-2$.
 \end{itemize}
This concludes the proof. 
\end{proof}

%
We conclude from Proposition~\ref{prop:local} and the construction that:
\begin{fact}\label{fac:preservesne}
The validity of a formula in a child of $r^\varphi$ is preserved in $\Tt^\varphi$. More formally:
\begin{itemize}
%
%
%
%
\item For each ${\bf v_1}\in {\bf \sisi}$,  $i\in\{1,2\}$ and $x,y\in T_i^{\bf v_1}$ we have $\Tt^\varphi,x \eqfull \Tt_i^{\bf v_1},x$ and $\Tt^\varphi,x,y \eqfull  \Tt_i^{\bf v_1},x,y$.

\item For each ${\bf v_2}\in {\bf \sino}$ and $x,y\in T^{\bf v_2}$ we have $\Tt^\varphi,x \eqfull  \Tt^{\bf v_2},x$ and \  $\Tt^\varphi,x,y \eqfull  \Tt^{\bf v_2},x,y$. 

\item For each ${\bf v_3}\in {\bf \nosi}$ and $x,y\in T^{\bf v_3}$ we have $\Tt^\varphi,x \eqfull  \Tt^{\bf v_3},x$ and \  $\Tt^\varphi,x,y \eqfull  \Tt^{\bf v_3},x,y$. 

\item For each ${\bf u}\in {\bf U}$, $i\in\{1,2\}$ and $x,y\in T^{\bf u}_i$
we have $\Tt^\varphi,x \eqfull \Tt^{\bf u}_i,x$ and \  $\Tt^\varphi,x,y \eqfull \Tt^{\bf u}_i,x,y$.
\end{itemize}
\end{fact}

\bigskip

%% file: verif-neq.tex
It only remains to prove that the conditions~\ref{labelneq} --
\ref{nopsialphaneqrhobetaneq} from the beginning of \S\ref{canonical model} are satisfied in the tree we have constructed:



\paragraph{Verification of \ref{labelneq}.}
This condition is trivially satisfied.

\paragraph{Verification of \ref{epsiloneqpsialphaneq}.} 
Suppose $\tup{\epsilon=\dow[\psi]\alpha}$ is a conjunct of $\varphi$. Then there are two possibilities, $(\psi,\alpha)\in {\bf V_{=,\neq}}$ or $(\psi,\alpha)\in {\bf V_{=, \lnot \neq}}$.
\begin{itemize}
\item In the first case, by Rule 1 and construction, 
there exists $x^{\bv_1}\in T^\varphi$ such that \mbox{$[r^\varphi]_{\pi^\varphi}=[x^{\bv_1}]_{\pi^\varphi}$} with $\bv_1=(\psi,\alpha)$.
By construction, we also know \mbox{$\Tt^{\bv_1}_1, r^{\bv_1}_1\models\psi$} and \mbox{$\Tt^{\bv_1}_1, r_1^{\bv_1},x^{\bv_1}\models\alpha$}.
Then, by Fact~\ref{fac:preservesne}, $\Tt^\varphi, r^\varphi\models\tup{\epsilon=\dow[\psi]\alpha}$.

\item In the second case, $\tup{\dow[\psi]\alpha}$ is consistent. Then, by construction plus Lemma~\ref{nextPathIsConjunctneq}, there is $x\in T^\varphi$ such that $\Tt^{\bv_2}, r^{\bv_2}\models\psi$, $\Tt^{\bv_2}, r^{\bv_2},x\models\alpha$ and $[r^\varphi]_{\pi^\varphi}=[x]_{\pi^\varphi}$ with ${\bf v_2}=(\psi,\alpha)$.
Then, by Fact~\ref{fac:preservesne}, $\Tt^\varphi,r^\varphi\models\tup{\epsilon=\dow[\psi]\alpha}$.
\end{itemize}

\paragraph{Verification of \ref{epsilonneqpsialphaneq}.} 
Suppose $\tup{\epsilon\neq\dow[\psi]\alpha}$ is a conjunct of $\varphi$. Then there are two possibilities, $(\psi,\alpha)\in {\bf V_{=,\neq}}$ or $(\psi,\alpha)\in {\bf V_{\lnot =,  \neq}}$.
\begin{itemize}
\item 
In the first case, by Rule 1 plus Lemmas~\ref{nextPathIsConjunctneq},~\ref{lema 1} and~\ref{lema 2}, there is $x \in T^\varphi$ such that (for ${\bf v_1}=(\psi,\alpha)$) $\Tt^{\bv_1}_2,r_2^{\bv_1},x\models\alpha$ and $x\not\in [z]_{\pi^\varphi}$ for all $z$ such that $\Tt^{\bv_1}_2,r_2^{\bv_1},z\models\beta$ for some $(\psi, \beta)\in {\bf V_{=, \lnot \neq}}$ (The argument is similar to the ones used in the proof of Fact~\ref{factnocambialaparticion1} to make conclusions from Lemma~\ref{lema 1}). We also know by construction that $\Tt^{\bv_1}_2,r^{\bv_1}_2\models\psi$. In order to conclude from Fact~\ref{fac:preservesne} that $\Tt^\varphi,r^\varphi\models\tup{\epsilon\neq\dow[\psi]\alpha}$, it only remains to observe that  $[r^\varphi]_{\pi^\varphi}\neq[x]_{\pi^\varphi}$ (for a sketch of the proof see Sketch~\ref{detalles de 2} in \S\ref{app}). 

\item In the second case, by Rule 3 plus Lemma~\ref{nextPathIsConjunctneq}, there exists $x \in T^\varphi$ such that (for ${\bf v_3}=(\psi,\alpha)$) $\Tt^{\bv_3},r^{\bv_3},x\models\alpha$. We also know by construction that $\Tt^{\bv_3},r^{\bv_3}\models\psi$. In order to conclude from Fact~\ref{fac:preservesne} that $\Tt^\varphi,r^\varphi\models\tup{\epsilon\neq\dow[\psi]\alpha}$, it only remains to observe that  $[r^\varphi]_{\pi^\varphi}\neq[x]_{\pi^\varphi}$ (the proof follows the same sketch than the previous case).
\end{itemize}

\paragraph{Verification of  \ref{psialphaeqrhobetaneq}.}
Suppose $\tup{\dow[\psi]\alpha=\dow[\rho]\beta}$ is a conjunct of $\varphi$. By the consistency of $\varphi$ plus \neqaxfour, neither $(\psi,\alpha)$ nor $(\rho,\beta)$ can be in ${\bf V_{\lnot =,  \lnot \neq}}$. By the consistency of $\varphi$ plus  \neqaxfour, it cannot  be the case that one of them belongs to ${\bf V_{=, \lnot \neq}}$ and the other one to ${\bf V_{\lnot =,  \neq}}$. Then there are five possibilities to consider (we are omitting symmetric cases):

\begin{itemize}
\item If $(\psi,\alpha)\in {\bf V_{=,\neq}}$ and $(\rho,\beta)\in {\bf V_{=,\neq}}$, by construction, 
there is $x^{\bv_1}\in T^\varphi$ such that $\Tt^{\bv_1}_1,r_1^{\bv_1}\models\psi$, $\Tt_1^{\bv_1},r_1^{\bv_1},x^{\bv_1}\models\alpha$ and $[r^\varphi]_{\pi^\varphi}=[x^{\bv_1}]_{\pi^\varphi}$, with $\bv_1=(\psi,\alpha)$. Since the same happens with $(\rho,\beta)$, we can conclude from Fact~\ref{fac:preservesne} that $\Tt^\varphi,r^\varphi\models\tup{\dow[\psi]\alpha=\dow[\rho]\beta}$.

\item If $(\psi,\alpha)\in {\bf V_{=,\neq}}$ and $(\rho,\beta)\in {\bf V_{=, \lnot \neq}}$, by construction, 
there is $x^{\bv_1}\in T^\varphi$ such that  $\Tt^{\bv_1}_1,r_1^{\bv_1}\models\psi$, $\Tt_1^{\bv_1},r_1^{\bv_1},x^{\bv_1}\models\alpha$ and $[r^\varphi]_{\pi^\varphi}=[x^{\bv_1}]_{\pi^\varphi}$, with $\bv_1=(\psi,\alpha)$.
By Lemma~\ref{nextPathIsConjunctneq} plus Rule 2, 
there is $x\in T^{\bv_2}$ (with $\bv_2=(\rho,\beta)$) such that $\Tt^{\bv_2},r^{\bv_2}\models\rho$ and $\Tt^{\bv_2},r^{\bv_2},x\models\beta$. Then, by construction,  $[r^\varphi]_{\pi^\varphi}=[x]_{\pi^\varphi}$ and so $[x^{\bf v_1}]_{\pi^\varphi}=[x]_{\pi^\varphi}$. 
We conclude from Fact~\ref{fac:preservesne} that $\Tt^\varphi,r^\varphi\models\tup{\dow[\psi]\alpha=\dow[\rho]\beta}$.

\item If $(\psi,\alpha)\in {\bf V_{=,\neq}}$ and $(\rho,\beta)\in {\bf V_{\lnot =,  \neq}}$, by the consistency of $\varphi$ plus \neqaxfour, $(\psi, \alpha, \rho, \beta)={\bf u} \in {\bf U}$. Then, by construction, there are $x^{{\bf u}}\in T_1^{\bf u}$, $y^{{\bf u}}\in T_2^{\bf u}$ such that $\Tt^{\bf u}_1,r^{\bf u}_1\models\psi$, $\Tt^{\bf u}_2,r^{\bf u}_2\models \rho$, $\Tt^{\bf u}_1,r^{\bf u}_1,x^{\bf u}\models\alpha$, $\Tt^{\bf u}_2,r^{\bf u}_2,y^{\bf u}\models\beta$ and  $[x^{\bf u}]_{\pi^\varphi}=[y^{\bf u}]_{\pi^\varphi}$ (If ${\bf u} \in {\bf U_2}$ the assertion is straightforward and if ${\bf u} \in {\bf U_1}$ these nodes exist because of the gluing related to the set $Z$). Then, we conclude from Fact~\ref{fac:preservesne} that  $\Tt^\varphi,r^\varphi\models\tup{\dow[\psi]\alpha=\dow[\rho]\beta}$. 

\item If $(\psi,\alpha)\in {\bf V_{=, \lnot \neq}}$ and $(\rho,\beta)\in {\bf V_{=, \lnot \neq}}$, by Rule 2 plus Lemma~\ref{nextPathIsConjunctneq}, 
there are $x\in T^{\bv_2}$ (with $\bv_2=(\psi,\alpha)$) and $y\in T^{\bv_2'}$ (with $\bv_2'=(\rho,\beta)$) such that $\Tt^{\bf v_2},r^{\bf v_2}\models\psi$, $\Tt^{\bv_2},r^{\bv_2},x\models\alpha$, $\Tt^{\bf v_2'},r^{\bf v_2'}\models\rho$ and $\Tt^{\bv_2'},r^{\bv_2'},y\models\beta$. By construction, $[x]_{\pi^\varphi}=[r^\varphi]_{\pi^\varphi}=[y]_{\pi^\varphi}$ and so, we conclude from Fact~\ref{fac:preservesne} that  $\Tt^\varphi,r^\varphi\models\tup{\dow[\psi]\alpha=\dow[\rho]\beta}$.

\item If $(\psi,\alpha)\in {\bf V_{\lnot =,  \neq}}$ and $(\rho,\beta)\in {\bf V_{\lnot =,  \neq}}$, then $(\psi,\alpha,\rho,\beta)={\bf u}\in {\bf U}$ or $(\psi,\alpha,\rho,\beta)={\bf z}\in {\bf Z}$. In the first case, the proof is exactly the same given for the case that \mbox{$(\psi,\alpha)\in {\bf V_{=,\neq}}$} and $(\rho,\beta)\in {\bf V_{\lnot =,  \neq}}$. In the other case, by Rule 3 plus Lemma~\ref{nextPathIsConjunctneq}, there are \mbox{$x\in T^{\bv_3}$}
(with $\bv_3=(\psi,\alpha)$), $y\in T^{{\bv_3'}}$ (with $\bv_3'=(\rho,\beta)$) such that $\Tt^{\bv_3},r^{\bv_3} \models\psi$, $\Tt^{\bv_3},r^{\bv_3},x\models\alpha$, $\Tt^{\bv_3'},r^{\bv_3'}\models \rho$ and $\Tt^{\bv_3'},r^{\bv_3'},y\models\beta$. Observe 
that 	 $[x]_{\pi^\varphi}=[y]_{\pi^\varphi}$ because of the way in which we have defined the partition $\pi^\varphi$. Then we conclude from Fact~\ref{fac:preservesne} that  
$\Tt^\varphi,r^\varphi\models\tup{\dow[\psi]\alpha=\dow[\rho]\beta}$.
\end{itemize}

\paragraph{Verification of  \ref{psialphaneqrhobetaneq}.}
Suppose $\tup{\dow[\psi]\alpha\neq \dow[\rho]\beta}$ is a conjunct of $\varphi$. By the consistency of $\varphi$ plus \neqaxfour, neither $(\psi,\alpha)$ nor $(\rho,\beta)$ can be in ${\bf V_{\lnot =,  \lnot \neq}}$. By the consistency of $\varphi$ plus \mbox{\neqaxeight,} it cannot  be the case that they both belong to ${\bf V_{=, \lnot \neq}}$. Then there are five possibilities to consider:
\begin{itemize}
\item If $(\psi,\alpha)\in {\bf V_{=,\neq}}$ and $(\rho,\beta)\in {\bf V_{=,\neq}}$, by items~\ref{epsiloneqpsialphaneq} and~\ref{epsilonneqpsialphaneq}, there exist $x,y \in T^\varphi$, such that $\Tt^\varphi,r^\varphi, x\models\dow[\psi]\alpha$, $\Tt^\varphi,r^\varphi, y\models\dow[\rho]\beta$,  $[r^\varphi]_{\pi^\varphi}=[x]_{\pi^\varphi}$ and $[r^\varphi]_{\pi^\varphi}\neq[y]_{\pi^\varphi}$. Then we conclude that  $\Tt^\varphi,r^\varphi\models\tup{\dow[\psi]\alpha\neq\dow[\rho]\beta}$. 


\item If $(\psi,\alpha)\in {\bf V_{=,\neq}}$ and $(\rho,\beta)\in {\bf V_{=, \lnot \neq}}$, or $(\psi,\alpha) \in {\bf V_{=,\neq}}$ and $(\rho,\beta)\in {\bf V_{\lnot =,  \neq}}$ or \mbox{$(\psi,\alpha)\in {\bf V_{=, \lnot \neq}}$} and $(\rho,\beta)\in {\bf V_{\lnot =,  \neq}}$, the proof is analogous to the previous one.

\item If $(\psi,\alpha)={\bf v}_3\in {\bf V_{\lnot =,  \neq}}$ and $(\rho,\beta)={\bf v'_3}\in {\bf V_{\lnot =,  \neq}}$. 
\begin{itemize}
 \item In case $(\psi,\alpha)\neq(\rho,\beta)$: If $\tup{\alpha \neq \alpha}$ is a conjunct of $\psi$ (if $\tup{\beta \neq \beta}$ is a conjunct of $\rho$, the proof is analogous), by Lemma~\ref{nextPathIsConjunctneq}, Rule 3 and Fact~\ref{factnocambialaparticion} there exist $x,y \in T^{\bv_3}$, $z\in T^{\bv_3'}$ such that $\Tt^{\bv_3}, r^{\bv_3} \models \psi$, $\Tt^{\bv_3},r^{\bv_3}, x\models \alpha$, $\Tt^{\bv_3},r^{\bv_3}, y\models \alpha$, $\Tt^{\bv_3'},r^{\bv_3'} \models \rho$, $\Tt^{\bv_3'},r^{\bv_3'}, z\models \beta$ and $[x]_{\pi^\varphi}\neq [y]_{\pi^\varphi}$. Then we conclude from Fact~\ref{fac:preservesne} that $\Tt^{\varphi}, r^{\varphi} \models \tup{\dow[\psi]\alpha\neq \dow[\rho]\beta}$ (either $x$ or $y$ is not in $[z]_{\pi^{\varphi}}$).
 
 Suppose then that $\lnot \tup{\alpha \neq \alpha}$ is a conjunct of $\psi$ and $\lnot \tup{\beta \neq \beta}$ is a conjunct of $\rho$. Then, as before, there exist $x \in T^{\bv_3}$, $z\in T^{\bv_3'}$ such that $\Tt^{\bv_3},r^{\bv_3} \models \psi$, $\Tt^{\bv_3},r^{\bv_3}, x\models \alpha$, $\Tt^{\bv_3'}, r^{\bv_3'} \models \rho$, $\Tt^{\bv_3'},r^{\bv_3'}, z\models \beta$. To conclude the proof, it only remains to observe that, in this case, $[x]_{\pi^\varphi}\neq [z]_{\pi^\varphi}$ (for a sketch of the proof see Sketch~\ref{detalles de 4a} in \S\ref{app}). Then we conclude from Fact~\ref{fac:preservesne} that $\Tt^{\varphi}, r^{\varphi} \models \tup{\dow[\psi]\alpha\neq \dow[\rho]\beta}$.

 \item In case $(\psi,\alpha)=(\rho,\beta)$, by consistency of $\varphi$, we have that $(\psi,\alpha,\psi,\alpha)={\bf u}\in {\bf U}$. If $\tup{\alpha \neq \alpha}$ is a conjunct of $\psi$, by Lemma~\ref{nextPathIsConjunctneq}, Rule 3 and Fact~\ref{factnocambialaparticion} there exist  $x,y \in T^{\bv_3}$ such that $\Tt^{\bv_3},r^{\bv_3} \models \psi$, $\Tt^{\bv_3},r^{\bv_3}, x\models \alpha$, $\Tt^{\bv_3},r^{\bv_3}, y\models \alpha$ and $[x]_{\pi^\varphi}\neq [y]_{\pi^\varphi}$. Then we conclude from Fact~\ref{fac:preservesne} that $\Tt^{\varphi}, r^{\varphi} \models \tup{\dow[\psi]\alpha\neq \dow[\psi]\alpha}$.
 
 Suppose then that $\lnot \tup{\alpha \neq \alpha}$ is a conjunct of $\psi$. Then, as before, there exist $x \in T^{\bv_3}$, $z\in T_1^{\bf u}$ such that $\Tt^{\bv_3}, r^{\bf v_3} \models \psi$, $\Tt^{\bv_3},r^{\bv_3}, x\models \alpha$, $\Tt_1^{\bf u},r_1^{\bf u} \models \psi$, $\Tt_1^{\bf u},r_1^{\bf u}, z\models \alpha$. To conclude the proof, it only remains to observe that, in this case, $[x]_{\pi^\varphi}\neq [z]_{\pi^\varphi}$ (for a sketch of the proof see Sketch~\ref{detalles de 4b} in \S\ref{app}). Then we conclude from Fact~\ref{fac:preservesne} that $\Tt^{\varphi}, r^{\varphi} \models \tup{\dow[\psi]\alpha\neq \dow[\rho]\beta}$.
\end{itemize}

\end{itemize}

\paragraph{Verification of  \ref{noepsiloneqpsialphaneq}.}
Suppose $\neg\tup{\epsilon=\dow[\psi]\alpha}$ is a conjunct of $\varphi$. Aiming for a contradiction, suppose that $\Tt^\varphi,r^\varphi\models\tup{\epsilon=\dow[\psi]\alpha}$. Then there is a successor $z$ of $r^\varphi$ in which $\psi$ holds, and, by construction plus Lemma~\ref{lemma:inconsistentneq}, $z$ is the root of some copy of the tree $\Tt^\psi$, i.e.\ $z=r^\psi$ (it might be $\widetilde{\Tt^{\psi}}$ and $\widetilde{r^{\psi}}$ but, in that case, the argument is the same).  Moreover, there is $x\in T^\psi$ such that $\Tt^\psi,r^\psi,x\models\alpha$, with $[x]_{\pi^\varphi}=[r^\varphi]_{\pi^\varphi}$. In addition to this, $(\psi,\alpha)\in {\bf V_{\lnot =,  \neq}}$ or $(\psi,\alpha)\in {\bf V_{\lnot =,  \lnot \neq}}$. If the latter occurs, by construction of $\Tt^\varphi$ plus Lemma~\ref{lema 5} and Lemma~\ref{lemma:inconsistentneq}, we have that $\lnot \tup{\alpha=\alpha}$ is a conjunct of $\psi$ which is a contradiction. In the former, observe that 
$[x]_{\pi^\varphi}\neq [r^\varphi]_{\pi^\varphi}$ (for a sketch of the proof see Sketch~\ref{detalles de 5} in \S\ref{app}) which is a 
contradiction.   


%

\paragraph{Verification of  \ref{noepsilonneqpsialphaneq}.}
Suppose $\neg\tup{\epsilon\neq\dow[\psi]\alpha}$ is a conjunct of $\varphi$. Aiming for a contradiction, suppose that  $\Tt^\varphi,r^\varphi\models\tup{\epsilon\neq\dow[\psi]\alpha}$. Then there is a successor $z$ of $r^\varphi$ in which $\psi$ holds, and by construction and Lemma~\ref{lemma:inconsistentneq}, $z$ is the root of some copy of the tree $\Tt^\psi$, i.e.\ $z=r^\psi$ (it might be $\widetilde{\Tt^{\psi}}$ and $\widetilde{r^{\psi}}$ but, in that case, the argument is the same).  Moreover, there is $x\in T^\psi$ such that $\Tt^\psi,r^\psi,x\models\alpha$, with $[x]_{\pi^\varphi}\neq[r^\varphi]_{\pi^\varphi}$. Then, by construction, $(\psi,\alpha)\not \in {\bf V_{=, \lnot \neq}}$. Since $\neg\tup{\epsilon\neq\dow[\psi]\alpha}$ is a conjunct of $\varphi$, the only remaining possibility is that $(\psi,\alpha)\in {\bf V_{\lnot =,  \lnot \neq}}$ but this 
is a contradiction by construction plus Lemma~\ref{lema 5} and Lemma~\ref{lemma:inconsistentneq}. 

\paragraph{Verification of  \ref{nopsialphaeqrhobetaneq}.}
Suppose $\neg\tup{\dow[\psi]\alpha=\dow[\rho]\beta}$ is a conjunct of $\varphi$. By the consistency of $\varphi$ plus  \neqaxfive, it cannot  be the case that both $(\psi,\alpha), (\rho,\beta)$ are in ${\bf V_{=,\neq}}\cup {\bf V_{=, \lnot \neq}}$. In case $(\psi,\alpha)\in {\bf V_{\lnot =,  \lnot \neq}}$ (if \mbox{$(\rho, \beta) \in {\bf V_{\lnot =,  \lnot \neq}}$}, the proof is analogous), suppose that 
$\Tt^\varphi,r^\varphi\models\tup{\dow[\psi]\alpha = \dow[\rho]\beta}$. In particular, there is a successor of $r^\varphi$, $z$ and a descendant $w$ such that $\Tt^\varphi, z,w \models [\psi]\alpha$. But this is a contradiction by construction plus Lemma~\ref{lema 5} and Lemma~\ref{lemma:inconsistentneq}. Then $\Tt^\varphi,r^\varphi\models \neg \tup{\dow[\psi]\alpha = \dow[\rho]\beta}$. We then have three remaining cases to analyze:
\begin{itemize}
\item
If $(\psi,\alpha)\in {\bf V_{=, \lnot \neq}}$ and $(\rho,\beta)\in {\bf V_{\lnot =,  \neq}}$, then, by items~\ref{noepsiloneqpsialphaneq} and~\ref{noepsilonneqpsialphaneq}, we have the result.

\item If $(\psi,\alpha)\in {\bf V_{=,\neq}}$ and $(\rho,\beta)\in {\bf V_{\lnot =,  \neq}}$ or $(\psi,\alpha), (\rho,\beta)\in {\bf V_{\lnot =,  \neq}}$. In order to conclude that $\Tt^\varphi,r^\varphi\models \lnot \tup{\dow[\psi]\alpha=\dow[\rho]\beta}$, one only have to observe that, if $x, y \in T^\varphi$ are such that $\Tt^{\varphi},r^{\varphi}, x\models \dow[\psi]\alpha$ and $\Tt^{\varphi},r^{\varphi}, y\models \dow[\rho]\beta$, then  $[x]_{\pi^\varphi}\neq[y]_{\pi^\varphi}$ (for a sketch of the proof see Sketch~\ref{detalles de 7} in \S\ref{app}). 
\end{itemize}

\paragraph{Verification of  \ref{nopsialphaneqrhobetaneq}.}
Suppose $\neg\tup{\dow[\psi]\alpha\neq\dow[\rho]\beta}$ is a conjunct of $\varphi$. By the consistency of $\varphi$ plus  \neqaxfive, it cannot  be the case that one of $(\psi,\alpha), (\rho,\beta)$ is from ${\bf V_{=,\neq}}$ and the other from ${\bf V_{=, \lnot \neq}}$, neither can one be from ${\bf V_{=,\neq}}$ and the other from ${\bf V_{\lnot =,  \neq}}$, or one from ${\bf V_{=, \lnot \neq}}$ and the other from ${\bf V_{\lnot =,  \neq}}$, or both from ${\bf V_{=,\neq}}$. In case $(\psi,\alpha)\in {\bf V_{\lnot =,  \lnot \neq}}$ (if $(\rho, \beta) \in {\bf V_{\lnot =,  \lnot \neq}}$, the proof is analogous), suppose that 
$\Tt^\varphi,r^\varphi\models\tup{\dow[\psi]\alpha \neq \dow[\rho]\beta}$. In particular, there is a successor $z$ of $r^\varphi$ and a descendant $w$ such that $\Tt^\varphi, z,w \models [\psi]\alpha$. But this is a contradiction by construction plus Lemma~\ref{lema 5} and Lemma~\ref{lemma:inconsistentneq}. We then have two remaining cases to analyze:
\begin{itemize}
\item If $(\psi,\alpha), (\rho,\beta) \in {\bf V_{=, \lnot \neq}}$, by item~\ref{noepsilonneqpsialphaneq}, we have the result.

\item If $(\psi,\alpha), (\rho,\beta) \in {\bf V_{\lnot =,  \neq}}$, by the consistency of $\varphi$ plus  \neqaxfour, $(\psi,\alpha,\rho,\beta) \in {\bf Z}$ and the result follows immediately from the construction of the model.
\end{itemize}


%% file: conclu.tex

The addition of an equivalence relation on top of a tree-like Kripke model, and the ability of the modal language to compare if two nodes at the end of path expressions are in the {\em same} or in {\em different}  equivalence classes has proved to change remarkably the canonical model construction of the basic modal logic. When the language has only comparisons by `equality', the situation is somewhat simpler, based on the fact that `equality' is a transitive relation. Also notice that while
\begin{equation}\label{eqn:concl:1}
\mbox{``all pairs of paths with certain properties end in {\em different} equivalence classes"}
\end{equation}
is expressible when tests by equality are present, 
\begin{equation}\label{eqn:concl:2}
\mbox{``all pairs of paths with certain properties end in the {\em same} equivalence classes"} 
\end{equation}
is only expressible when tests by inequality are also present.
Both properties are universal. However, in the construction of the canonical model,  \eqref{eqn:concl:1} is compatible with adding many disjoint copies of subtrees with disjoint partitions, while \eqref{eqn:concl:2} is not. The axiomatization for the fragment containing both the operators of `equality' and `inequality' proved to be much more involved than the one containing only `equality', as witnessed by the large amount of axioms reflecting the intricate relationships between both binary operators.

\bigskip

In this research we have considered $\xpd$ over arbitrary data trees. Furthermore, $\xpd$ is also suitable for reasoning about (finite or infinite) data {\em graphs}, as it is done in \cite{LV12,KR16}. In either of the alternatives (finite vs.\ infinite data trees vs.\ data graphs) it can be shown that $\xpd$ is also axiomatizable by the system given in this paper ---notice there are no specific axioms of an underlying tree topology. Since our construction of canonical models gives us a recursively bounded finite data tree, we conclude:
\begin{corollary}[Bounded tree model property]
There is a primitive recursive function $f$ such that any satisfiable node or path expression $\varphi$ of $\xpd$ of size $n$ over the class of finite/arbitrary data trees/data graphs is satisfiable in a data tree of size at most $f(n)$.
\end{corollary}
Hence, although in the database community one may restrict to the finite case, the above corollary shows that allowing or disallowing infinite models does not make any difference.

This already shows that the satisfiability problem of $\xpd$ is decidable over any of the classes of models stated above. Of course, this result ---at least for $\xpd$ over finite data trees--- is not new, as mentioned in the introduction \cite{Figueira12ACM}. However, the canonical model construction may give us new insights into obtaining sequent calculus axiomatizations, as done in \cite{datagl}, which might be useful for obtaining alternative proofs of complexity for the satisfiability problem of fragments or extensions of $\xpd$.

On the application side, the axioms may help to define effective rewrite rules for query optimization in $\xpd$. 

The study of $\xpath$ with `descendant' instead of `child' axis seems to be much harder. This question, or the addition of other axes such as `parent' or `sibling' (in the case of ordered trees), constitute future lines of research.

%% file: appendix-neq.tex

\noindent{\bf Lema \ref{lema A}.}
{\em \lemaA}

\begin{proof}
Let us start with the case of $*$ being $\neq$. 
Aiming for a contradiction, suppose that $\tup{\gamma \neq \dow[\psi_1]\dots\dow[\psi_{i_0}]\alpha} \land \lnot \tup{\gamma \neq \dow[\psi_1]\dots\dow[\psi_{i_0}]\beta}$ is consistent and that both $\lnot \tup{\alpha \neq \alpha}$ and $\tup{\alpha = \beta}$
are conjuncts of $\psi_{i_0}$.

First, let us prove some facts that will be useful in the rest of the proof:
\begin{enumerate}
  \item\label{it:gamma-conjbis} The following derivation:
  \begin{align*}
  \tup{\gamma \neq \dow[\psi_1]\dots\dow[\psi_{i_0}]\alpha} & \leq \tup{\dow[\psi_1]\dots\dow[\psi_{i_0}]\alpha}\tag{\neqaxtwo} \\
  &\leq  \tup{\dow[\psi_1]\dots\dow[\psi_{i_0}]} \tag{{\bf Der12}~Fact \ref{fact boolean}} \\
  &\leq  \tup{\dow[\psi_1]\dots\dow[\psi_{i_0}]\alpha = \dow[\psi_1]\dots\dow[\psi_{i_0}]\beta}\tag{\eqax{7} \& {\bf Der21} (Fact \ref{fact boolean})} \\
  & \leq \tup{\dow[\psi_1]\dots\dow[\psi_{i_0}]\beta}\tag{\eqax{5}}
  \end{align*}
  In particular, by {\bf Der13} (Fact~\ref{fact boolean}), we have that $\tup{\dow[\psi_{i_0}]\beta}$ is consistent and so $\tup{\beta=\beta}$ is a conjunct of $\psi_{i_0}$ (by Lemma~\ref{nextPathIsConjunctneq}).

  \item\label{it:psi-conbis} From the second line of Item~\ref{it:gamma-conjbis}, we have that $\tup{\dow[\psi_1]\dots\dow[\psi_{i_0}]}$ is
  consistent, and then, by {\bf Der13} (Fact \ref{fact boolean}), $\psi_{i_0}$ is
  consistent.

  \item\label{it:all-gammabis} Aiming for a contradiction, let us suppose that $\tup{\beta\neq\beta}$
  is a conjunct of $\psi_{i_0}$. Then
  \begin{align*}
  & \hspace{-20pt} \tup{\gamma \neq \dow[\psi_1]\dots\dow[\psi_{i_0}]\alpha} \land \lnot \tup{\gamma \neq \dow[\psi_1]\dots\dow[\psi_{i_0}]\beta} \\&
  \leq \tup{\gamma} \land \tup{\dow[\psi_1]\dots\dow[\psi_{i_0}]} \land \lnot \tup{\gamma \neq \dow[\psi_1]\dots\dow[\psi_{i_0}]\beta}\tag{\neqaxtwo \&  Item~\ref{it:gamma-conjbis}} \\ &
  \leq \tup{\gamma} \land \tup{\dow[\psi_1]\dots\dow[\psi_{i_0}]\beta\neq\dow[\psi_1]\dots\dow[\psi_{i_0}]\beta} \land \lnot \tup{\gamma \neq \dow[\psi_1]\dots\dow[\psi_{i_0}]\beta}\tag{\neqaxfifteen \& {\bf Der21} (Fact \ref{fact boolean})}\\&
  \equiv \tup{\gamma \neq \dow[\psi_1]\dots\dow[\psi_{i_0}]\beta}\land \lnot \tup{\gamma \neq \dow[\psi_1]\dots\dow[\psi_{i_0}]\beta}\tag{\neqaxeight} \\&
  \equiv \botNode \tag{Boolean}
\end{align*}
which is a contradiction. Then $\neg\tup{\beta\neq\beta}$ is a conjunct of $\psi_{i_0}$.

  \item\label{it:all-alpha-gammabis} Because $\psi_{i_0}$ is consistent (Item~\ref{it:psi-conbis}), by the previous Item plus \neqaxeight, $\lnot \tup{\alpha\neq \beta}$ is a conjunct of $\psi_{i_0}$.
%
%
  \end{enumerate}
Then we have 
\begin{align*}
  &\tup{\gamma{\neq}\dow[\psi_1]\dots\dow[\psi_{i_0}]\alpha} \land \lnot \tup{\gamma{\neq}\dow[\psi_1]\dots\dow[\psi_{i_0}]\beta}   \\&
   \leq \tup{\gamma{\neq}\dow[\psi_1]\dots\dow[\psi_{i_0}]\beta} \land \lnot \tup{\gamma{\neq}\dow[\psi_1]\dots\dow[\psi_{i_0}]\beta} \tag{Items~\ref{it:gamma-conjbis} and~\ref{it:all-alpha-gammabis} \&  \neqaxten  \& {\bf Der21} (Fact \ref{fact boolean})} \\ &
   \equiv \botNode\tag{Boolean}
\end{align*} 
%
which is contradiction, from the assumption that $\tup{\alpha=\beta}$ was a conjunct of $\psi_{i_0}$.
Therefore,  $\neg\tup{\alpha=\beta}$ is a conjunct of $\psi_{i_0}$.

For the case of $*$ being $=$, use \neqaxeighteen plus {\bf Der21} of Fact \ref{fact boolean}. 
\end{proof}


%

%

\noindent{\bf Lema \ref{paraverif}.}
{\em \paraverif}

\begin{proof}
Aiming for a contradiction, suppose that $\tup{\dow[\psi]\alpha \neq \dow[\psi]\alpha} \land \lnot \tup{\dow[\psi]\gamma \neq \dow[\psi]\gamma}$ is consistent and both $\lnot \tup{\alpha \neq \alpha}$ and $\tup{\alpha = \gamma}$
are conjuncts of $\psi$.

Let us prove some facts that will be useful in the rest of the proof:
\begin{enumerate}
  \item\label{it:gamma-conj} The following derivation:
\begin{align*}
  \tup{\dow[\psi]\alpha \neq \dow[\psi]\alpha} & \leq \tup{\dow[\psi]\alpha}\tag{\neqaxtwo} \\
  & \leq \tup{\dow[\psi]}\tag{{\bf Der12} (Fact \ref{fact boolean})} \\
  & \leq \tup{\dow[\psi]\alpha=\dow[\psi]\gamma}\tag{\eqax{7} \& {\bf Der21}(Fact \ref{fact boolean})} \\
  & \leq \tup{\dow[\psi]\gamma}\tag{\eqax{5}}
  \end{align*}

  In particular, we have that $\tup{\gamma=\gamma}$ is a conjunct of $\psi$ (by Lemma~\ref{nextPathIsConjunctneq}).

  \item\label{it:psi-con} From the second line of Item~\ref{it:gamma-conj}, we have that $\tup{\dow[\psi]}$ is
  consistent, and by  {\bf Der13} (Fact \ref{fact boolean}), $\psi$ is
  consistent.

  \item\label{it:all-gamma} Aiming for a contradiction, let us suppose that $\tup{\gamma\neq\gamma}$
  is a conjunct of $\psi$. Then
%
  \begin{align*}
  &\hspace{-20pt}  \tup{\dow[\psi]\alpha \neq \dow[\psi]\alpha} \land  \lnot \tup{\dow[\psi]\gamma \neq \dow[\psi]\gamma} \\&
  \leq  \tup{\dow[\psi]} \land  \lnot \tup{\dow[\psi]\gamma \neq \dow[\psi]\gamma} \tag{Item~\ref{it:gamma-conj}} \\&
  \leq \tup{\dow[\psi]\gamma\neq \dow[\psi]\gamma} \land  \lnot \tup{\dow[\psi]\gamma \neq \dow[\psi]\gamma}\tag{\neqaxfifteen \& {\bf Der21} (Fact \ref{fact boolean})}\\&
  \equiv  \botNode\tag{Boolean}
  \end{align*} 
which is a contradiction. Then $\neg\tup{\gamma\neq\gamma}$ is a conjunct of $\psi$.

  \item\label{it:all-alpha-gamma} Because $\psi$ is consistent (Item~\ref{it:psi-con}), by the previous item plus \neqaxeight, $\lnot \tup{\alpha\neq \gamma}$ is a conjunct of $\psi$.
   
  \end{enumerate}
%
Then we have 
\begin{align*}
 & \hspace{-20pt} \tup{\dow[\psi]\alpha \neq \dow[\psi]\alpha} \land  \lnot \tup{\dow[\psi]\gamma \neq \dow[\psi]\gamma} \\&
 \leq \tup{\dow[\psi]\gamma} \land \tup{\dow[\psi]\alpha \neq \dow[\psi]\alpha}\land  \lnot \tup{\dow[\psi]\gamma \neq \dow[\psi]\gamma}\tag{Item~\ref{it:gamma-conj}} \\&
  \leq \tup{\dow[\psi]\alpha\neq \dow[\psi]\gamma} \land  \lnot \tup{\dow[\psi]\gamma \neq \dow[\psi]\gamma}\tag{\neqaxeight} \\&
  \leq \tup{\dow[\psi]\gamma \neq \dow[\psi]\gamma} \land \lnot \tup{\dow[\psi]\gamma \neq \dow[\psi]\gamma}\tag{Items~\ref{it:gamma-conj} and~\ref{it:all-alpha-gamma}, \& \neqaxten  \& {\bf Der21} (Fact \ref{fact boolean})} \\&
  \equiv\botNode\tag{Boolean}
  \end{align*}
%
which is contradiction, from the assumption that $\tup{\alpha=\gamma}$ was a conjunct of $\psi$.
Therefore,  $\neg\tup{\alpha=\gamma}$ is a conjunct of $\psi$.
\end{proof}



\noindent{\bf Lema \ref{lema 5}.}
{\em \lemacinco}

\begin{proof}
Aiming for a contradiction, suppose that $\tup{\alpha=\alpha}$ is a conjunct of $\psi$. Then 
  \begin{align*}
   &\hspace{-20pt} \lnot \tup{\gamma= \dow[\psi]\alpha} \land \lnot \tup{\gamma \neq \dow[\psi]\alpha} \land \tup{\gamma} \land \tup{\dow[\psi]}\\&
   \leq \lnot \tup{\gamma= \dow[\psi]\alpha} \land \lnot \tup{\gamma \neq \dow[\psi]\alpha} \land \tup{\gamma} \land \tup{\dow[\psi]\alpha =\dow[\psi]\alpha}\tag{\eqax{7} \& {\bf Der21} (Fact \ref{fact boolean})}\\&
  \equiv \lnot \tup{\gamma} \land \tup{\gamma}\tag{\neqaxfour} \\&
  \equiv\botNode\tag{\hbox{Boolean}}
\end{align*}
and this concludes the proof.
%
%
\end{proof}



\noindent{\bf Lema \ref{lema 1}.}
{\em \lemauno}

 
\begin{proof}
Let us first prove the case for $*$ being $\neq$. 
Suppose that   $\tup{\gamma = \dow[\psi]\alpha} \land \lnot \tup{\gamma \neq \dow[\psi]\alpha} \land \lnot \tup{\gamma \neq \dow[\psi]\beta}$ is consistent. Aiming for a contradiction, suppose that $\tup{\alpha\neq\beta}$ is a conjunct of $\psi$. Then
%
  \begin{align*}&
\hspace{-20pt} \tup{\gamma = \dow[\psi]\alpha} \land \lnot \tup{\gamma \neq \dow[\psi]\alpha} \land \lnot \tup{\gamma \neq \dow[\psi]\beta} \\&
    \equiv  \tup{\gamma = \dow[\psi\land\tup{\alpha\neq\beta}]\alpha} \land \lnot \tup{\gamma \neq \dow[\psi]\alpha} \land \lnot \tup{\gamma \neq \dow[\psi]\beta} \tag{Hypothesis}  \\&
   \leq \tup{\gamma} \land \tup{\dow[\psi\land\tup{\alpha\neq\beta}]\alpha}  \land \lnot \tup{\gamma \neq \dow[\psi]\alpha} \land \lnot \tup{\gamma \neq \dow[\psi]\beta} \tag{\eqax{2}  \& \eqax{5}} \\&
   \leq \tup{\gamma} \land \tup{\dow[\psi\land\tup{\alpha\neq\beta}]}  \land \lnot \tup{\gamma \neq \dow[\psi]\alpha} \land \lnot \tup{\gamma \neq \dow[\psi]\beta} \tag{{\bf Der12} (Fact \ref{fact boolean})} \\&
   \leq \tup{\gamma} \land \tup{\dow[\psi]\alpha \neq \dow[\psi]\beta} \land \lnot \tup{\gamma \neq \dow[\psi]\alpha} \land \lnot \tup{\gamma \neq \dow[\psi]\beta} \tag{\neqaxfifteen \& {\bf Der21} (Fact \ref{fact boolean})} \\&
   \leq \tup{\gamma \neq \dow[\psi]\beta} \land \lnot \tup{\gamma \neq \dow[\psi]\beta} \tag{\neqaxeight}\\&
   \equiv \botNode\tag{Boolean}
\end{align*}
which is a contradiction. Then $\tup{\alpha\neq\beta}$ is a conjunct of $\psi$. For the case in which $*$ is $=$, the proof is similar but instead of \neqaxfifteen we use \eqax{7} and instead of \neqaxeight we use \neqaxfour.

\end{proof}


\noindent{\bf Lema \ref{lema 2}.}
{\em  \lemados}

\begin{proof}
Let us first prove the case for $*$ being $=$. 
Suppose that   $\tup{\gamma = \dow[\psi]\alpha} \land \lnot \tup{\gamma \neq \dow[\psi]\alpha}  \land  \tup{\gamma = \dow[\psi]\beta}$ is consistent. Aiming for a contradiction, suppose that $\lnot \tup{\alpha = \beta}$ is a conjunct of $\psi$. Also, since $\tup{\dow[\psi]\alpha}$ is consistent (by \eqax{5}), then by Lemma~\ref{nextPathIsConjunctneq} $\tup{\alpha=\alpha}$ is a conjunct of $\psi$. Then
\begin{align*}
  &\hspace{-20pt}  \tup{\gamma = \dow[\psi]\alpha}{\land}\lnot \tup{\gamma \neq \dow[\psi]\alpha}{\land}\tup{\gamma = \dow[\psi]\beta} \\&
    \leq \lnot \tup{\gamma \neq \dow[\psi]\alpha}\land\tup{\gamma = \dow[\psi{\land}\neg\tup{\alpha=\beta}{\land}\tup{\alpha}]\beta}\tag{\eqax{1}}\\&
  \leq  \lnot \tup{\gamma \neq \dow[\psi]\alpha}  \land  \tup{\gamma \neq \dow[\psi]\alpha}\tag{\neqaxnine \& {\bf Der21} (Fact \ref{fact boolean})}\\&
  \equiv  \botNode \tag{Boolean}
\end{align*}
which is a contradiction. Then $\tup{\alpha = \beta}$ is a conjunct of $\psi$.
For the case in which $*$ is $\neq$, the proof is similar but using \neqaxten instead of  \neqaxnine.
\end{proof}

 \begin{sketch}\label{detalles de 2}
 Thinking in terms of sequences as in the proofs of Facts~\ref{factnocambialaparticion1} and~\ref{factnocambialaparticion}, one only has to observe that:
 
 \begin{itemize}
  
  \item $[x]_{\pi_2^{\bf v_1}}\neq [z]_{\pi_2^{\bf v_1}}$ for all $z\in T^{\bf v_1}_2$ in a class that was glued to the class of the {\rm root} via a \emph{${\rm root}_{=,\lnot \neq}$}-kind gluing. 
  
  \item \emph{${\rm root}_{=,\neq}$}-kind gluings are made in different subtrees.
  
  \item By the same arguments given in the proofs of Facts~\ref{factnocambialaparticion1} and~\ref{factnocambialaparticion}, we can't have a sequence containing any of the following: 
\begin{center}
  \begin{tabular}{l@{\hskip .5in}l@{\hskip .5in}l}
   -- \emph{${\rm root}_{=,\neq}$}-\emph{$Z$}, 
   &
   -- \emph{$Z$}-\emph{${\rm root}_{=,\neq}$}, 
   &
   -- \emph{${\rm root}_{=,\lnot \neq}$}-\emph{$Z$},
\\ 
   -- \emph{$Z$}-\emph{${\rm root}_{=,\lnot \neq}$}, 
   &
   -- \emph{${\rm root}_{=,\neq}$}-\emph{$U_2$},
   &
   -- \emph{$U_2$}-\emph{${\rm root}_{=,\neq}$},
\\ 
   -- \emph{${\rm root}_{=,\lnot \neq}$}-\emph{$U_2$}, 
   &
   -- \emph{$U_2$}-\emph{${\rm root}_{=,\lnot \neq}$}.
   &
  \end{tabular}
  \end{center}
  
 \end{itemize}

 \end{sketch}

\begin{sketch}\label{detalles de 4a}
  Thinking in terms of sequences as in the proofs of Facts~\ref{factnocambialaparticion1} and~\ref{factnocambialaparticion}, one only has to observe that:

  \begin{itemize}
    \item $[x]_{\pi^{\bf v_3}}\neq [y]_{\pi^{\bf v_3}}$ for all $y\in T^{\bf v_3}$ in a class that was glued to the class of the {\rm root} via a \emph{${\rm root}_{=,\lnot \neq}$}-kind gluing (Use Lemma~\ref{lema 1}).
    
    \item $[z]_{\pi^{\bf v_3'}}\neq [y]_{\pi^{\bf v_3'}}$ for all $y\in T^{\bf v_3'}$ in a class that was glued to the class of the {\rm root} via a \emph{${\rm root}_{=,\lnot \neq}$}-kind gluing (Use Lemma~\ref{lema 1}). 
        
    \item \emph{${\rm root}_{=,\neq}$}-kind gluings are made in different subtrees.
    
   \item $[x]_{\pi^{\bf v_3}}$ and $[z]_{\pi^{\bf v_3'}}$ can not be glued together by a sequence of all \emph{$Z$}-kind gluings because of the consistency of $\varphi$ plus Lemmas~\ref{lema 0'} and~\ref{lema c}.
   
   \item $[x]_{\pi^{\bf v_3}}$ and $[z]_{\pi^{\bf v_3'}}$ can not be glued together by a sequence that begins or ends with a \emph{$U_2$}-kind gluing because we use new subtrees for that kind of gluings.
   
    \item By the same arguments given in the proofs of Facts~\ref{factnocambialaparticion1} and~\ref{factnocambialaparticion}, we can't have a sequence containing any of the following: 
\begin{center}
  \begin{tabular}{l@{\hskip .5in}l@{\hskip .5in}l}
   -- \emph{${\rm root}_{=,\neq}$}-\emph{$Z$}, 
   &
   -- \emph{$Z$}-\emph{${\rm root}_{=,\neq}$}, 
   &
   -- \emph{${\rm root}_{=,\lnot \neq}$}-\emph{$Z$},
   \\
   -- \emph{$Z$}-\emph{${\rm root}_{=,\lnot \neq}$}, 
   &
   -- \emph{${\rm root}_{=,\neq}$}-\emph{$U_2$},
   &
   -- \emph{$U_2$}-\emph{${\rm root}_{=,\neq}$},
   \\
   -- \emph{${\rm root}_{=,\lnot \neq}$}-\emph{$U_2$}, 
   &
   -- \emph{$U_2$}-\emph{${\rm root}_{=,\lnot \neq}$}.
   &
  \end{tabular}
  \end{center}
  
  \item By the same arguments given in the proof of Fact~\ref{factnocambialaparticion}, we can reduce sequences with two consecutive \emph{$Z$}-kind gluings to sequences that not have two consecutive \emph{$Z$}-kind gluings.

\item By the same arguments given in the proof of Fact~\ref{factnocambialaparticion}, we can't have sequences containing \emph{$U_2$}- \emph{$U_2$}.

\item One can prove that $[x]_{\pi^{\bf v_3}}$ and $[z]_{\pi^{\bf v_3'}}$ are not glued together by a sequence that alternates \emph{$Z$}-kind gluings and \emph{$U_2$}-kind gluings (starting and ending with \emph{$Z$}) by induction with arguments similar to the ones used in Lemma~\ref{no cambia la particion}.

  \end{itemize}

\end{sketch}
 
\begin{sketch}\label{detalles de 4b}
  Thinking in terms of sequences as in the proofs of Facts~\ref{factnocambialaparticion1} and~\ref{factnocambialaparticion}, one only has to observe that:

  \begin{itemize}
 
    \item $[x]_{\pi^{\bf v_3}}\neq [y]_{\pi^{\bf v_3}}$ for all $y\in T^{\bf v_3}$ in a class that was glued to the class of the {\rm root} via a \emph{${\rm root}_{=,\lnot \neq}$}-kind gluing (Use Lemma~\ref{lema 1}).
   
    \item $[z]_{\pi_1^{\bf u}}\neq [y]_{\pi_1^{\bf u}}$ for all $y\in T_1^{\bf u}$ in a class that was glued to the class of the {\rm root} via a \emph{${\rm root}_{=,\lnot \neq}$}-kind gluing (Use Lemma~\ref{lema 1}). 
   
   \item \emph{${\rm root}_{=,\neq}$}-kind gluings are made in different subtrees.
   
   \item $[x]_{\pi^{\bf v_3}}$ and $[z]_{\pi_1^{\bf u}}$ can not be glued together by a sequence that begins with a \emph{$U_2$}-kind gluing because we use new subtrees for that kind of gluings.
   
 \item $[x]_{\pi^{\bf v_3}}$ and $[z]_{\pi_1^{\bf u}}$ can not be glued together by a sequence that begins with a \emph{$Z$}-kind gluing because of the consistency of $\varphi$ plus  \neqaxeight and Lemma~\ref{paraverif}.
\end{itemize}  
\end{sketch}

\begin{sketch}\label{detalles de 5}
  Thinking in terms of sequences as in the proofs of Facts~\ref{factnocambialaparticion1} and~\ref{factnocambialaparticion}, one only has to observe that:

  \begin{itemize}
 
   \item  $[x]_{\pi^{\psi}}\neq [y]_{\pi^{\psi}}$ for all $y\in T^{\psi}$ in a class that was glued to the class of the {\rm root} via a \emph{${\rm root}_{=,\lnot \neq}$}-kind gluing (Use Lemma~\ref{lema 1}).
   
   \item  $[x]_{\pi^{\psi}}\neq [y]_{\pi^{\psi}}$ for all $y\in T^{\psi}$ in a class that was glued to the class of the {\rm root} via a \emph{${\rm root}_{=,\neq}$}-kind gluing (Rule 1).
   
      \item By the same arguments given in the proofs of Facts~\ref{factnocambialaparticion1} and~\ref{factnocambialaparticion}, we can't have a sequence containing any of the following: 
\begin{center}
  \begin{tabular}{l@{\hskip .5in}l@{\hskip .5in}l}
   -- \emph{${\rm root}_{=,\neq}$}-\emph{$Z$}, 
   &
   -- \emph{$Z$}-\emph{${\rm root}_{=,\neq}$}, 
   &
   -- \emph{${\rm root}_{=,\lnot \neq}$}-\emph{$Z$},
   \\
   -- \emph{$Z$}-\emph{${\rm root}_{=,\lnot \neq}$}, 
   &
   -- \emph{${\rm root}_{=,\neq}$}-\emph{$U_2$},
   &
   -- \emph{$U_2$}-\emph{${\rm root}_{=,\neq}$},
   \\
   -- \emph{${\rm root}_{=,\lnot \neq}$}-\emph{$U_2$}, 
   &
   -- \emph{$U_2$}-\emph{${\rm root}_{=,\lnot \neq}$}.
   &
  \end{tabular}
  \end{center}
\end{itemize}  
\end{sketch}

\begin{sketch}\label{detalles de 7}
  Thinking in terms of sequences as in the proofs of Facts~\ref{factnocambialaparticion1} and~\ref{factnocambialaparticion}, one only has to observe that:

  \begin{itemize}
 \item In case $\psi=\rho$, by consistency of $\varphi$ plus  \eqax{7} and {\bf Der21} of Fact \ref{fact boolean}, $\lnot \tup{\alpha = \beta}$ is a conjunct of $\psi$.
 
 \item $[x]_{\pi^{\psi}}$ and $[y]_{\pi^{\rho}}$ can not be glued together by a sequence of all \emph{$Z$}-kind gluings because of the consistency of $\varphi$ plus Lemmas~\ref{lema 0'} and~\ref{lema c}.
 
   \item By Lemma~\ref{lema 1} plus construction of $\Tt^{\varphi}$, $[y]_{\pi^{\rho}}\neq [z]_{\pi^{\rho}}$ for all $z\in T^{\rho}$ in a class that was glued to the {\rm root} via a \emph{${\rm root}_{=,\lnot \neq}$}-kind gluing.
   
   \item  By Rule 1, $[y]_{\pi^{\rho}}\neq [z]_{\pi^{\rho}}$ for all $z\in T^{\rho}$ in a class that was glued to the {\rm root} via a \emph{${\rm root}_{=,\neq}$}-kind gluing.
   
     \item By the same arguments given in the proofs of Facts~\ref{factnocambialaparticion1} and~\ref{factnocambialaparticion}, we can't have a sequence containing any of the following: 
\begin{center}
  \begin{tabular}{l@{\hskip .5in}l@{\hskip .5in}l}
   -- \emph{${\rm root}_{=,\neq}$}-\emph{$Z$}, 
   &
   -- \emph{$Z$}-\emph{${\rm root}_{=,\neq}$}, 
   &
   -- \emph{${\rm root}_{=,\lnot \neq}$}-\emph{$Z$},
   \\
   -- \emph{$Z$}-\emph{${\rm root}_{=,\lnot \neq}$}, 
   &
   -- \emph{${\rm root}_{=,\neq}$}-\emph{$U_2$},
   &
   -- \emph{$U_2$}-\emph{${\rm root}_{=,\neq}$},
   \\
   -- \emph{${\rm root}_{=,\lnot \neq}$}-\emph{$U_2$}, 
   &
   -- \emph{$U_2$}-\emph{${\rm root}_{=,\lnot \neq}$}.
   &
  \end{tabular}
  \end{center}
  
  \item By the same arguments given in the proof of Fact~\ref{factnocambialaparticion}, we can reduce sequences with two consecutive \emph{$Z$}-kind gluings to sequences that not have two consecutive \emph{$Z$}-kind gluings.

\item By the same arguments given in the proof of Fact~\ref{factnocambialaparticion}, we can't have sequences containing \emph{$U_2$}-\emph{$U_2$}.

\item One can prove by induction that $[x]_{\pi^{\psi}}$ and $[y]_{\pi^{\rho}}$ are not glued together by a sequence that alternates \emph{$Z$}-kind gluings and \emph{$U_2$}-kind gluings (neither starting with \emph{$Z$} or with \emph{$U_2$}) with arguments similar to the ones used in Lemma~\ref{no cambia la particion}.
\end{itemize}

(Notation: For $\psi, \rho \in N_n$, we use the notation $\Tt^{\psi}=(T^{\psi},\pi^{\psi}), \Tt^{\rho}=(T^{\rho},\pi^{\rho})$ with roots $r^{\psi}, r^{\rho}$ respectively to denote {\em any} tree in which $\psi, \rho$, respectively, are satisfiable.)

\end{sketch}